\pdfoutput=1
\documentclass[a4paper]{article}
\usepackage[utf8]{inputenc}
\usepackage[english]{babel}
\usepackage{amsmath}
\usepackage{amsfonts}
\usepackage{amssymb}
\usepackage{amsthm}
\usepackage[basic]{complexity}
\usepackage{tikz}
\usepackage{graphicx}
\usepackage{booktabs}
\usepackage[round,longnamesfirst]{natbib}
\usepackage{tabularx,environ}
\usepackage[hidelinks]{hyperref}
\usepackage{colortbl}
\usepackage[framemethod=tikz]{mdframed}
\mdfdefinestyle{mystyle}{
  linecolor=gray!35,leftmargin=0,
  rightmargin=0,
  backgroundcolor=gray!25,
  innertopmargin=0pt,
  skipabove=\topsep,
  skipbelow=\topsep
}

\newcommand{\Max}{\mbox{max}}
\newcommand{\Min}{\mbox{min}}
\makeatletter
\newcommand{\problemtitle}[1]{\gdef\@problemtitle{#1}}
\newcommand{\probleminput}[1]{\gdef\@probleminput{#1}}
\newcommand{\problemquestion}[1]{\gdef\@problemquestion{#1}}
\NewEnviron{problem}{
  \problemtitle{}\probleminput{}\problemquestion{}
  \BODY
  \par\addvspace{.4\baselineskip}
  \noindent
  \begin{tabular}{@{\hspace{0pt}}lp{0.752\columnwidth}}
    \multicolumn{2}{@{\hspace{0pt}}l}{\textsc{\@problemtitle}} \\
    \textbf{Input:} & \@probleminput \\
    \textbf{Question:} & \@problemquestion
  \end{tabular}
  \par\addvspace{.4\baselineskip}
}

\makeatletter
\NewEnviron{OLproblem}{
  \par\addvspace{.4\baselineskip}
  \noindent
  \begin{tabular}{@{\hspace{0pt}}lp{0.752\columnwidth}}
    \textbf{Input:} & A set $X$, a linear order $\leq$ on $X$, a family $\cX \subseteq \pow$ and a set of axioms $\mathcal{A}$. \\
    \textbf{Goal:} & Find an order $\preceq$ on $\cX$ that satisfies all axioms in $\mathcal{A}$ with respect to $\leq$.
  \end{tabular}
  \par\addvspace{.4\baselineskip}
}

\newcommand{\type}{\mbox{\textit{type}}}

\newcommand{\cX}{\mathcal{X}}
\newcommand{\cY}{\mathcal{Y}}

\newcommand{\nop}[1]{}

\newcommand{\Sat}{\textsc{Sat}}

\newcommand{\powerset}[1]{\mathcal{P}(#1)}
\newcommand{\pow}{\mathcal{P}(X) \setminus \{\emptyset\}}
\newcommand{\LetX}{Let $X$ be a set, $\leq$ a linear order on $X$ and 
$\cX \subseteq \pow$}

\newtheorem{Thm}{Theorem}
\newtheorem{Prop}[Thm]{Proposition}
\newtheorem{Cly}[Thm]{Corollary}
\newtheorem{Lem}[Thm]{Lemma}

 \newtheoremstyle{Empty}
      {\topsep}{\topsep}              
      {\itshape}                      
      {0pt}                              
      {\bfseries}                     
      {. }                             
      {  }                             
      {\thmnote{\bfseries #3}}
\theoremstyle{Empty}

\theoremstyle{definition}
\newtheorem{Def}[Thm]{Definition}

\newtheorem{Exp}[Thm]{Example}

\newtheorem{Obs}[Thm]{Observation}

\newtheorem{Axm}[Thm]{Axiom}

\begin{document}

\title{Ranking Sets of Objects: The Complexity of \\Avoiding Impossibility Results}

\author{Jan Maly \\ j.f.maly@uva.nl \\
       Institute for Logic, Language and Computation,\\ University of Amsterdam, Amsterdam, The Netherlands}

\maketitle

\begin{abstract}
The problem of lifting a preference order on a set of objects
to a preference order on a family of \emph{subsets} of this set
is a fundamental problem with a wide variety of applications in AI.
The process
is often guided by axioms postulating 
properties the lifted order
should have. 
Well-known impossibility results 
by Kannai and Peleg and by Barber{\`a} and Pattanaik 
tell us that some desirable axioms 
-- namely dominance and (strict) independence --
are not jointly satisfiable for any linear order on the objects
if \emph{all} non-empty sets of objects
are to be ordered.
On the other hand, if \emph{not all} non-empty sets of objects are to be ordered,
the axioms are jointly satisfiable
for all linear orders on the objects for some families of sets.
Such families are very important for applications as
they allow for the use of lifted orders, for example, in combinatorial voting.
In this paper, we determine the computational complexity of 
recognizing such families. 
We show that it is 
$\Pi_2^p$-complete to decide for a given family of subsets whether 
dominance and independence or dominance and strict independence
are jointly satisfiable \emph{for all} linear orders on the objects
if the lifted order needs to be total.
Furthermore, we show that the problem remains \coNP-complete 
if the lifted order can be incomplete.
Additionally, we show that the complexity of these problems 
can increase exponentially if the family of sets is not given explicitly 
but via a succinct domain restriction.
Finally, we show that it is \NP-complete to decide
for a family of subsets whether 
dominance and independence or dominance and strict independence
are jointly satisfiable for \emph{at least one} linear order on the objects.
\end{abstract}

\section{Introduction}
Modeling preferences over alternatives is a major challenge  
in many areas of AI, for example in knowledge representation 
and, especially, in computational social choice. 
If the number of alternatives is small enough,
preferences are most often modeled 
as a total order.
However, in many applications the alternatives are `combinatorial',
for example bundles of objects in packing or allocation problems \citep{BouveretCM16},
or committees in voting \citep{lang2016,faliszewski2017}.
In such situations,
the number of alternatives 
grows exponentially with the number of objects
which makes it unfeasible for agents to specify a full 
preference relation over all alternatives.

Different approaches to solving this problem have been discussed in the literature.
One approach -- which we will consider in this paper --
is inferring an order on sets of alternatives from
an order on the alternatives. 
This is also called lifting an order from objects to sets of objects.
This approach occurs frequently in different applications in
computer science, for example in the fair allocation of indivisible goods 
\citep{BouveretCM16}, in the study of strategic behaviors in voting 
with a tie-breaking mechanism \citep{Fishburn1972,Barbera1977,BrandtB18},
in decision making when there is uncertainty about the consequences
of an action \citep{LarbiKM10} or in structured argumentation, if we want to
compare arguments according to the strengthen of the (defeasible) assumptions
they are based on \citep{BeirlaenHPS18,AAAI21}.

One particularly interesting setting in which sets of objects need to be
ranked is combinatorial voting \citep{lang2016}, where 
many approaches for generating a ranking on sets of objects
have been studied, none of which can be considered completely convincing
if the preferences are not seperable (see Related Work).
Here, the order lifting approach could offer a new and promising approach 
to dealing with the intractable number of alternatives. Indeed,
a single winner (ordinal) voting rule combined with
an order lifting procedure would immediately yield a combinatorial voting 
method, either by applying the lifting procedure to the input of the voting rule
or to its output (see Figure~\ref{fig:1}).

\begin{figure}
\begin{minipage}{0.49\textwidth}
\begin{center}
\begin{tikzpicture}[font={\footnotesize},node distance=3ex and 4ex,
]
\begin{scope}[every node/.style={draw}]
\node[rounded corners,  minimum width = 27ex] at (0,0) (1) {Voting Rule};
\node[rounded corners,  minimum width = 27ex] at (0,-1.5) (2) {Lifting Procedure};
\node[rounded corners, ] at (0,-3) (4) {Voter 2};
\node[rounded corners] at (-1.5,-3) (3) {Voter 1};
\node[rounded corners] at (1.5,-3) (5) {Voter 3};
\end{scope}

\path[<-]
(1) edge (2)
(-1.5,-1.76) edge (3)
(1.5,-1.76) edge (5)
(2) edge (4)
(-1.5,-0.26) edge (-1.5,-1.26)
(1.5,-0.26) edge (1.5,-1.26)
(1.5,-1.76) edge (5)
;
\node[fill = white,draw = none] at (0,-2.25) (6) {Preferences on objects};
\node[fill = white, draw = none] at (0,-0.75) (6) {Preferences on sets};
\end{tikzpicture}
\end{center}
\end{minipage}
\begin{minipage}{0.49\textwidth}
\begin{center}
\begin{tikzpicture}[font={\footnotesize},node distance=3ex and 4ex,
]
\begin{scope}[every node/.style={draw}]
\node[rounded corners,  minimum width = 27ex] at (0,0) (1) {Lifting Procedure};
\node[rounded corners,  minimum width = 27ex] at (0,-1.5) (2) {Voting Rule};
\node[rounded corners, ] at (0,-3) (4) {Voter 2};
\node[rounded corners] at (-1.5,-3) (3) {Voter 1};
\node[rounded corners] at (1.5,-3) (5) {Voter 3};
\end{scope}

\path[<-]
(1) edge (2)
(-1.5,-1.76) edge (3)
(1.5,-1.76) edge (5)
(2) edge (4)
(-1.5,-0.26) edge (-1.5,-1.26)
(1.5,-0.26) edge (1.5,-1.26)
(1.5,-1.76) edge (5)
;
\node[fill = white,draw = none] at (0,-2.25) (6) {Preferences on objects};
\node[fill = white, draw = none] at (0,-0.75) (6) {Preferences on objects};
\end{tikzpicture}
\end{center}
\end{minipage}
\caption{The use of lifted orders in voting}
\label{fig:1}
\end{figure}
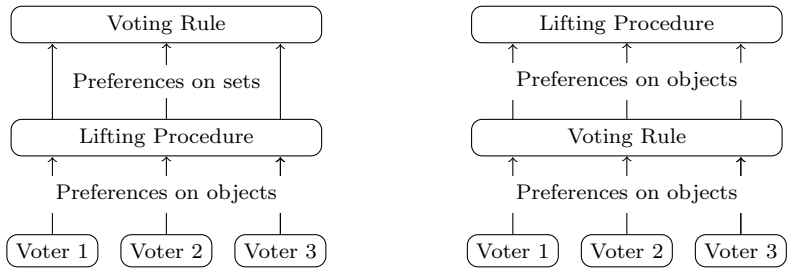

However it is not clear how such a lifted order on sets can be defined.
\citet{barbera2004} give an excellent survey on the 
progress that has been made to solve this question.
Essentially, there are two ways to study the order lifting approach.
First, one can analyze specific methods to infer an order 
on sets from an order on the objects \citep[see e.g.]{MorettiT12}.
Second, one can analyze which properties an optimal order on the sets 
should have in a given setting.
These desirable properties can then be formulated as axioms
and one can try to identify orders that satisfy these axioms. 
In some settings, it is possible to find orders that satisfy all properties deemed desirable
in the setting \citep[see e.g.]{pattanaik1990}.
Unfortunately, in many other settings 
no order can satisfy all desirable properties at once.
One of the most important of these so-called impossibility results
states that two axioms called \emph{Dominance} and \emph{Independence} are in general incompatible
\citep{kannai1984}.
Dominance states, intuitively,
that removing the least preferred element from a set improves the set
and removing the most preferred element worsens the set.
Independence states, roughly, that if a set is preferred to another set 
and the same element is added to both sets, this preference cannot be reversed.

\citet{kannai1984} showed that
dominance and independence are incompatible
if one wants to order all subsets of a set with at least six elements.
In the same year \citet{barbera1984} showed that 
dominance and a strengthening of independence called strict independence
are incompatible already if one wants to order all subsets of a three element set.
These results assume that all subsets of a given set need to be ordered,
but in many applications only some of these subsets 
are possible alternatives. 
For example, it is very common in multiwinner or combinatorial voting to have some
form of domain restrictions \citep{lang2016,Kilgour16}.
Now, it is possible to construct arbitrary large families of sets
-- for example families of disjoint sets -- that
can be ordered with an order satisfying dominance and (strict) 
independence.\footnote{We write `(strict) independence' as a shorthand for
`independence or strict independence'.}
Motivated by this observation, \citet{Maly2017} have shown that it is
\NP-complete to decide whether a given order on elements
can be lifted to an order satisfying dominance and (strict) independence.

However, for applications in voting or other social choice problems,
it is necessary to fix a voting method 
before the ballots are collected.
Therefore, it is more important to know 
for a given family of sets
if dominance and (strict) independence are compatible
for any preference order the agents may report.\footnote{Observe that the hardness of this problem
does not follow from the hardness of the aforementioned problem treated by
\citet{Maly2017}.} 
Following \citet{JAIR} 
we call families of sets for which any possible order on the elements
can be lifted to an order satisfying dominance and (strict) independence
\emph{strongly} orderable with respect to
dominance and (strict) independence.
\citet{JAIR} studied this concept
for a specific class of families of sets,
namely for families of sets that can be represented
as the family of all sets of vertices that induce a connected subgraphs in a given graph.
One of their main results is a classification result that
implies that strong orderability with respect
to dominance and strict independence
can be decided in polynomial time
for families in this restricted class.

In this paper, we show that this result cannot be generalized
to arbitrary families of sets.
We show that it is in general
$\Pi_2^p$-complete to decide whether a family of sets is
strongly orderable with respect to dominance and strict independence.
This result also holds if we replace strict independence by independence
or additionally require a natural axiom called the extension rule.
These results assume that we require the order on the family of sets to be total.
However, some authors argue that it is more sensible 
to only require incomplete preferences in combinatorial domains \citep{boutilier2016}.
Voting rules that facilitate the aggregation of partial orders or 
even weaker preference models exist \citep{xia2011,endriss2019}.
Therefore, we investigate how much the complexity of the studied problems can be reduced by
dropping the requirement that the lifted order needs to be total.
In particular, we show that it is \coNP-complete to decide
whether a family of sets is strongly orderable
if we require the order on the family to be a partial order.

These results assume that the family of sets is given explicitly.
However, in many applications, the family of sets is only given implicitly
due to its size which necessitates the use of lifted orders in the first place.
For example, the domain in combinatorial voting is often given as
a condition that has to be satisfied by the admissible sets.
Such conditions can for example be formulated as propositional formulas \citep{lang2016}.
These formulas are normally exponentially smaller than the 
actual family of sets, which can increase the complexity 
of deciding if the family is 
strongly orderable.
On the other hand, families of sets 
must have some internal structure to be succinctly represented.
This internal structure may decrease the complexity of the problem,
as is the case for the domain restrictions considered by \citet{JAIR}.
We show -- for boolean circuits, a succinct representation that is well studied 
in the literature -- that succinct representation can lead to a massive blow up in complexity.
It turns out that it can be \NEXP-hard to
decide whether a succinctly represented family is
strongly orderable with respect to dominance and strict independence for 
total orders and \coNEXP-complete for partial orders.
The first result also holds if only independence is required instead of strict independence.

Finally, in some applications, for example when determining the most preferred 
models in a knowledge representation formalism \citep{BrewkaNT03}, we may want to use lifted rankings when possible and
only require more information from the users if they submitted a ranking that cannot be lifted.
In this case, it is valuable to know if there is no chance 
that a submitted ranking can be lifted such that we can ask for more information immediately.
Therefore, we also study the complexity of deciding whether a family of sets 
is \emph{weakly} orderable, that is if there is at least one order on the elements that
can be lifted to an order satisfying dominance and (strict) independence.
In particular, we show that it is \NP-complete to decide whether a family
of sets is weakly orderable with respect to dominance and strict independence.
This result also holds if the extension rule is additionally required and if
only a partial order is required.

Together, our results give a nearly complete picture of the complexity
of determining whether dominance and (strict) independence are compatible
on a given family of sets.
To conclude the paper, we investigate the robustness of the obtained 
hardness results to modifications of the dominance axiom.
For example, one could define a variation of dominance where several elements can be
added at once.
We show that deciding whether strict independence is jointly satisfiable
with any dominance-like axiom is NP-complete.

\subsection{Related Work}

The order lifting problem has been extensively studied by
authors from a wide range of scientific disciplines,
including economists,
mathematicians, 
and computer scientists.
It is beyond the scope of this paper to give a complete overview over
this literature. Instead, we give only a very short overview over the history of 
the order lifting problem.

The idea of ranking sets of objects based on a ranking of the objects is very old.
For example, since antiquity humans order words lexicographically 
based on an order on the letters of the alphabet. \citet{Daly67}
gives an interesting account of the development of this 
technique by ancient scholars. Other orders, like ordering 
by maximal elements are most probably as old. 
A greater academic interest in the order lifting problem was sparked by the 
famous result by \citet{Gibbard73}
and \citet{Satterthwaite75} that no resolute voting rule can be strategyproof.
Studying the strategyproofness of irresolute
voting rules necessitated the definition of rankings of sets of candidates 
based on rankings of candidates. Noteworthy attempts were made by 
\citet{Fishburn1972}, \citet{gardenfors1976} and \citet{Kelly77}.
Additionally, lottery based rankings were often considered \citep{DugganS00}.
These are still the most widely used rankings for studying the strategyproofness 
of irresolute voting rules \citep{Handbook-Tournament,barbera2011strategyproof}.

The axiomatic approach to the order lifting problem considered in this paper
is significantly younger. First works were published from the 
fifties onward, for example by \citet{kraft1959} and \citet{Kim80}. 
The seminal result that sparked a lot of interest in this 
area was Kannai and Peleg's famous impossibility result \citep{kannai1984}.
We will discuss this result, that is one of the key motivations of this paper,
in Section~\ref{sec:impos} in detail. In the following years, a
significant number of papers more or less directly inspired 
by Kannai and Peleg's result were published. 
Important examples include papers by \citet{barbera1984}, \citet{Holzman84}, \citet{BarberaBP1984},
\citet{fishburn}, \citet{bandyopadhyay}, \citet{bossert1995}, \citet{Kranich96},
\citet{bossert2000} and \citet{Dutta05}.
For a complete overview over this line of research, we refer the reader again
to the great survey by \citet{barbera2004}.

More recent developments include characterization results for 
decision making under complete uncertainty \citep{LarbiKM10,bossert2012}
as well as novel approaches for ranking sets of interacting elements
\citep{MorettiT12,lucchetti2015ranking}.
Furthermore, \citet{GeistE11} managed to show several new impossibility 
results with computer generated proofs using \Sat-solving methods.
Finally, the social ranking problem of ranking elements based on a 
ranking of subsets, which could be seen as the dual 
of the order lifting problem, has received a lot of attention
lately. The problem was first introduced by \citet{MorettiO17}
and considerable progress has been made since then
\citep{Adrian18,Hossein19,Bernardi19,AlloucheEMO20}.
However, to our knowledge,  all of the previous 
works on the axiomatic approach to the order lifting problem
required a ranking of all possible subsets,
with the notable exception of \citet{bossert1995},
who studied rankings of subsets with fixed cardinality,
and \citet{Maly2017} as well as \citet{JAIR} on whose work we 
explicitly build.
In this sense, the work presented here is clearly different
from the existing literature on the order lifting problem. 

More generally, the work presented in this paper
can also be seen as a contribution to an area of 
research in AI concerned with models of preferences for \emph{combinatorial 
domains}. 
One successful approach to dealing with combinatorial domains that has
been extensively studied is the use of concise \emph{implicit} preference models 
\citep{Domshlak2011,Kaci2011}. 
Our research can be seen as belonging to this line of research. 
Whenever there is a (set of) lifted orders satisfying a collection of axioms,
then the original order on elements of a set
can serve as a concise representation of the lifted order on the much larger
domain of subsets of that set. In general, for any succinct representation,
it is crucial to know the complexity of reasoning about lifted orders based on the ``ground'' information
about preferences on elements. In this sense, our research can be seen as a first step 
in studying of the use of lifted orders as a tool for implicit representations.
In contrast to the lifted order approach, most implicit preference models in the literature
either build on logical languages \citep{DuboisP91,BrewkaBB04,BrewkaNT03} 
or employ intuitive graphical representations such as lexicographic trees 
\citep{Booth2010,brauning2012,LiuT2015}, CP-nets \citep{bbdh03} or CI-nets \citep{BouveretEL09}.
To this date, it is unclear whether the rankings obtained by such formalisms 
satisfy desirable properties that are formalized in the axiomatic approach.

Finally, many other approaches to combinatorial voting have been considered 
in the literature. The easiest solution is often to vote on each candidate separately.
However, this approach only works well if the voters have separable preferences,
i.e., if the preference on having a candidate in the committee
is independent on who else is in the committee.
Another option is eliciting the top ranked committee.
Then one can, for example, infer a preference order on the committees via
a distance measure like the Hamming distance.
This approach minimizes the communication cost 
but only takes very little of the agent's full preferences into account.
Alternatively, one can ask the agents to specify their preferences 
using a CP-net or similar representations.
This can be very effective but requires the agents to learn
a non-trivial preference representation.
In many cases this is an unacceptable requirement. 
Finally, there are some voting rules that select 
a winning committee directly from a preferences over candidates.
This approach has been mainly explored for committees of fixed size.
For more information on all of these approaches, we refer the 
reader to the excellent survey by \citet{lang2016}.

\subsection{Organization of the Paper}

In the next section, we will introduce the necessary technical background for
our results. First, we discuss some general preliminaries
(Subsection~\ref{genprelim}). Then, we introduce the
four axioms that are the main focus of the paper (Subsection~\ref{sec:4axm})
and discuss two famous impossibility results regarding these axioms 
(Subsection~\ref{sec:impos}). We define 
the three possible degrees to which a family can be considered orderable
(Subsection~\ref{sec:orderability}).
To conclude the section, we introduce two variants of our main axioms
and their relation to the original axioms
(Subsection~\ref{sec:additional}).
These will be useful to prove our complexity results.
In Section~\ref{chap:Complexity}, we discuss our complexity results, focusing first
on strong orderability with respect to total orders (Subsection~\ref{sec:leqstrongord}) and
partial orders (Subsection~\ref{sec:partial}).
Then, we discuss the effect of succinct representation
on our complexity results in Subsection~\ref{sec:succ}. Subsequently, we consider 
weak orderability (Subsection~\ref{sec:weak}) and close the section by studying
the effect of strengthening dominance on our complexity results (Subsection~\ref{sec:strengdom}).
Finally, in Section~\ref{sec:diss} we summarize and discuss our results.

\section{Background}

In the following, we will introduce the order lifting problem, the axioms that we consider
as well as the famous impossibility results regarding these axioms.

\subsection{General Preliminaries}\label{genprelim}

All sets we consider in the paper are finite. A binary relation is called 
a \emph{preorder} if it is reflexive and transitive and is called a
 \emph{weak order} if it is reflexive, transitive and total.\footnote{Weak orders 
are also called total preorders or just orders.} A preorder is called a \emph{partial
order} if it is also antisymmetric and a weak order is called
\emph{linear} if it is also antisymmetric. 
If $\preceq$ is a weak order on a set $X$, the corresponding \emph{strict} order 
$\prec$ on $X$ is defined by $x \prec y$ if $x \preceq y$ and $y \not\preceq 
x$, where $x,y$ are arbitrary elements of $X$; 
the corresponding 
\emph{equivalence} or \emph{indifference} relation $\sim$ is defined by 
$x \sim y$ if $x \preceq y$ and $y \preceq x$. If $\preceq$ is 
linear 
then $x \sim y$ 
holds
only if $x=y$.
We call the linear order $1 < 2 < 3 < \dots$ the \emph{natural
linear order} on the natural numbers. If objects are identified with the natural
numbers then we also call this order the natural order on these objects.

For a linear order $\preceq$ on a set $A$, we write $\max_\preceq (A)$ 
for the maximal element of $A$ with respect to $\preceq$. 
Similarly, we write $\min_\preceq (A)$ for the minimal 
element of $A$ 
with respect to $\preceq$. 
If no ambiguity arises, we drop 
the reference to the relation from the notation.

Given a set $X$ and a linear order $\leq$ on $X$, the \emph{order lifting} 
problem consists of deriving from $\leq$ an order $\preceq$ on a family 
$\mathcal{X}\subseteq \powerset{X} \setminus \{\emptyset\}$ of non-empty 
\emph{subsets} of $X$, guided by axioms formalizing some natural desiderata
for such lifted orders. Generally, these axioms take $\leq$ into account.
Therefore, an order $\preceq$ on $\cX$ in general only satisfies an axiom 
with respect to a specific linear order $\leq$ on $X$.

\begin{center}
\begin{mdframed}[style=mystyle,frametitle=The Order Lifting Problem for a set of axioms $\mathcal{A}$]
\begin{minipage}{0.99\textwidth}
\begin{OLproblem}
\end{OLproblem}
\end{minipage}
\end{mdframed}
\end{center}

In most applications, 
the set of objects either represents a 
collection of objects that the agent receives as a whole (e.g.\ fair division)
or a set of possible outcomes from which one will be selected (e.g.\ in 
voting with tie-breaking).
The first case is often called the conjunctive interpretation
and the second case the disjunctive interpretation.
Moreover, observe that we did not specify what kind of order $\preceq$ should be. 
In the following, we will consider the problem of lifting to a variety of different 
types of orders. The most common case will be lifting to a weak order,
which is also the problem considered most often in the literature. 
However, we will also consider lifting to preorders, partial orders 
and linear orders. 
Furthermore, we note that we assume that $\cX$ only contains \emph{nonempty}
sets because one of our axioms, namely dominance, immediately
leads to a contradiction if the empty set is contained in $\cX$.

For the uniformity of notation, we will stick to the following conventions:
In any instance of the order lifting problem, we will use uppercase letters 
to denote the set of objects or ground set, for example, $X$ or $Y$, and lowercase letters or natural numbers to denote its elements.
We use calligraphic letters for the family of subsets, for example, $\cX$ or $\cY$
and uppercase letters at the beginning of the alphabet for its elements,
i.e., for subsets of the ground set.
Similarly, we use $\leq$ for the linear order on the ground set, possible with an index for uniqueness,
and the calligraphic $\preceq$ for any order on the family of subsets, also possibly with an index.

\subsection{The Main Axioms}\label{sec:4axm}

In this section we introduce the four axioms that will be the main focus of
our investigation, namely dominance, independence, strict independence
and the extension rule.
The definitions of these axioms are not entirely consistent in the literature.
We will essentially follow \citet{barbera2004} with our definitions.
However, we need to a add a condition that states that the
axiom is only applicable if a set is in the family of sets that we want to lift to.
As \citet{barbera2004} only consider the case $\cX = \pow$,
such a condition is not needed in their version of the axioms.\footnote{Observe
that another way of adapting the axioms to the setting $\cX \neq \pow$
would be, to demand that the lifted order $\preceq$ is a binary relation 
on $\pow$ that satisfies all axioms and the restriction of $\preceq$ to 
$\cX$ is a weak order. A study of this more restrictive approach is left 
for future work.} 

Throughout this section, we will demonstrate the effect of the axioms 
on the following toy example:
Let $S_{oy} = \{1,2,3,4\}$ and let $\leq$ be the natural linear order on $S_{oy}$.
Furthermore, let 
\[\mathcal{T}_{oy} = \{\{2\}, \{4\}, \{2,4\},\{3,4\}, \{1,2,4\}, \{1,4\}\}.
\]
Now, let us introduce our main axioms. We begin with the so-called dominance axiom.

\begin{center}
\begin{mdframed}[style=mystyle,frametitle=Dominance]
\begin{minipage}{0.9\textwidth}
For all $A \in \mathcal{X}$
and all $x\in X$, such that
$A \cup \{x\} \in \mathcal{X}$:
\[
y<x \text{ for all } y \in A \text{ implies } A \prec A \cup \{x\};
\]
\[
x < y \text{ for all } y \in A \text{ implies } A \cup \{x\} \prec A.
\]
\end{minipage}
\end{mdframed}
\end{center}

Any relation $\preceq$ on $\mathcal{T}_{oy}$ that satisfies
dominance with respect to $\leq$ must set $\{2\} \prec \{2,4\}$,
$\{2,4\} \prec \{4\}$, $\{3,4\} \prec \{4\}$, $\{1,2,4\} \prec \{2,4\}$ and $\{1,4\} \prec \{4\}$.

Dominance is often also called G\"ardenfors' principle after Peter G\"ardenfors
who introduced a version of the axiom 
\citep{gardenfors1976}. It states that 
adding an element to a set that is better than all elements already in the set
increases the quality of the set and,
similarly, adding a element worse than all elements in the set decreases the quality of the set.

Dominance is often desirable if the order $\preceq$ should reflect the average quality of sets.
If sets represent possible outcomes, they can often be ranked by 
expected utility, which equals the average quality of the elements
if elements are sampled with uniform probability.
Such a ranking naturally satisfies dominance.
Therefore, dominance is often desirable under the disjunctive interpretation of sets,
where the sets represent incompatible alternatives from which one is chosen randomly.
Consider, for example, an election where a voter 
knows that depending on his vote different sets of candidates 
will be tied for the first place. Furthermore, he knows that 
the final winner will be chosen from the tied candidates randomly.
In such situations, dominance is a natural desideratum~\citep{can2009}.

Observe that there are different formulations of dominance in the literature.
The two most important ones are the definition given here and a version
where several elements can be added at the same time (see Definition~\ref{SetDom}
for a formal statement of this version) that could be called set-dominance.
Dominance and set-dominance, are equivalent if
$\mathcal{X} = \mathcal{P}(X) \setminus \{\emptyset\}$
and hence are used interchangeably in the literature. However, if
$\mathcal{X} \neq \mathcal{P}(X) \setminus \{\emptyset\}$, then the two version 
constitute different axioms, as do many other possible formulations of the axiom.
Given that dominance and independence are already hard to jointly satisfy, we will focus on the weakest possible formulation
of the axioms for the main part of the paper.
However, other definitions also hold intuitive appeal and there is not necessarily
a single best formulation of dominance if $\mathcal{X} \neq \mathcal{P}(X) \setminus \{\emptyset\}$.
Therefore, we will studied the effect of using other, stronger formulations of dominance in Section~\ref{sec:strengdom}.
Fortunately, our results indicate that changing the formulation of dominance does not seem to influence
the complexity of the studied problems.

Several commonly used orders on families
of sets satisfy dominance.
For example, the following  maxmin-based\footnote{Formally,
we call an order $\preceq$ a maxmin-based order if there exists an
order $\preceq^*$ on $X \times X$ such that $A \preceq B$ holds if and only if $(\min(A) ,\max(A)) \preceq^* (\min(B), \max(B))$
holds \cite[p.13]{barbera2004}.}
order satisfies dominance on all families of sets.

\begin{Exp}\label{Exp:precmm}
\LetX.
Then, we can define a weak order $\preceq_{mm}$ on $\cX$ by $A \preceq_{mm} B$
for $A,B \in \cX$ if
\begin{itemize}
\item $\max(A) < \max(B)$ or
\item $\max(A) = \max(B)$ and $\min(A) \leq \min(B)$.
\end{itemize}
Observe that this is the maxmin-based order defined by the lexicographic order on $X \times X$.
It is straightforward to check that this is a weak order.
Furthermore, $\preceq_{mm}$ satisfies dominance:
Let $x \in X$ and $A, A \cup \{x\} \in \cX$.
Assume $\max (A) < x$. Then, $\max(A) < \max( A \cup \{x\}) = x$
and hence $A \prec_{mm} A \cup \{x\}$.
On the other hand, if $x< \min(A)$, then
$\max( A \cup \{x\}) = \max(A)$
and $\min( A \cup \{x\}) = x < \min(A)$
and hence $A \cup \{x\} \prec_{mm} A$.
\end{Exp}

Other well-known examples of orders that satisfy dominance are the 
lifted orders proposed by \citet{Fishburn1972} and 
\citet{gardenfors1976}. Both are frequently used 
in the context of strategyproofness 
in elections with tie-breaking \citep{Handbook-Tournament,barbera2011strategyproof}.

\begin{Exp}\label{FishGExt}
\LetX.
Then, the so-called Fishburn extension $\preceq_f$ is defined by $A \preceq_f B$ if
all of the following conditions hold:
\begin{itemize}
\item $x < y$ for all $x \in A \setminus B$ and $y \in A \cap B$,
\item $y < z$ for all $y \in A \cap B$ and $z \in B \setminus A$,
\item $x < z$ for all $x \in A \setminus B$ and $z \in B \setminus A$.
\end{itemize}
Observe that the first two conditions imply the third unless $A \cap B = \emptyset$.
On $\mathcal{T}_{oy}$ Fishburn's extension looks as follows:
\begin{multline*}
\{2\} \prec_f \{2,4\} \prec_f \{4\}; \{2\} \prec_f \{3,4\} \prec_f \{4\};\\
\{1,2,4\} \prec_f \{4\}; \{1,4\} \prec_f \{4\};
\{1,2,4\} \prec_f \{2,4\}.
\end{multline*}
We claim that $\preceq_f$ satisfies dominance.
Assume $A,A\cup \{x\} \in \cX$ and $\max(A) < x$. Then, $A \setminus (A \cup \{x\}) = \emptyset$
and $y < x$ for all $y \in A \cap (A \cup \{x\}) = A$. Hence $A \prec_f A \cup \{x\}$.
The case $x < \min (A)$ is analogous. 

The so-called G\"ardenfors' extension $\preceq_g$ is defined by $A \preceq_g B$ if
one of the following holds
\begin{itemize}
\item $A \subseteq B$ and $x < y$ for all $x \in A$ and $y \in B \setminus A$,
\item $B \subseteq A$ and $x < y$ for all $x \in A \setminus B$ and $y \in B$, 
\item Neither $A \subseteq B$ nor $B \subseteq A$ and $x < y$ for all $x \in A \setminus B$
and $y \in B \setminus A$.
\end{itemize}
It is known that G\"ardenfors' extension is a superset of Fishburn's extension.
This means, $A \prec_f B$ implies $A \prec_g B $ for every $A,B \in \cX$.
Therefore, it follows directly that G\"ardenfors' extension satisfies dominance.
On $\mathcal{T}_{oy}$ G\"ardenfors' extension adds to Fishburn's extension
the following preferences:
\[\{1,4\} \prec_g \{2,4\} \prec_g \{3,4\}; \{1,2,4\} \prec_g \{3,4\}\]
\end{Exp}

The second axiom that we consider is called independence.

\begin{center}
\begin{mdframed}[style=mystyle,frametitle=Independence]
\begin{minipage}{0.9\textwidth}
For all $A, B \in \mathcal{X}$ and all $x \in X \setminus (A \cup B)$,
such that $A \cup \{x\}, B \cup \{x\} \in \mathcal{X}$:
\[A \prec B \text{ implies } A \cup \{x\} \preceq B \cup \{x\}.\]
\end{minipage}
\end{mdframed}
\end{center}

A relation $\preceq$ on $\mathcal{T}_{oy}$ that satisfies
independence must set $\{1,2,4\} \preceq \{1,4\}$ if it contains $\{2,4\} \prec \{4\}$
and $\{1,4\} \preceq \{1,2,4\}$ if it contains $\{4\} \prec \{2,4\}$.
As dominance implies $\{2,4\} \prec \{4\}$, dominance and independence together imply 
$\{1,2,4\} \preceq \{1,4\}$.

\emph{Independence} is a natural monotonicity axiom that states that
if we add the same element $x$ to two sets $A$ and $B$
where $B$ is strictly preferred to $A$, then $B \cup \{x\}$
must be at least weakly preferred to $A \cup \{x\}$.
Observe that in contrast to dominance a `set-based' version 
of independence would be much stronger that the studied version even if
$\mathcal{X} = \mathcal{P}(X) \setminus \{\emptyset\}$.
Indeed the definition of independence explicitly includes
the possibility that adding one element may already be enough
to equalize $A$ and $B$, in which case independence can not be applied again.
Therefore, we will not considering `set-independence'.

Independence is often a very desirable property under the conjunctive interpretation,
for example if sets are bundles of objects that are compared according to their overall quality
according to some additive utility~\citep{kraft1959}.
Indeed, if we define an order based on the sums of utilities, 
that order satisfies independence.

In this sense there is some tension between the motivations for dominance
and independence.
Nevertheless, there are cases where both axioms are natural 
desiderata.
These cases are often characterized by the fact that all
elements may influence the quality of a set
but the extent of this influence is unknown or unknowable.
An example under the disjunctive interpretation
for such a situation is choice under complete uncertainty:

\begin{Exp}
Consider a situation where an agent can perform actions $a_1, \dots, a_k$
for which he knows the (set of) possible outcomes
but he is not able or not willing to determine the (approximate) probability of 
each outcome. Such a situation can be modeled as a family of outcomes $X = \{o_1, o_2, \dots, o_l\}$
and a function $O:\{a_1, \dots ,a_k\} \to \powerset{X}\setminus\{\emptyset\}$
that maps every action to the set of possible outcomes of that action. 
If we assume that the agent has preferences over the set of possible outcomes $X$
that can be modeled as a linear order, the problem of ranking the different actions
can be modeled as an order lifting problem. Under this interpretation
the extension rule (see below), dominance and independence are usually considered natural desiderata
\citep{bossert2000,barbera2004}.
\hfill$\Box$
\end{Exp}

In voting, a comparable situation appears if ties are broken by an 
unknown chairman.\footnote{Which axioms or order are most appropriate in this setting
also depends on the risk tolerance of the voters.
Very risk averse and undeceive voters might be modeled using Fishburn's extension
\citep{BrandtB18}. For agents with other risk profiles $\preceq_{pmm}$,
 a lexicographic order or something else might be more appropriate.}
Similarly, situations exist under the conjunctive interpretation in which
it is unclear how much each object contributes to the 
quality of the set,
for example, in voting when
sets represent an elected committee in which each member has an
(a priori) unknown influence.

In contrast to dominance, independence on its own does not require
any preferences. In other words, the empty preference relation
always satisfies independence.
On the other hand, the weak order defined in Example~\ref{Exp:precmm}
does not satisfy independence.\footnote{\citet{barbera2004},
following \citet{bossert2000}, falsely claim that the order defined
in Example~\ref{Exp:precmm} can be characterized by simple dominance, independence
and two other axioms. \citet{Arlegi2003} was the first to point out
that this is not the case, because $\preceq_{mm}$ does not satisfy
independence. He also provided a different axiomatic characterization
of this order.}

\begin{Exp}
Let $\preceq_{mm}$ be the weak order defined in Example~\ref{Exp:precmm}.
Consider $X = \{1,2,3,4\}$, $\cX = \pow$ and let $\leq$ be the natural linear order on $X$.
Then, $\{2\} \prec_{mm} \{1,3\}$ but $\{1,3,4\} \prec_{mm} \{2,4\}$.
Therefore, $\preceq_{mm}$ does not satisfy independence.
\end{Exp}

Furthermore, neither Fishburn's nor G\"ardenfors' extension satisfy 
independence. For example on $\mathcal{T}_{oy}$ both extensions set $\{2,4\} \prec \{4\}$
but for both $\{1,2,4\}$ and $\{1,4\}$ are incomparable. This violates independence.
However, it is possible to define a maxmin-based preorder that satisfies dominance and independence together.

\begin{Exp}\label{Exp:precmm2}
It can be checked that the following maxmin-based preorder satisfies
dominance and independence on every family of sets.
We define $\preceq_{pmm}$ by $A \preceq_{pmm} B$ for $A,B \in \cX$ if
\[\max(A) \leq \max(B) \text{ and } \min(A) \leq \min(B).\]
This relation is obviously reflexive. Furthermore, because 
$\leq$ is transitive, $\preceq_{pmm}$ is also transitive.
Therefore, $\preceq_{pmm}$ is a preorder.
We leave it to the reader to check that $\prec_{pmm}$ additionally
satisfies dominance and independence.

\end{Exp}

It turns out that for $\cX = \pow$ this is the minimal preorder that satisfies 
dominance and independence, i.e., every preoder on $\pow$ that satisfies both axioms
is an extension of $\preceq_{pmm}$.
We observe that $\preceq_{pmm}$ is not total as, for example, $\{1,3\}$ and $\{2\}$ are incomparable.

This raises the question if it is also possible to define a weak order that satisfies both dominance and independence.
In 1984 Yakar Kannai and Bezalel Peleg proved in a seminal paper that this is, in general, not possible \citep{kannai1984}.
To be more precise, they showed that there is no weak order that satisfies both axioms 
for $\cX = \pow$ if $|X| \geq 6$. We will deal with this result in detail in Section~\ref{sec:impos}.
Before, we introduce a strengthening of independence called strict independence
 that requires that adding the same element to two sets
does not change a strict preference. 

\begin{center}
\begin{mdframed}[style=mystyle,frametitle=Strict Independence]
\begin{minipage}{0.9\textwidth}
For all $A, B \in \mathcal{X}$ and for all $x \in X \setminus (A \cup B)$,
such that $A \cup \{x\}, B \cup \{x\} \in \mathcal{X}$:
\[A \prec B \text{ implies } A \cup \{x\} \prec B \cup \{x\}.\]
\end{minipage}
\end{mdframed}
\end{center}

The effect of strict independence on $\mathcal{T}_{oy}$ is very similar to the effect of independence.
A relation $\preceq$ on $\mathcal{T}_{oy}$ that satisfies
strict independence must set $\{1,2,4\} \prec \{1,4\}$ if it contains $\{2,4\} \prec \{4\}$
and $\{1,4\} \prec \{1,2,4\}$ if it contains $\{4\} \prec \{2,4\}$.
Dominance implies $\{2,4\} \prec \{4\}$, hence dominance
and strict independence together imply $\{1,2,4\} \prec \{1,4\}$.

Clearly, any relation that satisfies strict independence also satisfies independence.
Furthermore, any antisymmetric relation that satisfies independence automatically satisfies
strict independence.
We observe that, like in the case of dominance, a `set-based' formulation of
strict independence would coincide with the given definition if $\mathcal{X} = \mathcal{P}(X) \setminus \{\emptyset\}$.
However, strict independence is already a very strong axiom if defined as above. Moreover,
it seems natural to define independence and strict independence similarly. 
Therefore, we will focus on the formulation of strict independence given above and
leave a study of `set-strict-independence' to future work.

Similarly to independence, this is a desirable property, for example, whenever sets should be ranked according to some additive utility.
However, it is a significantly stronger axiom and is not satisfied by the preorder $\preceq_{pmm}$
defined in Example~\ref{Exp:precmm2}.

\begin{Exp}
Let $\preceq_{pmm}$ be the preorder defined in Example~\ref{Exp:precmm2}.
Furthermore, let $X = \{1,2,3\}$, let $\leq$ be the natural order on $X$ and $\cX = \pow$.
Then, $\preceq_{pmm}$ does not satisfy strict independence with respect to $\leq$.
For example $\{1\} \prec_{pmm} \{1,2\}$ but $\{1,3\} \not\prec_{pmm} \{1,2,3\}$.
\end{Exp}

Indeed, Salvador Barber{\`a} and Prasanta Pattanaik have shown that no preorder 
can satisfy dominance and strict independence if $|X| \geq 3$ \citep{barbera1984}.
We will discuss this result in more detail in Section~\ref{sec:impos}.
There are, however, important examples of lifted orders that always satisfy 
strict independence, like the very well-known lexicographic order.
This order is a generalization of 
the way that words are ordered in a lexicon based on the
alphabetical order of the letters \citep{FishburnLex}.

One main axiom remains, the so-called extension rule.

\begin{center}
\begin{mdframed}[style=mystyle,frametitle=The Extension Rule]
\begin{minipage}{0.9\textwidth}
For all $x, y\in X$, such that $\{x\}, \{y\} \in \mathcal{X}$:
\[
x < y \text{ implies }  \{x\} \prec \{y\}.
\]

\end{minipage}
\end{mdframed}
\end{center}

In $\mathcal{T}_{oy}$ the extension rule implies only $\{2\} \prec \{4\}$.
In some sense, the extension rule (or just extension for short) is the most basic axiom
considered in this paper. It states that the singleton sets in $\cX$ need to be ordered 
the same way as the elements of $X$. In most scenarios, this is a necessary requirement 
for the lifted order to be acceptable. However, there are exceptions.
For example under one interpretation called ``freedom of choice''\footnote{For an explanation 
of this interpretation see either \citet{pattanaik1990} or \citet{barbera2004}.}
it could be argued that all singletons 
should be rated equally \citep[see e.g.]{pattanaik1990}.
If one has to rank the whole powerset i.e., if we assume $\cX = \pow$,
then the extension rule is implied by dominance for every transitive relation
as $\{x\} \prec \{x,y\} \prec \{y\}$ is implied by dominance for all $x,y \in X$
such that $x < y$. However, if we drop the assumption that $\cX = \pow$
then there are families of sets on which we can define
an order that satisfies dominance and strict independence without 
satisfying the extension rule. 
Indeed, \citet{JAIR} have shown that there are families 
of sets where dominance and independence can be jointly satisfied with respect 
to a linear order $\leq$ but dominance, independence and the extension rule 
are incompatible with respect to $\leq$.

We observe that the four main axioms, if they are compatible,
do not necessarily characterize a unique order.
For example both of the following linear orders on $\mathcal{T}_{oy}$
satisfy all our main axioms.
\[\{2\} \prec \{1,2,4\} \prec \{1,4\} \prec  \{2,4\} \prec \{3,4\} \prec \{4\},\] 
\[\{1,2,4\}  \prec \{2\} \prec  \{2,4\} \prec \{3,4\}\prec \{1,4\} \prec \{4\}.\]

We finish this section with two important observations about our main axioms
that follows directly from their definition.

\begin{Obs}\label{Obs:Subset}
\LetX.
Furthermore, let $\preceq$ be a relation on $\cX$,
let $\cY \subseteq \cX$ be a subset of $\cX$
and let $\preceq_\cY$ be the restriction of $\preceq$
to $\cY$. Then, if $\preceq$ satisfies any of our main axioms 
with respect to $\leq$ then $\preceq_\cY$ must satisfy the same 
axioms with respect to $\leq$.
\end{Obs}

\begin{proof}
All four axioms are universal statements about the ordered set $(\cX,\preceq)$.
Hence, if they are true for a model $(\cX,\preceq)$ they are also true for all its 
submodels $(\cY, \preceq_\cY)$, i.e., for all tuples $(\cY, \preceq_\cY)$ such that
$\cY \subseteq \cX$ and $x \preceq_\cY y$ if and only if $x \prec y$ for all $x,y \in \cY$.
\end{proof}

Furthermore, we observe that all our main axioms are symmetric 
in the following sense.

\begin{Lem}\label{Lem:inverse}
\LetX.
Furthermore, let $\prec$ be a order on $\cX$.
Then, $\prec$ satisfies dominance with respect to a linear order $\leq$
if and only if $\prec^{-1}$ satisfies dominance with respect to $\leq^{-1}$.
Here, $R^{-1}$ denotes the inverse of a relation $R$.
The same holds for the extension rule, independence and strict independence.
\end{Lem}

\begin{proof}
Let $\preceq$ be an order on $\cX$ that satisfies dominance
with respect to $\leq$. We claim that $\preceq^{-1}$ satisfies 
dominance respect to $\leq^{-1}$.
Assume $A,A\cup \{x\} \in \cX$, then \[\forall y \in A(y <^{-1} x)\] implies 
\[\forall y \in A(y > x),\] which,  by assumption, implies $A \succ A \cup \{x\}$ and
hence $A \prec^{-1} A \cup \{x\}$.
Similarly, \[\forall y \in A(x <^{-1} y)\] implies $A \cup \{x\} \prec^{-1} A$.
The argument for the extension rule is similar.

Now, assume $\preceq$ is an order on $\cX$ that satisfies strict independence
with respect to $\leq$. Then, we claim that $\preceq^{-1}$ satisfies 
strict independence with respect to $\leq^{-1}$.
Assume $A,B,A\cup\{x\},B\cup\{x\} \in \cX$ and $A \prec^{-1} B$.
Then, $A \succ B$ and hence, as $\prec$ satisfies strict independence
by assumption, $A \cup \{x\} \succ B \cup \{x\}$
which implies $A \cup \{x\} \prec^{-1} B \cup \{x\}$.
The argument for independence is the same.
\end{proof}

\subsection{Impossibility Results}\label{sec:impos}

As it turns out, it is impossible to jointly satisfy our main axioms, as 
was shown in two impossibility results by \citet{kannai1984}
and \citet{barbera1984}.
Let us first discuss Kannai and Peleg's result in more detail.

\begin{center}
\begin{mdframed}[style=mystyle,frametitle=Theorem (\citeauthor{kannai1984})]
\begin{minipage}{0.9\textwidth}
Let $X$ be a set such that $|X| \geq 6$. Furthermore, let $\leq$ be a linear order on $X$ and $\cX = \pow$. 
Then, there is no weak order on $\cX$ that satisfies dominance and independence with respect to $\leq$.
\end{minipage}
\end{mdframed}
\end{center}

This result is, in several ways, tight. First of all,
we cannot drop the requirement that $\preceq$ is total, as we have seen in
Example~\ref{Exp:precmm2} that it is always possible to define a preorder that 
satisfies dominance and independence.

Furthermore, the assumption $\cX = \pow$ is also, in some sense, tight if $|X| = 6$. We only have to remove one set 
from $\cX$ to make dominance and independence compatible, though it has to be the right one.
For example, the proof of Kannai and Peleg's theorem never mentions any set containing $\max(X)$ and 
$\min(X)$ at the same time. Therefore, removing such a set will not make dominance and independence compatible.
However, removing the set containing the smallest and the second smallest element of $X$
suffices to make dominance, independence and additionally the extension rule jointly satisfiable.
We formulate the result w.l.o.g.\ for $X = \{1,2,3,4,5,6\}$.

\begin{Prop}[\citeauthor{JAIR}]\label{Prop:setminus}
Let $X$ be $\{1,2,3,4,5,6\}$. Furthermore, let $\leq$ be the natural linear order on $X$
and $\cX = \powerset{X} \setminus \{\emptyset, \{1,2\}\}$. 
Then, there is a weak order on $\cX$ that satisfies dominance, independence
and the extension rule with respect to $\leq$.
\end{Prop}

Finally, the requirement $|X| \geq 6$ is tight as dominance and independence
can be jointly satisfied if $|X| \leq 5$.
This follows directly from Proposition~\ref{Prop:setminus} but
was already proven much earlier by \citet{bandyopadhyay}.

\begin{Prop}[\citeauthor{bandyopadhyay}]\label{Cly:DomIndless5}
Let $X$ be a set such that $|X| \leq 5$. Furthermore, let $\leq$ be a linear order on $X$ and $\cX = \pow$. 
Then, there is always a weak order on $\cX$ that satisfies dominance and independence with respect to $\leq$.
\end{Prop}

While dominance and independence are incompatible for total orders, it turns out that dominance 
and strict independence are incompatible already for partial orders. This was proven by 
Salvador Barber{\`a} and Prasanta Pattanaik shortly after Kannai and Peleg published their result 
\citep{barbera1984}. Additionally, Barber{\`a} and Pattanaik's result requires only 
a smaller set of elements.
Finally, their result also holds if dominance is replaced by a weaker axiom called simple dominance.

\begin{Axm}[Simple Dominance]
For all $x, y\in X$, such that $\{x\}, \{y\}, \{x,y\} \in \mathcal{X}$ and $x < y $:
\[\{x\} \prec \{x,y\} \prec \{y\}.\]
\end{Axm}

Clearly, dominance implies simple dominance on all families of sets.
Furthermore, \citet{barbera1984} have shown that simple dominance and independence
can be jointly satisfied. In contrast, simple dominance and strict independence 
are incompatible, as was shown in the same paper by \citet{barbera1984}.

\begin{center}
\begin{mdframed}[style=mystyle,frametitle=Theorem (\citeauthor{barbera1984})]
\begin{minipage}{0.9\textwidth}
Let $X$ be a set such that $|X| \geq 3$. Furthermore, let $\leq$ be a linear order on $X$ and $\cX = \pow$. 
Then, there is no binary relation on $\cX$ that satisfies simple dominance and strict independence with respect to $\leq$.
\end{minipage}
\end{mdframed}
\end{center}

It is easy to see that the condition $|X| \geq 3$ is minimal as $\{1\} \prec \{1,2\} \prec \{2\}$ 
satisfies dominance and strict independence.
Furthermore, it is again the case that removing one element from $\pow$ for $|X| = 3$
suffices to make dominance, strict independence and the extension rule compatible.
 
\begin{Exp}\label{Exp:StrctIndMinus}
Consider $\cX = \powerset{X} \setminus \{\emptyset, \{3\}\}$.
Then, the following linear order satisfies dominance and strict independence:
\[\{1\} \prec \{1,2\} \prec \{1,3\} \prec \{1,2,3\} \prec \{2\} \prec \{2,3\}.\]
It is straightforward to check that dominance is satisfied.
For strict independence, we have to consider all pairs $A,B \in \cX$ and all $x \not \in A \cup B$
such that $A \cup \{x\}, B \cup \{x\} \in \cX$ holds.
First assume $x = 1$. Then, $A,B$ must be $\{2\}$ and $\{2,3\}$.
We see that $\{2\} \prec \{2,3\}$ and $\{1,2\} \prec \{1,2,3\}$ hold,
hence strict independence is satisfied for this pair.
Now assume $x = 2$. Then, $A,B$ must be $\{1\}$ and $\{1,3\}$.
Now, as $\{1\} \prec \{1,3\}$ and $\{1,2\} \prec \{1,2,3\}$ hold,
strict independence is also satisfied in this case.
Finally, consider the case that $x=3$. 
Then, $A$ and $B$ must be $\{1\}$, $\{1,2\}$ or $\{2\}$.
Now, we have $\{1\} \prec \{1,2\} \prec \{2\}$
and $\{1,3\} \prec \{1,2,3\} \prec \{2,3\}$.
Hence, strict independence is satisfied.
\end{Exp}

\subsection{Three Types of Orderability}\label{sec:orderability}

Given a set $X$, a family $\cX \subseteq \pow$ and a set of axioms $\mathcal{A}$,
we can distinguish three degrees to which the axioms in $\mathcal{A}$ can be
compatible on $\cX$.
First, they can be compatible for at least one linear order on $X$.
Second, they can be compatible with respect to a specific linear order $\leq$ on $X$.
Finally, they can be compatible for every linear order on $X$.
We can view these as a property of a family of sets with respect to a set of axioms.
The following definitions were first introduced by \citet{MalyTW18}.

\begin{Def}
Let $X$ be a set and $\cX \subseteq \pow$. Furthermore, let $\mathcal{A}$ be a set of axioms.
Then, we say that $\cX$ is \dots
\begin{itemize}
\item \dots\emph{weakly orderable with respect to $\mathcal{A}$} 
if \textbf{there is a linear order} $\leq^*$ on $X$ such that
there is a weak order on $\cX$ that satisfies all axioms in $\mathcal{A}$ 
with respect to $\leq^*$.
\item \dots\emph{$\leq$-orderable with respect to $\mathcal{A}$} for a linear
order $\leq$ on $X$,
if there is a weak order on $\cX$ that satisfies all axioms in $\mathcal{A}$ 
\textbf{with respect to $\leq$}.
\item \dots\emph{strongly orderable with respect to $\mathcal{A}$} 
if \textbf{for all linear orders} $\leq^*$ on $X$
there is a weak order on $\cX$ that satisfies all axioms in $\mathcal{A}$ 
with respect to $\leq^*$.
\end{itemize}
\end{Def} 

For convenience, we will define a shorthand for orderability with respect 
to (sub)sets of our main axioms.

\begin{Def}
\LetX.
Assume $\cX$ is $\leq$-orderable with respect to a set of axioms $\mathcal{A}$.
Then, we say $\cX$ is\dots
\begin{itemize}
\item \dots$\leq$-$DI$-orderable if $\mathcal{A}$ consists of dominance and independence.
\item \dots$\leq$-$DIE$-orderable if $\mathcal{A}$ consists of dominance, independence and the extension rule.
\item \dots$\leq$-$DI^S$-orderable if $\mathcal{A}$ consists of dominance and strict independence.
\item \dots$\leq$-$DI^SE$-orderable if $\mathcal{A}$ consists of dominance, strict independence and the extension rule.
\end{itemize}
We use the same notation also for strong and weak orderability.
\end{Def}

The following example demonstrates the use of these shorthands. 

\begin{Exp}\label{Exp}
Let $X = \{1,2,3\}$ and $\cX = \powerset{X} \setminus \{\emptyset, \{3\}\}$.
Furthermore, let $\leq^*$ be the natural linear order $1<^*2<^*3$.
Then, $\leq^*$ can be lifted to a linear order on $\cX$ that
satisfies dominance and strict independence,
as we have seen in Example~\ref{Exp:StrctIndMinus}.
Therefore, $\cX$ is $\leq^*$-$DI^S$-orderable and
also weakly $DI^S$-orderable.
On the other hand, if we consider the linear order $1<'3<'2$ on $X$, then no
linear order on $\cX$ satisfies dominance and strict independence
with respect to $\leq'$. This is because the proof of Barber{\`a} and Pattanaik's
impossibility result does not mention $\{2\}$, which 
has the same position under the natural linear order as $\{3\}$ 
has under $\leq'$.
Therefore, $\cX$ is not strongly $DI^S$-orderable.
If we consider independence instead of strict independence,
then Proposition~\ref{Cly:DomIndless5} implies that
$\cX$ is strongly $DIE$-orderable because it has less than $6$ elements.
\end{Exp}

The classical works on ranking sets of objects as surveyed by
\citet{barbera2004} do not distinguish different kinds of orderability,
because they only consider the case $\cX = \pow$.
In this case all three aforementioned types of 
orderability coincide, because for every permutation $\pi$ of $X$
we have $\pi(\pow) = \pow$.

\subsection{Reverse Independence and Reverse Strict Independence}\label{sec:additional}

In this section, we discuss two additional axioms,
which are obtained by reversing the direction
of independence or strict independence.
These will be useful when proving our main results.

\begin{Axm}[Reverse independence]
For all $A, B \in \mathcal{X}$ and for all $x \in X \setminus (A \cup B)$
such that $A \cup \{x\}, B \cup \{x\} \in \mathcal{X}$:
\[A \cup \{x\} \prec B \cup \{x\} \text{ implies } A \preceq B.\]
\end{Axm}

\begin{Axm}[Reverse strict independence]
For all $A, B \in \mathcal{X}$ and for all $x \in X \setminus (A \cup B)$
such that $A \cup \{x\}, B \cup \{x\} \in \mathcal{X}$:
\[A \cup \{x\} \prec B \cup \{x\} \text{ implies } A \prec B.\]
\end{Axm}

These axioms are very similar to independence and
strict independence, intuitively as well as technically.
One could argue that they are a slightly less natural
formulation of the same monotonicity idea. 
Independence is equivalent to its reverse counterpart
for total orders but both versions differ for partial orders.
For example, if $X = \{1,2,3\}$ and $\cX = \{\{1\}, \{2\}, \{1,3\}, \{2,3\}\}$
then the partial order only containing the preference $\{1\} \prec \{2\}$
does satisfy reverse independence but not independence. 
On the other hand, the partial order only containing the preference
$\{1,3\} \prec \{2,3\}$ satisfies independence but not reverse independence. 
Similarly, strict independence is equivalent to its reverse 
counterpart for linear orders but both versions differ for non-linear orders.

\begin{Prop}\label{Prop:Reverse}
\LetX.
Then, a total relation on $\cX$ satisfies independence if and only if it satisfies 
reverse independence. A total, antisymmetric relation satisfies strict independence if and only if
it satisfies reverse strict independence.
\end{Prop}

\begin{proof}
Let $\preceq$ be a total relation on $\cX$.
Furthermore, assume that $A,B, A\cup \{x\}$ and $B \cup \{x\}$ are in $\cX$.
We show that reverse independence implies independence.
The other direction is analogous.
Assume, that $A \prec B$ and $\preceq$ satisfies reverse independence.
Then, by totality, we must have either $A\cup \{x\} \preceq B \cup \{x\}$
or $B \cup \{x\} \prec A \cup \{x\}$. In the second case, reverse
independence would imply $B \prec A$ which contradicts our assumption 
that $A \prec B$ holds. Hence. $A \prec B$ always implies $A\cup \{x\} \preceq B \cup \{x\}$.
In other words, $\preceq$ satisfies independence.

Now, let $\preceq$ be a total, antisymmetric relation on $\cX$.
If $\preceq$ satisfies reverse strict independence,
it must, by definition, also satisfy reverse independence.
As we have proven above, this implies that $\preceq$
satisfies independence. Now, because $\preceq$ is 
an antisymmetric relation that satisfies independence it must,
by definition, also satisfy strict independence. 
The other direction is analogous.
\end{proof}

We observe that dominance is ``more compatible'' with
reverse strict independence than with strict independence.

\begin{Obs}
Let $X$ be a set with three elements, $\leq$ a linear order on $X$
and $\cX = \pow$. Then, there exists a total order on $\cX$ satisfying 
dominance and reverse strict independence with respect to $\leq$.
\end{Obs}

\begin{proof}
We assume w.l.o.g.\ that $X = \{1,2,3\}$ and $\leq$ is the natural linear 
order on $X$.
Then, we claim that 
\[\{1\} \prec \{1,2\} \prec \{1,2,3\} \sim \{1,3\} \prec \{2\} \prec \{2,3\} \prec \{3\}.\]
is a total order that satisfies dominance and reverse strict independence.
We have seen in Example~\ref{Exp:StrctIndMinus} that the restriction of $\preceq$
to $\cX \setminus \{3\}$ is a linear order that satisfies dominance 
and strict independence, hence also reverse strict independence.
Therefore, we only need to look at applications of dominance and 
reverse strict independence that involve $\{3\}$.
Clearly, the only applications of dominance 
involving $\{1,3\}$ are $\{1,3\} \prec \{3\}$
and $\{2,3\} \prec \{3\}$, both of which are satisfied.
Now, consider the case that $A \cup \{x\} \prec B \cup \{x\}$ holds
for $A, B \in \pow$ and $x \not \in A \cup B$.
Clearly $A\cup \{x\} = \{3\}$ or $B \cup \{x\} = \{3\}$
is not possible.
Therefore, we can assume that and $A = \{3\}$ or $B  = \{3\}$.
Assume first that $x = 2$. Then $B = \{3\}$ and 
$A = \{1\}$ or $A = \{1,3\}$. In both cases,
reverse strict independence is satisfied.
Now assume that $x = 1$. Then $B = \{3\}$ and 
$A = \{2\}$. For this pair
reverse strict independence is also satisfied.
\end{proof}

On the other hand, we have seen that reverse strict independence implies 
independence for total relations. Hence, there is no weak order that satisfies
dominance and reverse independence if $\cX = \pow$
for any set $X$ with $|X| \geq 6$.

\section{Complexity Results}\label{chap:Complexity}

In this section we present our results.
We study several problems related to $\leq$-orderability,
strong orderability and weak orderability. 
Together, these results give a nearly complete
picture of the complexity of deciding whether
dominance and (strict) independence are jointly satisfiable
for a given family. 

In the following, we formally define the studied
problems and describe the structure of this section.
The problems are named according to the following system:
First, it is indicated whether the problem
concerns strong, weak or $\leq$-orderability.
The last case is specified by the absence of the words
strong and weak. Then, the set of axioms in defined using
the same notation as for orderability. Finally,
the letters LO, WO, or PO declare whether we expect the lifted
order to be linear, weak or partial. 
Hence, for example, the problem of deciding whether
a family $\cX$ is $\leq$-$DI^S$-orderable
is called \textsc{$DI^S$-WO-Orderability}.

In the beginning of the chapter, we study $\leq$-orderability 
and strong orderability.
First we consider the problem of lifting a
linear order on $X$ to a linear order on $\cX \subseteq \pow$.
In this setting, independence and strict independence
coincide so we study this problem only for
strict independence. This gives us four problems to study.
Two of these problems are:

\begin{problem}
  \problemtitle{$DI^{S}$-LO-Orderability}
  \probleminput{A set $X$, a family of sets $\cX \subseteq \powerset{X}$
and a linear order $\leq$ on $X$.}
  \problemquestion{Is there a linear order on $\cX$ that satisfies dominance and 
   strict independence with respect to $\leq$?}
\end{problem}
 
\begin{problem}
  \problemtitle{Strong $DI^{S}$-LO-Orderability}
  \probleminput{A set $X$ and a family of sets $\cX \subseteq \powerset{X}$.}
  \problemquestion{Is there for every linear order $\leq$ on $X$ a linear order on $\cX$
   that satisfies dominance and strict independence with respect to $\leq$?}
\end{problem}

The other two problems are obtained by adding the extension rule to the requirements.
 
\begin{problem}
  \problemtitle{$DI^{S}E$-LO-Orderability}
  \probleminput{A set $X$, a family of sets $\cX \subseteq \powerset{X}$
and a linear order $\leq$ on $X$.}
  \problemquestion{Is there a linear order on $\cX$ that satisfies dominance,
   strict independence and the extension rule with respect to $\leq$?}
\end{problem}
 
\begin{problem}
  \problemtitle{Strong $DI^{S}E$-LO-Orderability}
  \probleminput{A set $X$ and a family of sets $\cX \subseteq \powerset{X}$.}
  \problemquestion{Is there for every linear order $\leq$ on $X$ a linear order on $\cX$
   that satisfies dominance, strict independence and the extension rule with respect to $\leq$?}
\end{problem}

We first prove that \textsc{Strong $DI^S$-LO-Orderability} is \NP-hard 
(Proposition~\ref{DI^S-NP-hard}) before we improve the result to $\Sigma_2^p$-completeness
in Theorem~\ref{PI1-LO}. The reduction used for this result 
will immediately also prove the \NP-completeness of
\textsc{$DI^S$-LO-Orderability} (Corollary~\ref{Cly:StrictNPcomp}).
Furthermore, we show that adding the extension rule does not change the complexity of the problems,
i.e., we show that \textsc{Strong $DI^SE$-LO-Orderability} is \NP-hard 
and  \textsc{$DI^SE$-LO-Orderability} is \NP-complete (Corollary~\ref{Cly:StrictExt}).

Next, we study the complexity of lifting a linear order on $X$ to a total, but not necessarily 
linear, order on $\cX$ that satisfies dominance and strict independence. This gives us again two problems:

\begin{problem}
  \problemtitle{$DI^{S}$-WO-Orderability}
  \probleminput{A set $X$, a linear order $\leq$ on $X$ and a family of sets $\cX \subseteq \powerset{X}$.}
  \problemquestion{Is $\cX$ $\leq$-$DI^{S}$-orderable?}
\end{problem}

\begin{problem}
  \problemtitle{Strong $DI^{S}$-WO-Orderability}
  \probleminput{A set $X$ and a family of sets $\cX \subseteq \powerset{X}$.}
  \problemquestion{Is $\cX$ strongly $DI^{S}$-orderable?}
\end{problem}

Modifying the reduction used in Proposition~\ref{DI^S-NP-hard}
we additionally prove that \textsc{$DI^{S}$-WO-Orderability} is \NP-complete
(Theorem~\ref{DI^S-NP-comp}).
A further modification of the reduction used for Proposition~\ref{DI^S-NP-hard},
shows that \textsc{Strong $DI^{S}$-WO-Orderability} is $\Pi_2^p$-complete (Theorem~\ref{PI1}).
As before we can define two more problems by adding the extension rule to the requirements,
\textsc{$DI^SE$-WO-Orderability} and 
\textsc{Strong $DI^SE$-WO-Orderability}.
The same complexity results also hold for these problems i.e.,
\textsc{$DI^SE$-WO-Orderability} is \NP-complete and
\textsc{Strong $DI^SE$-WO-Orderability} is $\Pi_2^p$-complete 
(Corollary~\ref{Cly:DI^SE-NP-comp} and \ref{Cly:DI^SE-Pi-comp}).

To conclude the section on lifting to a total order,  
we study the complexity of lifting a linear order on $X$
to a weak order on $\cX$ that satisfies dominance and independence.
This gives us the following problems:

\begin{problem}
  \problemtitle{Strong $DI$-WO-Orderability}
  \probleminput{A set $X$ and a family of sets $\cX \subseteq \powerset{X}$.}
  \problemquestion{Is $\cX$ strongly $DI$-orderable?}
\end{problem}

\begin{problem}
  \problemtitle{$DI$-WO-Orderability}
  \probleminput{A set $X$, a linear order $\leq$ on $X$ and a family of sets $\cX \subseteq \powerset{X}$.}
  \problemquestion{Is $\cX$ $\leq$-$DI$-orderable?}
\end{problem}

Again, we additionally study the problems obtained by adding the extension rule.
With another modification of the reduction used before, we can prove that
\textsc{Strong $DI$-WO-Orderability} is $\Pi_2^p$-complete 
and that \textsc{$DI$-WO-Orderability} is \NP-complete (Theorem~\ref{PI2}).
As before, both results hold if we add the extension rule, i.e.,
\textsc{Strong $DIE$-WO-Orderability} is $\Pi_2^p$-complete 
and \textsc{$DIE$-WO-Orderability} is \NP-complete (Corollary~\ref{DIE-PI-comp}).

The next problem that we consider is the problem of lifting to an partial order.
As we have seen in Example~\ref{Exp:precmm2}, it is always possible to find a 
preorder that satisfies dominance, independence and the extension rule.
Therefore, we only consider the problem of lifting a linear order on $X$ to a partial 
order on $\cX$ that satisfies either dominance and strict independence
or dominance, strict independence and the extension rule.
\citet{Maly2017} gave
a constructive, polynomial time procedure for constructing a  minimal
transitive, reflexive binary relation 
that satisfies dominance and strict independence, resp.\ 
dominance, strict independence and the extension rule.
In this paper, we consider strong orderability with respect to partial orders
and show by a reduction from \textsc{Taut} that the following problems are 
\coNP-complete (Theorem~\ref{partialCoNP}).

\begin{problem}
  \problemtitle{Strong $DI^{S}$-PO-Orderability}
  \probleminput{A set $X$ and a family of sets $\cX \subseteq \powerset{X}$.}
  \problemquestion{Is there for every linear order $\leq$ on $X$ a partial order on $\cX$
   that satisfies dominance and strict independence with respect to $\leq$?}
\end{problem}

\begin{problem}
  \problemtitle{Strong $DI^{S}E$-PO-Orderability}
  \probleminput{A set $X$ and a family of sets $\cX \subseteq \powerset{X}$.}
  \problemquestion{Is there for every linear order $\leq$ on $X$ a partial order on $\cX$
   that satisfies dominance, strict independence and the extension rule with respect to $\leq$?}
\end{problem}

Afterwards, we discuss succinctly represented families of sets.
We introduce the historical and necessary technical background 
on problems that are represented by boolean circuits. 
Then, we show for most of the problems that we studied in this
chapter that their complexity increases exponentially 
if the instances are represented by boolean circuits.

Next, we shift our attention to weak orderability. In this context, we only study 
strict independence. We leave the question whether the following results also hold for independence
for future work. The first problem that we consider is 
weak $DI^S$-orderability.

\begin{problem}
  \problemtitle{Weak $DI^{S}$-WO-Orderability}
  \probleminput{A set $X$ and a family of sets $\cX \subseteq \powerset{X}$.}
  \problemquestion{Is $\cX$ weakly $DI^{S}$-orderable?}
\end{problem}

We show with a reduction from \textsc{Betweenness}
that this problem is \NP-complete (Theorem~\ref{WeakNP}).
A close inspection of the proof shows that this 
also holds if we additionally require the extension rule
and if we require the lifted order to be linear
or only partial. This means \textsc{Weak $DI^{S}$-PO-Orderability},
\textsc{Weak $DI^{S}$-LO-Orderability} as well as \textsc{Weak $DI^{S}E$-PO-Orderability},
\textsc{Weak $DI^{S}E$-WO-Orderability} and
\textsc{Weak $DI^{S}E$-LO-Orderability} are all \NP-complete (see Corollary~\ref{Cly:Weak}) .

We conclude the paper by exploring if strengthening dominance is 
a viable way to reduce the complexity of the studied problems.
We restrict our attention to $\leq$-$DI^S$-orderability
and show that it stays \NP-hard 
for all ``reasonable'' strengthenings of dominance (Theorem~\ref{StrengthDom1}).
Formally, we say that an axiom $A$ is a \emph{reasonable strengthening of dominance} if
\begin{itemize}
\item $A$ implies dominance and
\item $A$ is implied by a very strong axiom called maximal dominance (Axiom~\ref{MaxDom}).
\end{itemize}

Table~\ref{tab:1} summarizes the results proven in this chapter.
The [Thm] column specifies the Theorem, Proposition or
Corollary in which the result is proven. ``known'' indicates that a result is folklore.
Results marked with $(\star)$ were already shown by \citet{Maly2017}.

\begin{table}
\begin{center}
\begin{tabular}{lllll}
\toprule
    Orderability & Dom + Ind &[Thm] & Dom + strict Ind& [Thm]\\
\midrule
                        $\leq$-PO-ord. &  always & known & in \P{} $(\star)$ & \ref{Cly:DI^S-partial}\\
                        strong PO-ord. &  always & known & \coNP-c. $(\dagger)$ & \ref{partialCoNP} \\
                        weak PO-ord. &  always & known & \NP-c. & \ref{Cly:Weak}\\  
\rowcolor{gray!25}      $\leq$-WO-ord.  & \NP-c. $(\star)$& \ref{PI2} & \NP-c. & \ref{DI^S-NP-comp}\\
\rowcolor{gray!25}      strong WO-ord. &  $\Pi_2^p$-c. $(\dagger)$ & \ref{PI2} & $\Pi_2^p$-c. & \ref{PI1}\\
\rowcolor{gray!25}      weak WO-ord.  & \emph{open}   &  n/a & \NP-c. & \ref{WeakNP}\\ 
                        $\leq$-LO-ord.  & \NP-c. & \ref{PI2-LO} & \NP-c.$(\star)$ & \ref{Cly:StrictNPcomp}\\
                        strong LO-ord. &  $\Pi_2^p$-c. & \ref{PI2-LO}  & $\Pi_2^p$-c. $(\dagger)$ & \ref{PI1-LO}\\
                        weak LO-ord.  & \emph{open}    & n/a  & \NP-c. & \ref{Cly:Weak} \\ 
\rowcolor{gray!25}      succ. strong PO-ord. &  always  & known & \coNEXP-c. $(\dagger)$& \ref{Succinct2}\\
\rowcolor{gray!25}      succ. $\leq$-WO-ord. & \NEXP-c. $(\dagger)$& \ref{Succinct3} & \NEXP-c. & \ref{Succinct3} \\
\rowcolor{gray!25}      succ. strong WO-ord. & \NEXP-hard $(\dagger)$& \ref{Succinct3} & \NEXP-hard & \ref{Succinct3} \\
\rowcolor{gray!25}      succ. $\leq$-LO-ord. & \NEXP-c. & \ref{Succinct3}& \NEXP-c. $(\dagger)$& \ref{Succinct1}\\
\rowcolor{gray!25}      succ. strong LO-ord. & \NEXP-hard & \ref{Succinct3}& \NEXP-hard $(\dagger)$& \ref{Succinct1} \\
 \end{tabular}
  \caption{The complexity of orderability with respect to dominance and (strict) independence. $(\star)$ indicates that the result was already shown by \citet{Maly2017}. $(\dagger)$ indicates that the result was alredy contained in the conference version of this paper \citep{AAAIpaper}.}
  \label{tab:1}
\end{center}
\end{table}

\subsection{$\leq$-Orderability and Strong Orderability}\label{sec:leqstrongord}

In this section, we discuss the complexity of several 
variants of $\leq$-orderability and strong orderability
with respect to some subsets of our main axioms. 
The unifying feature of the problems discussed in this section 
is the fact that their hardness can be proven by a variation 
of the same reduction from either \textsc{Sat} or $\Pi_2$-\textsc{Sat}.

\subsubsection*{$DI^S$-LO-Orderability}

We show first that \textsc{Strong $DI^S$-LO-Orderability} is \NP-hard,
even tough we will improve this result in Corollary~\ref{PI1-LO} by showing that 
\textsc{Strong $DI^S$-LO-Orderability} is $\Pi_2^p$-complete.
This approach allows us to present the simplest form of a 
reduction from \SAT{} that will be used -- with some 
modifications -- to prove several other hardness results 
in this chapter.

Let us first give an intuitive description of the main ideas 
of the proof.
The goal of the reduction is to encode a $3$-CNF $\phi$ as 
a families of sets $\cX$.
First the variables in $\phi$ are encoded:
Every variable $V_i$ in $\phi$ will be encoded by two sets
$X^\mathrm{t}_i$ and $X_i^\mathrm{f}$. 
Then, we can equate every linear order $\preceq$ on $\cX$
with a truth assignment to the variables in $\phi$ 
by saying that $V_i$ is set to true if 
$X^\mathrm{f}_i \prec X_i^\mathrm{t}$
and $V_i$ is set to false if 
$X^\mathrm{t}_i \prec X_i^\mathrm{f}$.
Because $\prec$ is a linear order and therefore 
antisymmetric and total, this defines a complete and consistent 
truth assignment. Then, we will add sets to $\cX$ that lead to a cycle
in every linear order $\preceq$ that satisfies dominance and strict independence with respect to $\leq$
if $\preceq$ does not encode a satisfying truth assignment to $\phi$.
To achieve this, we will use the following observation:
If $\leq$ is a linear order on a set $X$
and $A$ and $B$ are two subsets of $X$
such that 
\[\min(B) < \min (A) < \max(A) < \max(B)\]
then there always exists a collection of sets $\cY$ such 
that any linear order $\preceq$ on $\cY \cup \{A,B\}$ that satisfies
dominance and strict independence with respect to $\leq$
has to set $A \prec B$. At the same time, there also has to exist 
a collection of sets $\cY^*$ such 
that any linear order $\preceq$ on $\cY^* \cup \{A,B\}$ that satisfies
dominance and strict independence with respect to $\leq$
has to set $B \prec A$.
Let us illustrate this by an example.

\begin{Exp}
Let $X = \{1,\dots,5\}$ and let $\leq$ be the natural linear order on $X$.
Now consider $A = \{3\}$, $B = \{2,3,4\}$.
Then clearly 
\[\min(B) < \min (A) < \max(A) < \max(B).\]
First, we claim that for the following collection
\[\cY = \{\{1,2,3,4\}, \{1,3\}, \{1\},\{1,2\},\{1,2,3\}\}\]
any linear order $\preceq$ on $\cY \cup \{A,B\}$ that satisfies
dominance and strict independence with respect to $\leq$
has to set $A \prec B$. 
Assume for the sake of contradiction that there is a linear order $\preceq$ on $\cY \cup \{A,B\}$
with $B \prec A$ that satisfies
dominance and strict independence with respect to $\leq$.
Then, strict independence implies 
\[B \cup \{1\} = \{1,2,3,4\} \prec \{1,3\} = A \cup \{1\}.\]
However, by dominance we have $\{1\} \prec \{1,2\}$ 
and hence by strict independence and dominance
\[\{1,3\} \prec \{1,2,3\} \prec \{1,2,3,4\},\]
a contradiction.
On the other hand, it can be checked that 
\[\{1\} \prec \{1,2\} \prec \{1,2,3\} \prec \{1,3\} \prec \{3\} \prec \{1,2,3,4\} \prec \{2,3,4\}\]
is a linear order on $\cY \cup \{A,B\}$ that satisfies
dominance and strict independence with respect to $\leq$.

By a similar argument, we can see that  
\[\cY^* = \{\{2,3,4,5\}, \{3,5\}, \{5\},\{4,5\},\{3,4,5\}\}.\]
has the property that 
that any linear order $\preceq$ on $\cY^* \cup \{A,B\}$ that satisfies
dominance and strict independence with respect to $\leq$
has to set $B \prec A$.
\end{Exp}
 
We will use this observation and define the sets encoding variables 
in a way such that for all $a,b \in \{\mathrm{t}, \mathrm{f}\}$ and
$i,j \leq n$, where $n$ is the number of variables in $\phi$, we have either 
\[\min(X_i^a) < \min (X_j^b) < \max(X_j^b) < \max(X_i^a)\]
or 
\[\min(X_j^b) < \min (X_i^a) < \max(X_i^a) < \max(X_j^b).\]
This can be ensured for example by a construction 
where every set has a common middle part and a unique 
minimal and maximal element. Then, for every new set
we increase the minimal and decrease the maximal
element at the same time (See Figure~\ref{fig:variables}).

\begin{figure}
\begin{center}
\begin{tikzpicture}[scale=0.44]
\fill (1,20) circle (4pt); 
\fill (2,18) circle (4pt); 
\fill (5,14) circle (4pt); 
\fill (20,20) circle (4pt); 
\fill (19,18) circle (4pt); 
\fill (17,14) circle (4pt); 
\draw[line width=4pt] (6,20) -- (16,20);
\draw[line width=4pt] (6,18) -- (16,18);
\draw[line width=4pt] (6,14) -- (16,14);
\draw[line width=1pt, loosely dotted] (11,17.5) -- (11,14.5);
\end{tikzpicture}
\caption{Sets encoding variables.}
\label{fig:variables}
\end{center}
\end{figure}
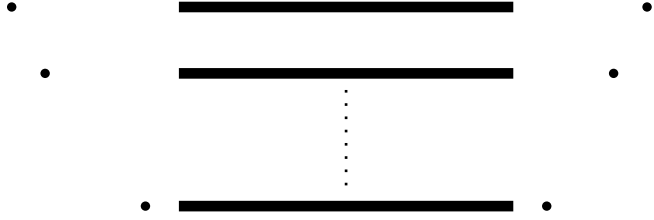

Then we can enforce any preference we need between sets encoding
different variables by adding the correct collection of set.
We use this to enforce for every clause preferences 
that lead to a contradiction whenever no literal in the clause
is satisfied. Consider for example the clause $C = x_1 \vee \neg x_2$.
Then, we enforce the preferences
\[X_1^\mathrm{f} \prec X_2^\mathrm{f} \text{ and }
X_2^\mathrm{t} \prec X_1^\mathrm{t}\]
Now consider a linear order $\preceq$ that contains 
\[X_1^\mathrm{t} \prec X_1^\mathrm{f} \text{ and } X_2^\mathrm{f} \prec X_2^\mathrm{t}\]
and hence encodes an assignment that sets $x_1$ to false and $x_2$ to true.
This assignment does not satisfy $C$ and indeed $\preceq$ contains the following cycle
\[X_1^\mathrm{f} \prec X_2^\mathrm{f} \prec X_2^\mathrm{t} \prec X_1^\mathrm{t} \prec X_1^\mathrm{f}\]
On the other hand,  
\[X_1^\mathrm{f} \prec X_2^\mathrm{f} \prec X_2^\mathrm{t} \prec X_1^\mathrm{t}\]
is a linear order that is compatible with the enforced preferences and
encodes the satisfying assignment that sets $x_1$ and $x_2$ to true. 

Using this idea, we can add for every clause sets that lead to a cycle
if the clause is not satisfied by the assignment coded by a linear order.
The main technical difficultly of the proof will be to implement this approach
in a way that ensures that
no cycle occurs if $\preceq$ encodes a satisfying truth assignment.

\begin{Prop}\label{DI^S-NP-hard}
\textsc{Strong $DI^S$-LO-Orderability} is \NP-hard.
\end{Prop}

\begin{proof}
Let $\phi$ be an instance of \textsc{Sat} with $n$ variables and $m$ clauses.
We will produce an instance $(X,\cX)$ of \textsc{Strong $DI^S$-LO-Orderability}.
Furthermore, we will fix a specific linear order $\leq$ on $X$
such that there is a linear order on 
$\cX$ that satisfies dominance and strict independence with respect to $\leq$ only if $\phi$ is satisfiable. 
Furthermore, we want to make sure that we can use a satisfying assignment of $\phi$
to construct for any arbitrary linear order $\leq'$ on $X$ 
a linear order $\preceq$ on $\cX$ that satisfies dominance and strict independence.

First, we construct the set of elements $X$. 
For every variable $V_i$, the set $X$ contains elements
$x_{i,1}^-, x_{i,2}^-, x_{i,1}^+$ and $x_{i,2}^+$.
These will be used to construct the sets for different 
variables. Furthermore, it contains for every clause $C_j$
variables $z_j^a, y_{j}^a, \min_j^a$  and $\max_j^a$ for $a \in \{1,2,3\}$.
In the following, we will call $\min_j^a$ and $\max_j^a$ extremum elements.
These will be used to ensure that only orders encoding a
satisfying assignment to the variables in $C_j$ can satisfy dominance and strict independence.
Finally, it contains 
two elements $v_1$ and $v_2$.
These will determine the ``orientation'' of a linear order $\leq'$ on $X$. 
In general, the order lifting problem is symmetric
(see Lemma~\ref{Lem:inverse}) whereas \Sat{} is not symmetric
with respect to truth. To overcome this difficultly,
the preference between $v_1$ and $v_2$ determines
if $X^\mathrm{f}_i \prec X_i^\mathrm{t}$ means that $V_i$ is true
or that $V_i$ is false. This way, $\leq'$ and 
$\leq'^{-1}$ encode the same truth assignment.
Next, we fix a linear order $\leq$ that we define by:
\begin{multline*}
\Min_1^1 < \Min_1^2 < \dots < \Min_m^3 < x^-_{1,1} < x^-_{1,2} <\dots
 < x^-_{n,2} \\  <v_1 < v_2 
< z_1^1 < z_1^2 \dots < z_m^3 < y_1^1 < y_1^2 < \dots
  < y_m^3 < \\ x^+_{1,1} < x^+_{1,2} < \dots 
< x^+_{n,2} < \Max_1^1 < \Max_1^2 < \dots < \Max_m^3
\end{multline*}
We call this the critical linear order.
Now, we construct the family $\cX$. 
Our first goal is to ensure that there does not exists a linear order 
on $\cX$ that satisfies dominance and strict independence with respect
to $\leq$ if $\phi$ is not satisfiable.
In the following, we write $Y := \{x \in X \mid v_1 \leq x \leq y^3_m\}$.
First we add the sets representing the variables of $\phi$.
We add for every variable $V_i$
sets $X^\mathrm{t}_i = Y \cup \{x_{i,1}^-,x_{i,1}^+\}$
and $X^\mathrm{f}_i = Y \cup \{x_{i,2}^-,x_{i,2}^+\}$.
We call these the Class 1 sets and write $Cl_1$ for the collection of all Class 1 sets.

Now, let $C_i$ be a clause with variables $V_j, V_k$ and $V_l$.
We want to ``enforce'' specific preferences between the sets representing 
$V_j, V_k$ and $V_l$ depending on whether they appear positively or negatively in $C_i$.
However, this could lead to problems if the same variables also occur in another clause.
Therefore, we add what could be considered local instantiations of the sets representing
the variables $V_j, V_k$ and $V_l$:
\[
X^\mathrm{t}_j \setminus \{y_i^1\}, X_j^\mathrm{f} \setminus \{y_i^1\}, X^\mathrm{t}_k \setminus \{y_i^2\},
X_k^\mathrm{f} \setminus \{y_i^2\}, X^\mathrm{t}_l \setminus \{y_i^3\}  \mbox{ and } X_l^\mathrm{f} \setminus \{y_i^3\}.
\]
We call these the Class 2 sets and write $Cl_2$ for the collection of all Class 2 sets.
Now, let $\preceq$ be a linear order on $\cX$ that satisfies strict independence. 
Then, it also satisfies reverse strict independence by Proposition~\ref{Prop:Reverse}.
By reverse strict independence
we know that for $\preceq$ the preference between 
$X^\mathrm{t}_j \setminus \{y_i^1\}$ and $X_j^\mathrm{f}\setminus \{y_i^1\}$ must be the same as the 
preference between $X^\mathrm{t}_j$ and $X_j^\mathrm{f}$. The same holds for the other two variables.
In this sense, these ``local instantiations'' correctly reflect the truth assignment
encoded by any linear order on $\cX$, if the linear order satisfies strict independence.
On the other hand, any preference between local instantiations of sets representing different variables,
say $X^\mathrm{t}_j \setminus \{y_i^1\}$ and $X_k^\mathrm{f}\setminus \{y_i^2\}$
stays local, because $y_i^1 \neq y_i^2$.

Now, if all variables occur positively in $C_i$,
we add sets such that $X_j^\mathrm{f}\setminus \{y_i^1\} \prec X^\mathrm{t}_k \setminus \{y_i^2\}$,
$X_k^\mathrm{f}\setminus \{y_i^2\} \prec X^\mathrm{t}_l \setminus \{y_i^3\}$ and
$X_l^\mathrm{f}\setminus \{y_i^3\} \prec X^\mathrm{t}_j \setminus \{y_i^1\}$
must hold for any order $\preceq$ on $\cX$ that satisfies dominance and strict independence
with respect to $\leq$.
We call this enforcing these preferences.
Then, we get a contradiction if $X^\mathrm{t}_a \setminus \{y_i^b\} \prec X_a^\mathrm{f}\setminus \{y_i^b\}$ for 
all $a \in \{j,k,l\}$, i.e., if $\preceq$ corresponds to a truth assignment where $V_j$, $V_k$ and $V_l$, are false because
\[
X^\mathrm{t}_j \setminus \{y_i^1\} \prec X_j^\mathrm{f}\setminus \{y_i^1\} \prec X^\mathrm{t}_k \setminus \{y_i^2\} \prec 
X_k^\mathrm{f}\setminus \{y_i^2\} \prec X^\mathrm{t}_l \setminus \{y_i^3\} \prec
X_l^\mathrm{f}\setminus \{y_i^3\} \prec X^\mathrm{t}_j \setminus \{y_i^1\}
\]
holds and implies $X^\mathrm{t}_j \setminus \{y_i^1\} \prec X^\mathrm{t}_j \setminus \{y_i^1\}$
by transitivity.
If a variable, say $V_j$, occurs negatively in $C_i$, we switch $X^\mathrm{t}_j$ and $X_j^\mathrm{f}$
and enforce $X_j^\mathrm{t}\setminus \{y_i^1\} \prec X^\mathrm{t}_k \setminus \{y_i^2\}$ and
$X_l^\mathrm{f}\setminus \{y_i^3\} \prec X^\mathrm{f}_j \setminus \{y_i^1\}$.

Next, we show how we can enforce these preferences. Assume we want to enforce 
$X_j^a\setminus \{y_i^1\} \prec X^b_k \setminus \{y_i^2\}$ for $a,b \in \{\mathrm{t},\mathrm{f}\}$.
We add 
\[\{z_i^1\}, \{z_i^1, \Max_i^1\} \text{ and }(X_j^a\setminus \{y_i^1\}) \cup \{\Max_i^1\}.\]
Our goal is to enforce $(X_j^a\setminus \{y_i^1\}) \cup \{\Max_i^1\}\prec \{z_i^1, \Max_i^1\} $
which forces by reverse strict independence
$X_j^a\setminus \{y_i^1\} \prec \{z_i^1\}$.
Then, we enforce 
$\{z_i^1\} \prec X_k^b \setminus \{y_i^2\}$
to get by transitivity
$X_j^a\setminus \{y_i^1\} \prec X^b_k \setminus \{y_i^2\}$
as desired.
To enforce 
\[(X_j^a\setminus \{y_i^1\}) \cup \{\Max_i^1\}\prec \{z_i^1, \Max_i^1\}\]
we add a sequence of sets $A_1, A_2, \dots, A_l$ such that
\begin{itemize}
\item $A_1 = (X_j^a\setminus \{y_i^1,z_i^1)\} \cup \{\Max_i^1\}$,
\item $A_{i+1} = A_i \setminus \{\min_\leq (A_i)\}$
\item and $A_l = \{\Max_i^1\}$.
\end{itemize}
This enforces by dominance $A_1 \prec A_2 \prec \dots \prec A_l$
which enforces by transitivity
\[A_1 = (X_j^a\setminus \{y_i^1,z_i^1\})\cup \{\Max_i^1\} \prec \{\Max_i^1\} = A_l.\]

Finally, as desired, it follows from strict independence that
$(X_j^a\setminus \{y_i^1\}) \cup \{\Max_i^1\}\prec \{z_i, \Max_i^1\}$.
Using the same idea and $\Min_i^1$ we enforce
$\{z_i^1\} \prec X_k^b \setminus \{y_i^2\}$
finishing the construction for
$X_j^a\setminus \{y_i^1\} \prec X^b_k \setminus \{y_i^2\}$.
We enforce the other preferences for that clause, i.e.,
$X_k^c\setminus \{y_i^2\} \prec X^d_l \setminus \{y_i^3\}$ and
$X_l^e\setminus \{y_i^3\} \prec X^a_j \setminus \{y_i^1\}$ for $c,d,e \in \{\mathrm{t},\mathrm{f}\}$,
similarly using $z_i^2, \Max_i^2$ and $\Min_i^2$ resp.\
$z_i^3, \Max_i^3$ and $\Min_i^3$.
We repeat this procedure for every clause.
We call the sets added in this step the Class 3 sets
and write $Cl_3$ for the collection of all Class 3 sets.
Furthermore, we write $Cl_3^+$ for the Class 3 sets that
contain an element $\Max_i^a$ for some $i$ and $a$.
Similarly, we write $Cl_3^-$ for the Class 3 sets that
contain an element $\Min_i^a$ for some $i$ and $a$.
Finally, we write $Cl_3^0$ for all other sets in Class 3
(which are all of the form $\{z_i^a\}$).
Now, by construction, $\cX$ can only be $DI^S$-orderable with respect to $\leq$ 
if $\phi$ is a positive instance of \textsc{Sat}.

Next, we pick an arbitrary linear order $\leq'$ on $X$.
We distinguish two cases $v_1 <' v_2$ and $v_2 <' v_1$.
By Lemma~\ref{Lem:inverse} it suffices to show that a linear
satisfying dominance and strict independence exists
in the first case, because $v_2 <' v_1$ implies $v_1 <'^{-1} v_2$.
Hence, we can assume in the following w.l.o.g.\ $v_1 <' v_2$.
Now, we want to construct a linear order $\preceq$ on $\cX$
that satisfies dominance and strict independence with respect to $\leq'$
if $\phi$ is satisfiable. For the readers convenience,
we summarize which sets are contained in $\cX$ and
hence must be taken into account when constructing $\preceq$:

\medskip

\noindent \fbox{
\begin{minipage}{0.96\textwidth}
\begin{description}
\item[Class 1:] $X^\mathrm{t}_i = Y \cup \{x_{i,1}^-,x_{i,1}^+\}$
and $X^\mathrm{f}_i = Y \cup \{x_{i,2}^-,x_{i,2}^+\}$ for every variable $V_i$ of $\phi$.
\item[Class 2:] $X^\mathrm{t}_j \setminus \{y_i^a\}$ and $X_j^\mathrm{f} \setminus \{y_i^a\}$,
where $a \in \{1,2,3\}$ and $i$ and $j$ are such that $V_j$ is in clause $C_i$.
\item[Class 3:] For each $i \leq m$ and $a \in \{1,2,3\}$ 
and for some $b,c \in \{\mathrm{t}, \mathrm{f}\}$, $d \in \{1,2,3\}$
and $j,k$ such that $V_j$ and $V_k$ are variables in clause $C_i$:
\begin{itemize}
\item $\{z_i^a\}$, $\{z_i^a, \max_i^a\}$, $\{z_i^a, \min_i^a\}$,
$(X_j^b \setminus \{y_i^a\}) \cup \{\Max_i^a\}$, $(X_k^c \setminus \{y_i^d\})\cup\{\Min_i^a\}$,
\item $A_1 = (X_j^b \setminus \{y_i^a, z_i^a\}) \cup \{\max_i^a\},
\dots, A_{l+1} = A_l \setminus \{\min_\leq(A_l)\}, \dots, A_o = \{\max^a_i\}$,
\item $B_1 = (X_k^c \setminus \{y_i^d, z_i^a\}) \cup \{\min_i^a\},
\dots, B_{l+1} = B_l \setminus \{\max_\leq(B_l)\}, \dots, B_p = \{\min^a_i\}$.
\end{itemize}
\end{description}
\end{minipage}
}

\bigskip

Now, we construct the order $\preceq$ on $\cX$ in several steps.
First, we construct from a satisfying assignment of $\phi$
an order on $Cl_1 \cup Cl_2$.

\paragraph{Ordering the sets in $Cl_1 \cup Cl_2$:}
Intuitively, we just order all pairs
$(X_i^\mathrm{f},X_i^\mathrm{t})$ according to the satisfying
assignment of $\phi$ and then `project' this order down
on the Class 2 sets, i.e., we order
$X^\mathrm{f}_j \setminus \{y_i^a\}$ and  $X_j^\mathrm{t} \setminus \{y_i^a\}$
the same way as $(X_i^\mathrm{f}$ and $X_i^\mathrm{t})$.
Finally, we want to extend this partial order to a linear order on
$Cl_1 \cup Cl_2$. In doing so, we only have to pay attention 
to the two possible applications of dominance,
namely if $y_i^a$ is either the minimal or maximal element of $X_i^\mathrm{f}$ resp.\ $X_i^\mathrm{t}$.
(All other applications
of dominance and independence are already covered by the order implied by the 
assignment and its projection.) Formally, we construct the order as follows.

We start by ordering the sets $X_i^\mathrm{f}$ and $X_i^\mathrm{t}$ according to the
satisfying assignment of $\phi$, i.e., $X_i^\mathrm{t} \prec X_i^\mathrm{f}$ if
$V_i$ is false in the assignment and $X_i^\mathrm{f} \prec X_i^\mathrm{t}$ if it is true.
Then, we project this order down to the Class 2 sets by reverse strict independence.
Furthermore, we add the preferences that we enforced in the previous section.
Finally, we take the transitive closure of this order.
It is clear by construction
that this is an acyclic partial order
if and only if $\phi$ is satisfiable.
Now, for any clause $C_i$, we fix any linear order on the sets
\[X^\mathrm{f}_j \setminus \{y_i^1\}, X_j^\mathrm{t} \setminus \{y_i^1\}, X^\mathrm{f}_k \setminus \{y_i^2\},\\
X_k^\mathrm{t} \setminus \{y_i^2\}, X^\mathrm{f}_l \setminus \{y_i^3\}
\mbox{ and } X_l^\mathrm{t} \setminus \{y_i^3\}.\]
that extends this order.

For the Class 1 sets we have ordered all pairs $(X_i^\mathrm{f},X_i^\mathrm{t})$
but we still have to fix an order between these pairs. 
For the Class 2 sets, we have fixed an order on between all sets introduced
for a single clause, but we have to fix an order between sets from different clauses.
Now, we observe that $A,A\cup\{x\} \in Cl_1 \cup Cl_2$ implies that
$A \in Cl_1$, $A \cup\{x\} \in Cl_2$ and $x = y_i^a$ for $j\leq m$ and $a\leq 3$
as all Class 1 sets differ from all other Class 1 sets in 
at least two elements and all Class 2 sets differ from
all other Class 2 sets in at least two elements. 
Hence the only possible application of strict independence
on Class 1 and 2 is the one already covered by construction.
Dominance is applicable only if $y_i^a$ for some $i$ and $a$ is
the minimal or maximal element of the set it gets removed from.
Observe that $y_i^a$ can only be the maximal element 
of the set it is removed from if $y_j^b < y_i^a$ holds for all $j \neq i$
and $b \neq a$. Hence, there can only be one case where $y_i^a$ is the maximal
element of the set it is removed from and, by a similar argument, one instance
where it is the minimal element. Therefore, there are at most two possible
applications of dominance.
We fix two orders, one on the pairs that will be used to order the Class 1 sets
and one on the clauses that will be used to order the Class 2 sets. 
We have to make sure that these orders are compatible with these applications
of dominance.

First, assume the minimal element $y_{i^-}^{a^-}$ of the form $y_i^a$
and the maximal element $y_{i^+}^{a^+}$ of the form $y_i^a$ are used 
for the same clause. Let $X_j^b$ and $X_k^c$ be the sets
such that $X_j^b \setminus \{y_{i^-}^{a^-}\} \in \cX$ and
$X_k^c \setminus \{y_{i^+}^{a^+}\} \in \cX$ holds.
Then, by construction $j \neq k$, as, for a given clause, we remove the same
element $y_i^a$ from both $X^\mathrm{t}_j$ and $X^\mathrm{f}_j$. 
In that case
we fix any linear order $\leq''$ on the pairs $(X_i^\mathrm{f},X_i^\mathrm{t})$ such that 
$(X_j^\mathrm{f},X_j^\mathrm{t}) \leq'' (X_k^\mathrm{f},X_k^\mathrm{t})$ holds
and an arbitrary order on the clauses.
Then, we set $X_i^\mathrm{f} \prec A$ and $X_i^\mathrm{t} \prec A$ for every
$A \in Cl_2$ if $(X_i^\mathrm{f},X_i^\mathrm{t}) <'' (X_j^\mathrm{f},X_j^\mathrm{t})$.
Furthermore, we set $A \prec X_i^\mathrm{f}$ and $A \prec X_i^\mathrm{t}$ for every
$A \in Cl_2$ if $(X_j^\mathrm{f},X_j^\mathrm{t})<''(X_i^\mathrm{f},X_i^\mathrm{t})$.
This is obviously a linear order and we have
$X_j^b \prec X_j^b \setminus \{y_{i^-}^{a^-}\}$
and $X_k^c \setminus \{y_{i^+}^{a^+}\} \prec X_k^c$ for $b,c \leq 2$.
Hence the constructed order on $Cl_1 \cup Cl_2$ satisfies dominance.

Now, assume the minimal element $y_{i^-}^{a^-}$ of the form $y_i^a$
and the maximal element $y_{i^+}^{a^+}$ of the form $y_i^a$ are used 
for different clauses $C_{i^-}$ and $C_{i^+}$.
We fix any order $\leq''$ on the clauses such that $C_{i^+}$
is smaller than $C_{i^-}$ and an arbitrary order on the pairs.
Additionally, we set $A \prec B$ for all $A \in Cl_2$ and $B \in Cl_1$
if $A$ was introduced for a clause that is smaller equal $C_{i^+}$
with respect to $\leq''$.
Furthermore, we set $B \prec A$ for all $A \in Cl_2$ and $B \in Cl_1$
if $A$ was introduce for a clause that is larger than $C_{i^+}$
with respect to $\leq''$.
This is obviously a linear order and we have
$X_j^b \prec X_j^b \setminus \{y_{i^-}^{a^-}\}$
and $X_k^c \setminus \{y_{i^+}^{a^+}\} \prec X_k^c$ for $b,c \leq 2$.
Hence the constructed order on $Cl_1 \cup Cl_2$ satisfies dominance.

\paragraph{Ordering sets in $Cl_3$ with same extremum element and $\{z_i^a\}$:}
We recall that we call the elements $\max_j^a$ and $\min_j^a$ extremum elements.
We will first define an order between the Class 3 sets that have the 
same extremum element and $\{z_i^a\}$, i.e., that have been introduced to force 
a same preference. 
In the following, we write for a set $A$ that contains an extremum-element $mm$ 
(which is by construction unique) $A_S := \{x \in A \mid x <' mm\}$  
for the set of elements in $A$ that are smaller than $mm$
and $A_L := \{x \in A \mid mm <' x\}$  
for the set of elements in $A$ that are larger than $mm$.

We set $A \prec B$ for sets $A,B$
that both contain the same extremum element of the form $\Max_i^c$ if:
\begin{itemize}
\item $\max_{\leq'}(A_L \triangle B_L) \in B$,
\item $A_L = B_L$ and $\min_{\leq'}(A_S \triangle B_S) \in A$.
\end{itemize}
Here, $\triangle$ is the symmetric difference operator, i.e., 
$A \triangle B := (A \cup B) \setminus (A \cap B)$.
Intuitively, this order is a `nesting' of lexicographic orders.
We first check whether $A_L$ is smaller than $B_L$ according to the leximax
order and only if $A_L$ and $B_L$ are equal we compare $A_S$ and $B_S$ according
to the leximin order.
We claim that this order satisfies dominance and strict independence.
It satisfies strict independence because
for all sets $S,T$ by definition $S \cup \{x\} \triangle T \cup \{x\} = S \triangle T$
for any $x \not \in S \cup T$.
For dominance, assume $x <' \min_{<'}(A)$ and $\Max_i^c \in A,A \cup \{x\}$.
Then, $A_L = (A \cup \{x\})_L$ and 
$\min_{\leq'}(A_S \triangle (A \cup \{x\})_S) = x$.
Hence, $A\cup\{x\} \prec A$. The case $\max_{<'}(A) <' x$ is similar. 

Next, we add $\{z^c_i\}$ to the order.
We observe that we
have either 
$X_j^a\setminus \{y_i^b\} \cup \{\Max_i^c\}\prec \{z^c_i, \Max_i^c\}$
or 
$\{z^c_i, \Max_i^c\} \prec X_j^a\setminus \{y_i^b\} \cup \{\Max_i^c\}$.
In the first case, we add $\{z^c_i\}$ in the order exactly after $X_j^a\setminus \{y_i^b\}$,
i.e., we set $\{z^c_i\} \prec A$ if $X_j^a\setminus \{y_i^b\} \prec A$
and $A \prec \{z^c_i\}$ if $A \preceq X_j^a\setminus \{y_i^b\}$ for all $A \in \cX \setminus
\{z^c_i\}$.
In the second case we set $\{z_i^c\}$ exactly before $X_j^a\setminus \{y_i^c\}$.

Now, let $X_k^d \setminus \{y_i^e\}$ be the set for which
we enforce the preference
$X_j^a\setminus \{y_i^c\} \prec X_k^d\setminus \{y_i^e\}$.
Then, this implies $\{z^c_i\} \prec X_k^d\setminus \{y_i^e\}$,
because $X_j^a\setminus \{y_i^b\} \prec X_k^d\setminus \{y_i^e\}$
holds by construction and we added $\{z^c_i\}$ just before or just after $X_j^a\setminus \{y_i^b\}$.
Therefore, we have to make sure that 
$\{z^c_i, \Min_i^c\} \prec X_k^d\setminus \{y_i^e\} \cup \{\Min_i^c\}$
holds as intended by the construction to avoid a contradiction.
For this we use the fact that $v_1 <' v_2$ holds.
We set $A \prec B$ for elements $A,B$ 
if they both contain an element of the form $\Min_i^c$ if:
\begin{itemize}
\item $v_2 \in B$ and $v_2 \not \in A$, {\hfill$(\star)$}
\item $v_2 \in A,B$ or $v_2 \not \in A,B$ and $\max_{\leq'}(A_L \triangle B_L) \in B$,
\item $v_2 \in A,B$ or $v_2 \not \in A,B$, $A_L = B_L$ and $\min_{\leq'}(A_S \triangle B_S) \in A$.
\end{itemize}
It is clear that $(\star)$ implies
$\{z^c_i, \Min_i^c\} \prec X_k^d\setminus \{y_i^e\} \cup \{\Min_i^c\}$.
It is also clear that it satisfies strict independence
because the $(\star)$ implies a preference between sets $A \cup \{x\}$ and $B\cup\{x\}$
for $x \not \in A \cup B$ iff it implies the same preference for $A$ and $B$. 
If $(\star)$ is not applicable, strict independence is satisfied 
by the same argument as above.
Now, for dominance $v_2 \in (A \triangle (A\cup \{x\}))$ implies
$x = v_2$. Then, $x < \min_{<'} (A)$ is not possible because by construction
$v_1 \in A$ holds and we assume $v_1 <' v_2$. 
If we have $\max_{<'} (A) <' x$ then dominance is satisfied
because $A \prec A \cup \{x\}$ holds by $(\star)$.
If $x \neq v_2$, then $(\star)$ is not applicable and
dominance is satisfied by the same argument as above.

\paragraph{Interaction between $Cl_1 \cup Cl_2$ and $Cl_3$:}
The interaction between the Classes $Cl_1 \cup Cl_2$ and $Cl_3$
can be handled quite straightforwardly, as the only possible
application of dominance and strict independence are the ones
intended by the construction, i.e., when adding a extremum element
to force a preference. 

First, we observe that there is no set $A \in Cl_3$
such that $A \cup \{x\} \in Cl_1 \cup Cl_2$ holds,
as every set in $C_3$ either contains an extremum-element or it is a singleton
and no set in Class 1 and 2 contains an extremum-element 
and every set in Class 1 and 2 has more than three elements.

Furthermore, if $A \in Cl_1 \cup Cl_2$
and $A \cup \{x\} \in Cl_3$, then 
$A$ must be of the form $X_i^a \setminus \{y_j^b\}$ and
$x$ must be $\Max_j^c$ or $\Min_j^d$. Additionally, by construction,
for every $A \in Cl_1 \cup Cl_2$ there is at most one $x$ such that
$A \cup \{x\}$ holds.
Therefore, the only possible application of dominance 
involving sets from $Cl_1 \cup Cl_2$ and $Cl_3$ 
is adding an extremum element $\Max_j^c$ or $\Min_j^d$
to a set $A$ in $Cl_1 \cup Cl_2$ such that the extremum 
element is smaller (larger) than all elements in $A$.
We have to add preferences that satisfy this 
application of dominance. 
Consider $A \in Cl_1 \cup Cl_2$ and $B \in Cl_3^+ \cup Cl_3^-$
and let $mm_B$ be the unique extremum element in $B$.
Furthermore, let $i \leq$ and $a \in \{1,2,3\}$ be the unique
values for which $\{z_i^a,mm_B\}$ is in $\cX$.
Then, we extend $\preceq$ by
\begin{itemize}
\item $A \prec B$ if $z_i^a < mm_B$,
\item $B \prec A$ if $mm_B < z_i^a$.
\end{itemize}
Clearly, the resulting order is still a partial order 
and it satisfies all possible applications of dominance 
involving sets from $Cl_1 \cup Cl_2$ and $Cl_3$,
because $z_i^a \in A$ for all $i \leq m$, $a \in \{1,2,3\}$
and $A \in Cl_1 \cup Cl_2$.
Furthermore, observe that we do not add any 
new preferences between sets in $Cl_1 \cup Cl_2$,
but we do add, by transitivity, new comparisons between sets
in $Cl_3$ that contain different extremum elements.

Now, for strict independence, the only application
with sets from Class 3 and sets not from Class 3 
is lifting a preference between a set $A$ and $\{z_i^a\}$ to a preference
between $A \cup \{mm\}$ and $\{z_i^a, mm\}$ for some specific extremum element
$mm$. By construction $A$ must be of the form  $X_k^b\setminus \{y_i^a\}$
and hence strict independence is satisfied as we have seen in the paragraph above.

\paragraph{Ordering sets in $Cl_3$ with different (or no) extremum elements:}
Finally, we have to extend the order $\preceq$ to the whole Class 3.
First, observe that two sets with different extremum elements differ in at least two elements.
Hence, there is no possible application of dominance with sets containing different 
extremum elements.
The only possible application of strict independence that is not already satisfied
by $\preceq$ is adding the same element to two sets containing different 
extremum elements.
In order to make sure that these applications of strict independence are satisfied,
we order the elements with different extremum element with an order that
is only based on what extremum elements they contain.  
We fix an arbitrary linear order $\leq_{mm}$ on the extremum elements that is 
compatible with the preferences introduced in the paragraph above,
i.e., if $\{mm_1, z_i^a\}, \{mm_2,z_j^b\} \in \cX$, $mm_1 < z_i^a$ and
$z_j^b < mm$ hold then $mm_1 \leq_{mm} mm_2$ must hold.
We extend $\preceq$ by adding for all sets $A$ and $B$ 
that contain different extremum elements $mm_A$ and $mm_B$ respectively
\begin{itemize}
\item $A \prec B$ if $mm_A <_{mm} mm_B$,
\item $B \prec A$ if $mm_B <_{mm} mm_A$.
\end{itemize}

Clearly, this is compatible with the preference between sets
with different extremum elements added in the paragraph above.
Furthermore, we know for all $A,B$ that contain an extremum element that
for all $x \in X$ such that $A\cup\{x\}, B\cup\{x\} \in \cX$ holds that 
$A\cup\{x\}$ contains the same unique extremum element as $A$ and 
$B\cup\{x\}$ contains the same unique extremum element as $B$.
Hence, if $A$ and $B$ are Class 3 elements that contain different 
extremum elements, then $A \prec B$ implies $A \cup \{x\} \prec B \cup \{x\}$.

Furthermore, we observe that any Class 3 set that does not contain an 
extremum element must be of the form $\{z_i^a\}$ for some $z_i^a \in X$.
Now, if $\{z_i^a\} \cup \{x\}$ is in $\cX$ then $x$ must be
either $\Min_i^a$ or $\Max_i^a$. Furthermore, there is no other Class 3 set $A$
such that $A \cup \{x\}$ is in $\cX$. Therefore, no application of strict
independence involving $\{z_i^a\}$ is possible.
Finally, as $\{z_i^a\}$ is placed next to a Class 2 set $X_j^b \setminus \{y_i^c\}$,
$\{mm, z_i^a\} \in \cX$ and $mm < z_i^a$ implies $\{mm, z_i^a\} \prec \{z_i^a\}$
by the construction in the paragraph above.
Similarly, $\{mm, z_i^a\} \in \cX$ and $z_i^a <mm$ implies $\{z_i^a\} \prec \{z_i^a, mm\}$.
Therefore, any application of dominance involving $\{z_i^a\}$ is satisfied.

\paragraph{Completion of $\preceq$:}
Finally, because $\preceq$  is a partial order there is a linear order
that is a completion of $\preceq$. 
As we have seen above, $\preceq$ satisfies all possible 
applications of strict independence and dominance. 
Thus, any completion of $\preceq$ satisfies dominance 
and strict independence. 
\end{proof}

We have shown that there is a linear order on 
$\cX$ that satisfies dominance and strict independence with respect to a fixed order $\leq$ if and only if $\phi$ is satisfiable. 
Furthermore, we have shown how to use a satisfying assignment of $\phi$
to construct for any linear order $\leq'$ on $X$ 
a linear order $\preceq$ on $\cX$ that satisfies dominance and strict independence.
This construction clearly also works for the fixed order $\leq$.
In other words $\phi$ is satisfiable if and only if 
$(X,\cX,\leq)$ is a positive instance of \textsc{$DI^{S}$-LO-Orderability}.
Hence, the reduction above also shows that \textsc{$DI^{S}$-LO-Orderability}
is \NP-complete. \begin{center}
\begin{mdframed}[style=mystyle]
\begin{minipage}{\textwidth}

\medskip
\begin{Cly}\label{Cly:StrictNPcomp}
\textsc{$DI^S$-LO-Orderability} is \NP-complete.
\end{Cly}
\end{minipage}
\end{mdframed}
\end{center}

\begin{proof}
It follows from the reduction above that \textsc{$DI^{S}$-LO-Orderability}
is \NP-hard. Furthermore, it is clear that \textsc{$DI^{S}$-LO-Orderability} is in \NP:
We can guess a linear order $\preceq$ on $\cX$ and then check in polynomial time
if $\preceq$ satisfies dominance and strict independence with respect to $\leq$.
\end{proof}

This result was already shown by \citet{Maly2017}.
However, the new reduction presented here will be useful to prove 
further results.
A variation of this reduction shows that
we can additionally add extension to the required axioms 
without changing the complexity of the problem.
Essentially, we modify the family $\cX$ such that 
it does not contain any singletons.
This can be achieved by replacing 
singletons with two element sets.

\begin{center}
\begin{mdframed}[style=mystyle]
\begin{minipage}{\textwidth}

\medskip
\begin{Cly}\label{Cly:StrictExt}
\textsc{Strong $DI^SE$-LO-Orderability} is \NP-hard.

\noindent
\textsc{$DI^SE$-LO-Orderability} is \NP-complete.
\end{Cly}
\end{minipage}
\end{mdframed}
\end{center}

\begin{proof}
First, observe that \textsc{$DI^SE$-LO-Orderability} is in \NP{} for the same reason that 
\textsc{$DI^S$-LO-Orderability} is, because it can also be checked in 
polynomial time if a linear order satisfies the extension rule.

We modify the reduction above slightly
to show that \textsc{Strong $DI^SE$-LO-Orderability} and
\textsc{$DI^SE$-LO-Orderability} are
 \NP-hard.
In general, the order on the singletons does not satisfy the extension rule.
One way to solve this problem is replacing all singletons
by two element sets.
All singletons that appear in the reduction are of the form
$\{z_k^a\}$, $\{\min_i^b\}$ or $\{max_j^c\}$ for some
$z_k^a, \min_i^b, \max_j^c \in X$.
For all $z_i^a \in X$, replace in the reduction above every mentioning of $z_i^a$
by $z_i^{a,1}, z_i^{a,2}$ and similarly replace for all $\min_i^b, \max_j^c \in X$ 
every mentioning of $\min_i^b$ and
$\max_j^c$ by $\overline{\min_i^{b}}, \min_i^{b}$ and
$\overline{\max_i^{c}}, \max_i^{c}$, respectively.
In the critical linear order,
we replace $z_i^a$ by $z_i^{a,1} < z_i^{a,2}$,
$\min_i^b$ by $\overline{\min_i^{b,*}} < \min_i^{b}$ and $\max_j^c$ by
$\max_i^{c} <\overline{\max_i^{c}}$.
Finally, we add the additional set 
$(X_j^a\setminus \{y_i^1, z_i^{a,1}\}) \cup \{\Max_i^1, \overline{\Max_i^1}\}$,
$\{z_i^{a,2}, \Max_i^1, \overline{\Max_i^1}\}$,
$(X_j^a\setminus \{y_i^1\}) \cup \{\Max_i^1\}$ and 
$\{z_i^{a,1}, z_i^{a,2}, \Max_i^1\}$.

Then \[(X_j^a\setminus \{y_i^1, z_i^{a,1},z_i^{a,2}\}) \cup \{\Max_i^1, \overline{\Max_i^1}\}\prec
\{\Max_i^1, \overline{\Max_i^1}\}\]
implies by two applications of strict independence
\[(X_j^a\setminus \{y_i^1\}) \cup \{\Max_i^1, \overline{\Max_i^1}\}\prec
\{z_i^{a,1},z_i^{a,2},\Max_i^1, \overline{\Max_i^1}\}.\]
This, in turn, implies by two applications of reverse strict independence
\[(X_j^a\setminus \{y_i^1\}) \prec
\{z_i^{a,1},z_i^{a,2}\}.\]

When constructing a linear order $\preceq$ based on a 
satisfying assignment, we can treat the four new sets 
$(X_j^a\setminus \{y_i^1, z_i^{a,1}\}) \cup \{\Max_i^1, \overline{\Max_i^1}\}$,
$\{z_i^{a,2}, \Max_i^1, \overline{\Max_i^1}\}$,
$(X_j^a\setminus \{y_i^1\}) \cup \{\Max_i^1\}$ and 
$\{z_i^{a,1}, z_i^{a,2}, \Max_i^1\}$
the same as any other set containing $\Max_i^1$.
It can be checked that the rest of the construction works as before. 
Now, $\cX$ does not contain any singletons.
Therefore, the linear order constructed in the second part of the reduction
satisfies the extension rule vacuously.
\end{proof}

\subsubsection*{$DI^S$-WO-Orderability}

Now, we relax the requirement that the lifted order needs to be antisymmetric.
As it turns out, this does not change the complexity of deciding whether 
dominance and strict independence are compatible.
To show this, we need to modify the reduction given above.
Again, we first state the result for strong orderability,
even though we will improve that result in the next section (Theorem~\ref{PI1}).
The \NP-completeness of the $\leq$-orderability version will follow 
directly.

\begin{Prop}\label{Prop:StronpNP}
\textsc{strong $DI^S$-WO-Orderability} is \NP-hard.
\end{Prop}

\begin{proof}
We will modify the reduction used for Proposition~\ref{DI^S-NP-hard}.
We have to compensate
for the fact that the lifted order $\preceq$ does not need to be antisymmetric. 
First, we have to make sure that all preferences between
sets $X^\mathrm{f}_i$ and $X_i^\mathrm{t}$ are strict. 
Otherwise, $\preceq$ would not encode a valid truth assignment.
We can add elements and sets that lead to a contradiction
if the preference is not strict.
This can be done as follows:
Assume we want to enforce that the preferences between two sets 
$X^\mathrm{f}_i$ and $X_i^\mathrm{t}$ is strict. Then we add sets $A$, $B$, $C$ and $D$
and enforce preferences such that $X_i^\mathrm{f} \preceq X^\mathrm{t}_i$
implies $A \prec B$ by transitivity and $X_i^\mathrm{t} \preceq X^\mathrm{f}_i$
implies $D \prec C$ by transitivity. Furthermore, we add sets such that 
$A \prec B$ and $D \prec C$ lead to a contradiction if they hold at the same time.
Then, $X_i^\mathrm{f} \preceq X^\mathrm{t}_i$ and $X_i^\mathrm{t} \preceq X^\mathrm{f}_i$
cannot hold at the same time. Hence, either  $X_i^\mathrm{f} \prec X^\mathrm{t}_i$ or
$X_i^\mathrm{t} \prec X^\mathrm{f}_i$ must hold.

Furthermore, if $\preceq$ is not linear, then it does not need to satisfy 
reverse strict independence even if it satisfies strict independence. 
Reverse strict independence was only used once on the reduction
for Proposition~\ref{DI^S-NP-hard}, namely to infer that
$(X_j^a\setminus \{y_i^1\}) \cup \{\Max_i^1\}\prec \{z_i^1, \Max_i^1\} $
implies
$X_j^a\setminus \{y_i^1\} \prec \{z_i^1\}$.
Therefore, we need to adapt the way that we enforce preferences 
to a method that does not require reverse strict independence.
This can be done as follows
for a desired preference 
$X_i^a \setminus \{y_j^b\} \prec X_i^c \setminus \{y_j^b\}$:
We replace every element $z_j^b$ by two elements 
$z_j^b$ and $\overline{z_j^b}$,
set $z_j^b < \overline{z_j^b}$ and add the sets
$\{z_j^b\}, \{z_j^b, \overline{z_j^b}\}, \{\overline{z_j^b}\}$ to $\cX$.
By dominance this enforces $\{z_j^b\} \prec \{\overline{z_j^b}\}$.
Now, using a similar construction as before, we can enforce 
$X_i^a \setminus \{y_j^b\} \preceq \{z_j^b\}$
and $\{\overline{z_j^b}\} \preceq X_k^c \setminus \{y_j^b\}$.
Then, we have
\[X_i^a \setminus \{y_j^b\} \preceq \{z_j^b\} \prec \{\overline{z_j^b}\}
\preceq X_k^c \setminus \{y_j^b\}\]
and hence by transitivity the desired strict preference  
$X_i^a \setminus \{y_j^b\} \prec X_i^c \setminus \{y_j^b\}$.
To enforce $X_i^a \setminus \{y_j^b\} \preceq \{z_j^b\}$
we add the same sequence $A_1, \dots, A_l$ as in the proof of Proposition~\ref{DI^S-NP-hard}.
This enforces as before
\[(X_i^a\setminus \{y_j^b\}) \cup \{\Max_j^b\} \prec \{z_j^b, \Max_j^b\}.\]
Now, as $\preceq$ is total, by Proposition~\ref{Prop:Reverse}
it satisfies reverse independence.
Therefore, we get the desired
$X_i^a \setminus \{y_j^b\} \preceq \{z_j^b\}$.
We can enforce $\{\overline{z_j^b}\} \preceq X_k^c \setminus \{y_j^d\}$
similarly.

Next, we want to enforce that all preferences between
sets $X^\mathrm{f}_i$ and $X_i^\mathrm{t}$ are strict using
the idea described at the beginning of the proof. 
We will add additional sets that lead to cyclic preferences
if $X^\mathrm{f}_i \preceq X_i^\mathrm{t}$
and $X^\mathrm{t}_i \preceq X_i^\mathrm{f}$ hold at the same time.
The idea is illustrated in Figure~\ref{fig:strict}.
We add for every variable $X_i$ new elements
\[a_i^-, b_i^-, c_i^-, d_i^-, r_i, s_i, d_i^+, c_i^+, b_i^+, a_i^+.\]
Furthermore, we add them to the critical linear order $\leq$ as follows
\begin{multline*}
\Min_1^1 < \Min_1^2 < \dots < \Min_m^3 < x^-_{1,1} < x^-_{1,2} <\dots
 < x^-_{n,2} <\\
\mathbf{a_i^- < b_i^- <c_i^- < d_i^-}
 <v_1 < v_2 
\mathbf{ < r_i < s_i < d_i^+ < c_i^+ < b_i^+ < a_i^+}\\
< z_1^1 < z_1^2 \dots < z_m^3 < y_1^1 < y_1^2 < \dots
  < y_m^3 < \\ x^+_{1,1} < x^+_{1,2} < \dots 
< x^+_{n,2} < \Max_1^1 < \Max_1^2 < \dots < \Max_m^3
\end{multline*}
Then, we add new sets 
\begin{multline*}
A_i := \{a_i^-,v_1,v_2, r_i,s_i,a_i^+\}, B_i := \{b_i^-,v_1,v_2, r_i,s_i,b_i^+\},\\
C_i := \{c_i^-,v_1,v_2, r_i,s_i,c_i^+\} \text{ and } D_i := \{d_i^-,v_1,v_2, r_i,s_i,d_i^+\}.
\end{multline*}
Now, let  
$z_i^a$, $\overline{z_i^a}$, $\Max_i^a$ and $\Min_i^a$ be new elements
where we set $z_i^a,\overline{z_i^a} \in Y$.
Then, we enforce with the method described above $A_i \prec X_i^\mathrm{f}$ using these new elements.
Furthermore, we enforce $X_i^\mathrm{t} \prec B_i$, $X_i^\mathrm{f} \prec C_i$ and $D_i \prec X_i^\mathrm{t}$.
Finally, we add the sets $A_i \setminus \{r_i\}$, 
$B_i \setminus \{r_i\}$, $C_i \setminus \{s_i\}$ and $D_i \setminus \{s_i\}$
and enforce $B_i \setminus\{r_i\} \prec D_i \setminus \{s_i\}$ and
$C_i \setminus \{s_i\} \prec A_i \setminus \{r_i\}$.
We call the sets added in this step the Class 4 sets.
These enforced preference are shown as solid arrows in Figure~\ref{fig:strict}.

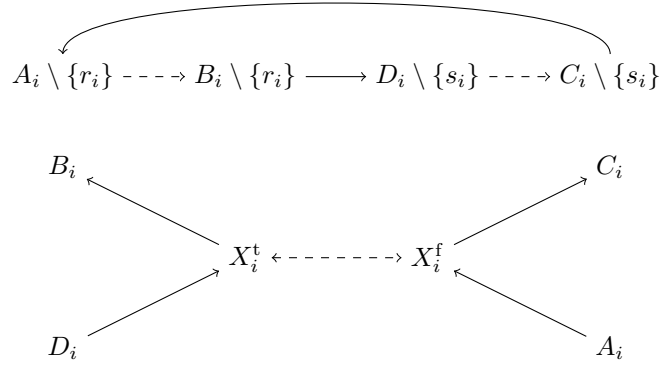
\begin{figure}
\begin{center}
\begin{tikzpicture}[scale=1.2]
\node (a) at  (2,2) {$X_i^\mathrm{t}$} ;
\node (d) at  (0,3) {$B_i$} ;
\node (f) at  (0,1) {$D_i$} ;
\node (b) at (4,2) {$X_i^\mathrm{f}$};
\node (c) at (6,1) {$A_i$};
\node (e) at (6,3) {$C_i$};
\node (cx) at (0,4) {$A_i\setminus \{r_i\}$};
\node (dx) at (2,4) {$B_i \setminus \{r_i\}$};
\node (ey) at (4,4) {$D_i\setminus \{s_i\}$};
\node (fy) at (6,4) {$C_i\setminus \{s_i\}$};
\draw[->] (dx) -- (ey) ;
\draw[<-] (b) -- (c);
\draw[<-] (d) -- (a);
\draw[<-] (a) -- (f);
\draw[<-] (e) -- (b);
\draw[dashed, <->] (a) -- (b);
\draw[dashed, ->] (cx) -- (dx);
\draw[dashed, ->] (ey) -- (fy);
\draw[<-] (cx)  .. controls(0,5) and (6,5) .. (fy);
\end{tikzpicture}
\end{center}
\caption{Enforcing strictness.}\label{fig:strict} 
\vspace{1cm}
\end{figure} 
 
Now, we claim that it is not possible for a weak order $\preceq$
to satisfy dominance and strict independence with respect to $\leq$
if $X_i^\mathrm{t} \sim X_i^\mathrm{f}$ holds.
Assume otherwise that $\preceq$ is a weak order that satisfies
dominance and strict independence with respect to $\leq$ such that
$X_i^\mathrm{t} \sim X_i^\mathrm{f}$ holds.
Then, $D_i \prec X_i^\mathrm{t} \preceq X_i^\mathrm{f} \prec C_i$ implies 
$D_i \prec C_i$ by transitivity and hence $D_i \setminus \{s_i\} \preceq C_i \setminus \{s_i\}$
by reverse independence.
Similarly, $A_i \prec X_i^\mathrm{f} \preceq X_i^\mathrm{t} \prec B_i$ implies 
$A_i \prec B_i$ by transitivity and hence $A_i \setminus \{r_i\} \preceq B_i \setminus \{r_i\}$
by reverse independence.
However, this leads to a contradiction by 
\[A_i \setminus \{r_i\} \preceq B_i \setminus \{r_i\} \prec
 D_i \setminus \{s_i\} \preceq C_i \setminus \{s_i\} \prec A_i \setminus \{r_i\}.\]
Now, as before
$\cX$ can only be $DI^S$-orderable with respect to $\leq$ 
if $\phi$ is a positive instance of \textsc{Sat}.

It remains to show that the modified family $\cX$
is strongly $DI^S$-orderable if $\phi$ is a positive instance 
of \textsc{Sat}.
As before, let $\leq'$ be an arbitrary linear order on $X$.
If $\phi$ is a positive instance of \textsc{Sat}, 
we can construct a weak order on $\cX$ 
that satisfies dominance and strict independence as before,
with the following modifications:
We replace $z_i^e$ by the block $\{z_i^e\} \prec \{z_i^e,\overline{z_i^e}\} \prec \{\overline{z_i^e}\}$ 
if $z_i^e < \overline{z_i^e}$ and 
by the block $\{\overline{z_i^e}\} \prec \{z_i^e,\overline{z_i^e}\} \prec \{z_i^e\}$ 
if $\overline{z_i^e} < z_i^e$.

It remains to add the Class 4 sets to the order.
The Class 4 sets used to enforce the preferences
$A_i \prec X_i^\mathrm{f}$,  $X_i^\mathrm{t} \prec B_i$, $X_i^\mathrm{f} \prec C_i$ and $D_i \prec X_i^\mathrm{t}$
can be ordered the same way 
as the Class 3 sets. As before, we use the fact that we can assume $v_1 <' v_2$
to ensure that this order is compatible with the enforced preferences.
For a specific variable $V_i$ we set by construction
either $X_i^\mathrm{t} \prec X_i^\mathrm{f}$ or $X_i^\mathrm{f} \prec X_i^\mathrm{t}$.
We assume $X_i^\mathrm{t} \prec X_i^\mathrm{f}$. The other case is symmetric.
Then, we add $D_i$ in $\preceq$ exactly before $X_i^\mathrm{t}$ and $B_i$
exactly after $X_i^\mathrm{t}$. Similarly, we add $A_i$
exactly before $X_i^\mathrm{f}$ and $C_i$
exactly after $X_i^\mathrm{f}$.
Then, we have 
\[D_i \prec X_i^\mathrm{t} \prec B_i \prec A_i \prec X_i^\mathrm{f} \prec C_i\]
which is compatible with the forced preferences.

Now, consider the group 
$A_i \setminus \{r_i\}, B_i \setminus \{r_i\}, C_i \setminus \{s_i\}$
and $D_i \setminus \{s_i\}$.
We observe that all sets in this group differ in at least two elements.
Therefore, we only have to set 
$B_i\setminus \{r_i\} \preceq A_i\setminus \{r_i\}$ and 
$D_i\setminus \{s_i\} \preceq C_i\setminus \{s_i\}$ 
in order to satisfies reverse independence.
Furthermore, we have to satisfy dominance
if $r_i$ and/or $s_i$ are the largest element
of the set they are removed from.
This can be satisfied by a straightforward construction
that respects the enforced preferences
unless $r_i$ is the maximal element of $A_i$
and $s_i$ is the minimal element of $C_i$
or alternatively if
$r_i$ is the maximal element of $B_i$
and $s_i$ is the minimal element of $D_i$.
We describe the construction for the first case:
We have to set 
$A_i \setminus \{r_i\} \prec A_i$ and 
$C_i \prec C_i \setminus \{s_i\}$
which implies 
$A_i \setminus \{r_i\} \prec C_i \setminus \{s_i\}$
contrary to the preference we wanted to enforce in the construction.
We use the fact that $r_i$ is the maximal element of $A_i$
and $s_i$ is the minimal element of $C_i$ to define an order that allows this.
Let $z_i^{r,a}$ and $\max_i^{r,a}$ be the new elements used to enforce 
$\{z_i^{r,a}\} \prec A_i \setminus \{r_i\}$.
Then, we use the same order as for the other Class 3 sets, i.e.,
we set $A \preceq B$ for the sets $A,B$ such that $\max_i^{r,a}\in A,B$ if
\begin{itemize}
\item $v_2 \in A$ and $v_2 \not \in B$
\item $v_2 \in A,B$ or $v_2 \not \in A,B$ and $\max_{\leq'}(A_L \triangle B_L) \in B$,
\item $v_2 \in A,B$ or $v_2 \not \in A,B$, $A_L = B_L$ and $\min_{\leq'}(A_S \triangle B_S) \in A$.
\end{itemize}
where $A_L := \{x \in A \mid \max_i^{r,a} <' x\}$
and $A_S := \{x \in A \mid x <' \max_i^{r,a}\}$.
This order satisfies dominance and strict independence because
the element $s_i$ is smaller than $v_2$ by assumption 
and removed later in the sequence $A_1, \dots ,A_k$.
Furthermore, this implies
$(A_i \setminus \{r_i\}) \cup \{\max_i^{r,a}\} \prec \{z_i^{r,a}, \max_i^{r,a}\}$,
which implies 
$A_i \setminus \{r_i\} \prec \{z_i^{r,a}\}$.
This allows us to set 
$A_i \setminus \{r_i\} \prec C_i \setminus \{s_i\}$.
Then, we can place 
$A_i \setminus \{r_i\}$ just before $A_i$ and 
$C_i \setminus \{s_i\}$ just after $C_i$ to get an order that satisfies
dominance and independence.
\end{proof}

As in the case of \textsc{Strong $DI^{S}E$-LO-Orderability}
we have shown that there is a weak order on 
$\cX$ that satisfies dominance and strict independence with respect to a fixed order $\leq$
if only if $\phi$ is satisfiable. 
Therefore, this reduction again also works for \textsc{$DI^{S}$-WO-Orderability}.

\begin{center}
\begin{mdframed}[style=mystyle]
\begin{minipage}{\textwidth}

\medskip
\begin{Thm}\label{DI^S-NP-comp}
\textsc{$DI^S$-WO-Orderability} is \NP-complete.
\end{Thm}
\end{minipage}
\end{mdframed}
\end{center}

\begin{proof}
It follows from the reduction above that \textsc{$DI^{S}$-WO-Orderability}
is \NP-hard. Furthermore, it is clear that \textsc{$DI^{S}$-WO-Orderability} is in \NP{}
as we can guess a weak order $\preceq$ on $\cX$ and then check in polynomial time
if $\preceq$ satisfies dominance and strict independence with respect to $\leq$.
\end{proof}

Moreover, we can adapt the reduction to the case that extension is additionally required
in a similar way as in the proof of Corollary~\ref{Cly:StrictExt}.

\begin{center}
\begin{mdframed}[style=mystyle]
\begin{minipage}{\textwidth}

\medskip
\begin{Cly}\label{Cly:DI^SE-NP-comp}
\textsc{$DI^SE$-WO-Orderability} is \NP-complete.
\end{Cly}
\end{minipage}
\end{mdframed}
\end{center}

\begin{proof}
We can adapt the reduction in a similar way as 
in the proof of Corollary~\ref{Cly:StrictExt}.
We need to change two things compared to Corollary~\ref{Cly:StrictExt}.
First of all, in contrast to reverse strict independence we cannot iterate 
reverse independence. Therefore, we need to adapt the way that we enforce preferences.
We add additionally the sets
$(X_i^a\setminus \{y_j^b, x_i^-\}) \cup \{\Max_j^b\}$.
and
$(X_i^a\setminus \{y_j^b, x_i^-\}) \cup \{\Max_j^b, \overline{\Max_i^b}\}$.
We observe that for the sequence $A_1, \dots, A_l$ added to enforce the preference
$(X_j^a\setminus \{y_i^b\}) \cup \{\Max_i^c\}\prec \{z_i, \Max_i^d\}$,
we have $l > 3$ and that the following preference is enforced by dominance
\[A_2 = (X_i^a\setminus \{y_j^b,x_i^-,z_i^{a,1},z_i^{a,2}\}) \cup \{\Max_j^b, \overline{\Max_i^b}\}
\prec \{\Max_j^b, \overline{\Max_i^b}\} = A_l\]
this enforces by strict independence
\[(X_i^a\setminus \{y_j^b,x_i^-\}) \cup \{\Max_j^b, \overline{\Max_i^b}\}
\prec \{z_i^{a,1},z_i^{a,2}, \Max_j^b, \overline{\Max_i^b}\}.\]
Now, by one application of reverse independence, we have
\[(X_i^a\setminus \{y_j^b,x_i^-\}) \cup \{\Max_j^b,\}
\preceq \{z_i^{a,1},z_i^{a,2}, \Max_j^b\}.\]
and hence by dominance
\[(X_i^a\setminus \{y_j^b\}) \cup \{\Max_j^b\} \prec \{z_i^{a,1},z_i^{a,2}, \Max_j^b\}. \]
Then, another application of reverse independence implies
\[(X_j^a\setminus \{y_i^1\}) \preceq
\{z_i^{a,1},z_i^{a,2}\}.\]

Furthermore, the new singleton in the reduction $\overline{z_i^{e}}$
also needs to be dealt with.
We replace $z_i^e$ again by two elements
$z_i^{e,1}$ and $z_i^{e,2}$. Similarly
we replace $\overline{z_i^{e}}$ by
$\overline{z_i^{e,1}}$ and $\overline{z_i^{e,2}}$.
In the critical linear order we set 
\[z_i^{e,1} < z_i^{e,2} < \overline{z_i^{e,1}} <\overline{z_i^{e,2}}.\]
Then, we use these to enforce as in Corollary~\ref{Cly:StrictExt}
\[X_i^a \setminus \{y_j^b\} \preceq \{z_j^{b,1}, z_j^{b,2}\}
\mbox{ and } \{\overline{z_j^{b,1}},\overline{z_j^{b,2}}\} \preceq X_k^c \setminus \{y_j^d\}.\]
Further, instead of $\{z_i^e, \overline{z_i^e}\}$ we add 
$\{z_i^{e,1}, z_i^{e,2}, \overline{z_i^{e,1}}\}$,
$\{z_i^{e,2}, \overline{z_i^{e,1}}\}$ and $\{z_i^{e,2},\overline{z_i^{e,1}}, \overline{z_i^{e,2}}\}$.
Then, the following preferences are enforced:
\[\{z_i^{e,1}, z_i^{e,2}\} \prec \{z_i^{e,1}, z_i^{e,2}, \overline{z_i^{e,1}}\} \prec
 \{z_i^{e,2}, \overline{z_i^{e,1}}\}\prec
\{z_i^{e,2},\overline{z_i^{e,1}}, \overline{z_i^{e,2}}\} \prec
\{\overline{z_i^{e,1}}, \overline{z_i^{e,2}}\}.\]
Therefore, by transitivity we have 
\[X_i^a \setminus \{y_j^b\} \preceq \{z_j^{b,1}, z_j^{b,2}\}
\prec \{\overline{z_j^{b,1}},\overline{z_j^{b,2}}\} \preceq X_k^c \setminus \{y_j^d\}.\]
and therefore $X_i^a \setminus \{y_j^b\} \prec X_k^c \setminus \{y_j^d\}$
as intended.

When constructing an order on $\cX$ from an satisfying assignment of $\phi$
we can just replace the block 
$\{z_i^e\} \prec \{z_i^e,\overline{z_i^e}\} \prec \{\overline{z_i^e}\}$ 
resp.\ $\{\overline{z_i^e}\} \prec \{z_i^e,\overline{z_i^e}\} \prec \{z_i^e\}$ 
by the following sets 
$\{z_i^{e,1}, z_i^{e,2}\}, \{z_i^{e,1}, z_i^{e,2}, \overline{z_i^{e,1}}\},
 \{z_i^{e,2}, \overline{z_i^{e,1}}\},
\{\overline{z_i^{e,1}}, \overline{z_i^{e,2}}\},
\{z_i^{e,2},\overline{z_i^{e,1}}, \overline{z_i^{e,2}}\}$
ordered as demanded by dominance. It can be checked that the rest of the construction
works as before, as all other new sets contain an extremum element.
\end{proof}

\subsubsection*{Strong $DI^S$-WO-Orderability and Strong $DI^SE$-WO-Orderability}
\subsubsectionmark{Strong $DI^S$-WO-Orderability}

We have shown before that strong $DI^S$-orderability is \NP-hard.
However, we claim that the problem is even $\Pi_2^P$-hard,
and furthermore also $\Pi_2^p$-complete.
To show this we will extend the reduction used in the previous results 
from \Sat{} to $\Pi_2$-\Sat.

\begin{center}
\begin{mdframed}[style=mystyle]
\begin{minipage}{\textwidth}

\medskip
\begin{Thm}\label{PI1}
\textsc{Strong $DI^S$-WO-Orderability} is $\Pi_2^p$-complete.
\end{Thm}
\end{minipage}
\end{mdframed}
\end{center}

\begin{proof}
$\Pi_2^p$-membership holds, because there is a non-deterministic Turing machine
with access to an NP-oracle such that an instance $(X, \cX)$ is a yes-instance
if and only if all runs of the machine end positively.
This machine uses the non-determinism to guess a linear order $\leq$ on $X$
and then checks via the \NP-oracle if $(X,\cX, \leq)$ is a positive instance of 
\textsc{$DI^S$-WO-Orderability}, which is possible due to Theorem~\ref{DI^S-NP-comp}. 

It remains to show
that \textsc{Strong $DI^S$-WO-Orderability} is $\Pi_2^p$-hard. 
We do this by extending the reduction above to a reduction
from a \textsc{$\Pi_2$-Sat} instance $\phi = \forall \vec{W} \exists \vec{V} \psi(\vec{W},\vec{V})$.
Recall that in the previous reductions we added for each 
variable $v_i$ two sets $X_i^\mathrm{t}$ and $X_i^\mathrm{f}$ such that
the preference between these sets encoded the truth assignment 
of the variable $v_i$. Then, we ensured that a non-satisfying
assignment lead to a contradiction under a specific critical 
linear order $\leq^*$.  
In this reduction, we also have to take care of the 
universally quantified variables $w_1 \dots w_l$.
Intuitively, we want to add for every $w_i$ new elements 
$w_i^\mathrm{t}$ and $w_i^\mathrm{f}$ such that the preference
between these elements determines the preference
between $X_i^\mathrm{t}$ and $X_i^\mathrm{f}$, i.e., the sets 
encoding the truth value of $w_i$.
Then, there is a class of critical linear orders 
such that every truth assignment to the universally 
quantified variables is encoded by at least one critical linear order.
Furthermore, we will ensure that there exists an order satisfying
dominance and strict independence with respect to a critical linear 
order $\leq$ if and only if there exists a satisfying assignment 
to $\psi$ that extends the assignment encoded by $\leq$.

We set up the reduction similarly to the 
one for Proposition~\ref{Prop:StronpNP}.
Additionally, we add for every universally quantified variable $w_i$ represented by $X_i^\mathrm{t}$ and $X_i^\mathrm{f}$
sets \[X_i^\mathrm{t} \setminus \{y_i^q\}, X_i^\mathrm{f} \setminus \{y_i^q\}, \{w_i^\mathrm{t}\}, \{w_i^\mathrm{f}\}
\text{ and } \{w_i^\mathrm{t},w_i^\mathrm{f}\},\]
where $y_i^q,w_i^\mathrm{t}$ and $w_i^\mathrm{f}$ are new elements.
Then, we enforce -- with the same method as before --
$X_i^\mathrm{t} \setminus \{y_i^q\} \prec \{w_i^\mathrm{t}\}$ and $\{w_i^\mathrm{f}\} \prec X_i^\mathrm{f} \setminus \{y_i^q\}$
using new elements $\min_i^q$ and $\max_i^q$.
This will ensure for every critical linear order $\leq'$ with $w_i^\mathrm{t} <' w_i^\mathrm{f}$ 
that $X_i^\mathrm{t} \prec X_i^\mathrm{f}$ must hold for
every order $\preceq$ on $\cX$ that satisfies dominance and
strict independence with respect to $\leq'$. 
We add additionally
sets $X_i^1 \setminus \{\overline{y_i^q}\}$ and $X_i^2 \setminus \{\overline{y_i^q}\}$
and similarly enforce   
$X_i^\mathrm{t} \setminus \{\overline{y_i^q}\} \prec \{w_i^\mathrm{t}\}$ and
$\{w_i^\mathrm{f}\} \prec X_i^\mathrm{f} \setminus \{\overline{y_i^q}\}$
using $\overline{\min_i^q}$ and $\overline{\max_i^q}$.
Hence, $X_i^\mathrm{t} \prec X_i^\mathrm{f}$
must hold for every order $\preceq$ on $\cX$
that satisfies dominance and strict independence with respect to any critical linear order $\leq''$ on $X$
with $w_i^\mathrm{t} <'' w_i^\mathrm{f}$.

Now, we claim that $(X,\cX)$ can only be a positive instance of \textsc{Strong $DI^S$-WO-Orderability}
if $\phi$ is satisfiable.
We call a linear order critical if: 
\begin{itemize}
\item It equals the critical linear order $\leq$ in the proof of Proposition~\ref{Prop:StronpNP}
on the elements already occurring in that reduction.
\item The  new elements $y_i^q,w_i^\mathrm{t}, w_i^\mathrm{f},\min_i^q$ and $\max_i^q$ are added
in that order as follows: $\min_i^q$ is smaller than $\min_1^1$ for all $i$.
Similarly, $\max_i^q$ is larger than $\max_m^3$ for all $i$. The other elements are added between
$v_2$ and $z_1^1$.
\end{itemize}
Then, for every truth assignment $T$ to the variables in $\vec{W}$ there is a critical linear order
$\leq^*$ on $X$ such that $w_i^\mathrm{t} <^* w_i^\mathrm{f}$ if $w_i$ is assigned false in $T$ and
$w_i^\mathrm{f} <^* w_i^\mathrm{t}$ if $w_i$ is assigned true in $T$.
Now, if there is no satisfying assignment to $\phi$ that extends $T$,
then there can be no order on $\cX$ satisfying
dominance and strict independence with respect to $\leq^*$.
Hence $(X,\cX)$ can only be $DI^S$-orderable with respect to every linear order $\leq^*$
if $\phi$ is satisfiable.

It remains to show that 
if $\phi$ is satisfiable
then $(X,\cX)$ is a positive instance of \textsc{Strong $DI^S$-WO-Orderability}.
This can be done using nearly the same construction as above
treating $X_i^\mathrm{t} \setminus \{y_i^q\}$ and $X_i^\mathrm{f} \setminus \{y_i^q\}$
as Class 2 sets, all other new sets as Class 3 sets
and inserting $\{w_i^\mathrm{t}\} \prec \{w_i^\mathrm{t},w_i^\mathrm{f}\} \prec \{w_i^\mathrm{f}\}$ 
resp.\ $\{w_i^\mathrm{f}\} \prec \{w_i^\mathrm{t},w_i^\mathrm{f}\} \prec \{w_i^\mathrm{t}\}$ 
where we would insert $z_i^j$ and $\overline{z_i^j}$.
The only exception has to be made if there is an $i$
such that $y_i^q = \min (X_i^\mathrm{t})$ and $\overline{y_i^q} = \max (X_i^\mathrm{t})$
or $\overline{y_i^q} = \min (X_i^\mathrm{f})$ and $y_i^q = \max (X_i^\mathrm{f})$.
In the first case, we set $A \prec B$ for
the sets containing $\overline{\min_i^q}$ if:
\begin{itemize}
\item $y_i^q \in A$ and $y_i^q \not \in B$
\item $y_i^q \in A,B$ or $y_i^q \not \in A,B$ and $\max(A_L \triangle B_L) \in B$,
\item $y_i^q \in A,B$ or $y_i^q \not \in A,B$, $A_L = B_L$ and $\min(A_S \triangle B_S) \in A$.
\end{itemize}
where $A_L := \{x \in A \mid \overline{\min_i^q} <' x\}$
and $A_S := \{x \in A \mid x <'\overline{\min_i^q}\}$.
And for $A \prec B$ for the sets containing 
$\max_i^q$ if
\begin{itemize}
\item $\overline{y_i^q} \in B$ and $\overline{y_i^q} \not \in A$
\item $y_i^q \in A,B$ or $y_i^q \not \in A,B$ and $\max(A_L \triangle B_L) \in B$,
\item $y_i^q \in A,B$ or $y_i^q \not \in A,B$, $A_L = B_L$ and $\min(A_S \triangle B_S) \in A$.
\end{itemize}
where $A_L := \{x \in A \mid \max_i^q <' x\}$
and $A_S := \{x \in A \mid x <'\max_i^q\}$.

It is clear that these orders satisfy dominance and strict independence,
similarly to the orders on the Class 3 sets defined above.
Furthermore, we have 
$\{w_i^\mathrm{t}, \max_i^q\} \prec (X_i^\mathrm{t} \setminus \{y_i^q\}) \cup \{\max_i^q\}$
and 
$(X_i^\mathrm{t} \setminus \{\overline{y_i^q}\}) \cup \{\overline{\min_i^q}\} \prec \{w_\mathrm{f},\overline{\min_i^q}\}$
which allows us to set 
$X_i^\mathrm{t} \setminus \{\overline{y_i^q}\} \prec \{w_i\} \prec X_i^\mathrm{t} \setminus \{y_i^q\}$
which is consistent with the enforced
$X_i^\mathrm{t} \setminus \{\overline{y_i^q}\} \prec X_i^\mathrm{t} \prec X_i^\mathrm{f} \setminus \{y_i^q\}$.
The second case can be treated analogously.
\end{proof}

As before in Corollary~\ref{Cly:StrictExt} and \ref{Cly:DI^SE-NP-comp}
we can double the elements that would otherwise appear in singletons
to show the following corollary.

\begin{center}
\begin{mdframed}[style=mystyle]
\begin{minipage}{\textwidth}

\medskip
\begin{Cly}\label{Cly:DI^SE-Pi-comp}
\textsc{Strong $DI^SE$-WO-Orderability} is $\Pi_2^p$-complete.
\end{Cly}
\end{minipage}
\end{mdframed}
\end{center}

\begin{proof}
We additionally need to double the elements $w_i^\mathrm{t}$ and $w_i^\mathrm{f}$
and add, as for $z_i^a$ and $\overline{z_i^a}$, the sets
$\{w_i^{\mathrm{t},1}, w_i^{\mathrm{t},2}, w_i^{\mathrm{f},1}\}$,
$\{w_i^{\mathrm{t},2}, w_i^{\mathrm{f},1}\}$ and $\{w_i^{\mathrm{t},2},w_i^{\mathrm{f},1}, w_i^{\mathrm{f},2}\}$.
Furthermore, as before, we enforce 
\[X_i^\mathrm{t} \setminus \{y_i^q\} \preceq \{w_i^{\mathrm{t},1}, w_i^{\mathrm{\mathrm{t},2}}\}
\mbox{ and } \{w_i^{\mathrm{t},1},z_j^{\mathrm{t},2}\} \preceq X_i^\mathrm{f} \setminus \{y_i^q\}.\]
Then, $w_i^{\mathrm{t},1} < w_i^{\mathrm{t},2} < w_i^{\mathrm{f},1} <  w_i^{\mathrm{f},2}$,
enforces $X_i^\mathrm{t} \setminus \{y_i^q\} \prec X_i^\mathrm{f} \setminus \{y_i^q\}$
and hence $X_i^\mathrm{t} \prec X_i^\mathrm{f}$.
Similarly, we can ensure that
$w_i^{\mathrm{f},2} < w_i^{\mathrm{t},1} < w_i^{\mathrm{t},2} <  w_i^{\mathrm{t},1}$ 
enforces $X_i^\mathrm{f} \setminus \{\overline{y_i^q}\} \prec X_i^\mathrm{t} \setminus \{\overline{y_i^q}\}$
and hence $X_i^\mathrm{f} \prec X_i^\mathrm{t}$.
Therefore, every assignment to the universally quantified variables in $\phi$ is encoded by some  
linear order on $X$. The rest of the proof works as before.
\end{proof}

We conclude this part by observing that the orders constructed in the proof of 
Theorem~\ref{PI1} and Corollary~\ref{Cly:DI^SE-Pi-comp} are not only weak but even linear orders.
Therefore, the reductions used for these results show also the $\Pi_2^p$
completeness of \textsc{Strong $DI^S$-LO-Orderability} and \textsc{Strong $DI^SE$-LO-Orderability}.

\begin{center}
\begin{mdframed}[style=mystyle]
\begin{minipage}{\textwidth}

\medskip
\begin{Cly}\label{PI1-LO}
\textsc{Strong $DI^S$-LO-Orderability} and
\textsc{Strong $DI^SE$-LO-Orderability} are $\Pi_2^p$-complete.
\end{Cly}
\end{minipage}
\end{mdframed}
\end{center}

\subsubsection*{Strong $DI$-WO-Orderability and Strong $DIE$-WO-Orderability}
\subsubsectionmark{Strong $DI$-WO-Orderability}

To conclude the section on $\leq$-orderability and 
strong orderability, we show that the hardness results we showed before
also hold if we replace strict independence with regular independence.

\begin{center}
\begin{mdframed}[style=mystyle]
\begin{minipage}{\textwidth}

\medskip
\begin{Thm}\label{PI2}
\textsc{Strong $DI$-WO-Orderability} is $\Pi_2^p$-complete.

\textsc{$DI$-WO-Orderability} is \NP-complete.
\end{Thm}
\end{minipage}
\end{mdframed}
\end{center}

\begin{proof}
We observe that the reductions for Proposition~\ref{Prop:StronpNP}
and  Theorem~\ref{PI1}
only use strict independence once:
Namely when enforcing a preference 
$X_i^a \setminus \{y_j^b\} \preceq \{z_j^b\}$
strict independence ensures that
\[A_1 = (X_i^a\setminus \{y_j^b,z_j^b,\}) \cup \{\Max_j^b\} \prec \{\Max_j^b\} = A_l\]
enforces the desired preference 
\[(X_i^a\setminus \{y_j^b\}) \cup \{\Max_j^b\} \prec \{z_j^b, \Max_j^b\}.\]

We can achieve the same result with dominance and independence as follows:
We add the same sequence $A_1, \dots, A_l$ as in the proof of Proposition~\ref{Prop:StronpNP}
and additionally, the set $(X_i^a\setminus \{y_j^b, x_i^-\}) \cup \{\Max_j^b\}$.
We observe that for the sequence $A_1, \dots, A_l$ added to enforce the preference
$(X_j^a\setminus \{y_i^b\}) \cup \{\Max_i^c\}\prec \{z_i, \Max_i^d\}$,
we have $l > 3$ and that the following preference is enforced by dominance
\[A_2 = (X_i^a\setminus \{y_j^b,z_j^b,x_i^-\}) \cup \{\Max_j^b\} \prec \{\Max_j^b\} = A_l\]
which enforces by independence
\[(X_i^a\setminus \{y_j^b, x_i^-\}) \cup \{\Max_j^b\} \prec \{z_j^b, \Max_j^b\}.\]
Now, by construction $x_i^- < \min (X_i^a\setminus \{y_j^b, x_i^-\})$ holds
and hence by dominance
\[(X_i^a\setminus \{y_j^b\}) \cup \{\Max_j^b\} \prec
X_i^a\setminus \{y_j^b, x_i^-\}) \cup \{\Max_j^b\} \preceq \{z_j^b, \Max_j^b\}. \]
Therefore, we get by transitivity the desired 
\[(X_i^a\setminus \{y_j^b\}) \cup \{\Max_j^b\} \prec \{z_j^b, \Max_j^b\}.\]

Now, for any linear order $\leq'$ the newly added sets can easily be 
added to an order satisfying dominance and independence:
Any new set $(X_i^a\setminus \{y_j^b, x_i^-\}) \cup \{\Max_j^1\}$
can be added in the order $\preceq$ right after $(X_i^a\setminus \{y_j^b\}) \cup \{\Max_j^1\}$ 
if $x_i^-<' v_1$ or right before $(X_i^a\setminus \{y_j^b\}) \cup \{\Max_j^1\}$ if $v_1 <' x_i^-$.
\end{proof}

The newly added sets are not singletons. 
Hence, the fact that we can add the extension rule without changing the 
complexity of the problem still holds as before.

\begin{center}
\begin{mdframed}[style=mystyle]
\begin{minipage}{\textwidth}

\medskip
\begin{Cly}\label{DIE-PI-comp}
\textsc{Strong $DIE$-WO-Orderability} is $\Pi_2^p$-complete and
\textsc{$DIE$-WO-Orderability} is \NP-complete.
\end{Cly}
\end{minipage}
\end{mdframed}
\end{center}

As above, we observe that the orders constructed in Theorem~\ref{PI2} and
Corollary~\ref{DIE-PI-comp} are linear orders.
Therefore, we can also conclude that
\textsc{Strong $DI$-LO-Orderability} and \textsc{Strong $DIE$-LO-Orderability}
are $\Pi_2^p$-complete and that \textsc{$DI$-LO-Orderability} and  \textsc{$DI$-LO-Orderability} are \NP-complete.

\begin{center}
\begin{mdframed}[style=mystyle]
\begin{minipage}{\textwidth}

\medskip
\begin{Cly}\label{PI2-LO}
\textsc{Strong $DI$-LO-Orderability} and
\textsc{Strong $DIE$-LO-Orderability} are $\Pi_2^p$-complete.
Moreover, \textsc{$DI$-LO-Orderability} and  \textsc{$DIE$-LO-Orderability} are \NP-complete.
\end{Cly}
\end{minipage}
\end{mdframed}
\end{center}

We conclude this section with an important observation:
The fact that strong orderability is $\Pi_2^p$-complete implies
under the commonly believed complexity theoretic assumption $\coNP \neq  \Pi_2^p$ that
constructing an order satisfying dominance, (strict) independence 
and the extension rule is hard even if we know that the family $\cX$
is strongly orderable. In other words, even if we know that a weak order
satisfying our axioms must exist, it may be hard to construct it.

\begin{center}
\begin{mdframed}[style=mystyle]
\begin{minipage}{\textwidth}

\medskip
\begin{Cly}\label{Cly}
Assume $\coNP \neq  \Pi_2^p$. Then there exists no polynomial time algorithm $\mathbb{A}$
with the following specifications: 
\begin{itemize}
\item $\mathbb{A}$ takes as input a set $X$, a family of sets $\cX \subseteq \powerset{X}$ 
and a linear order $\leq$ on $X$.
\item if $\cX$ is strongly $DI$-orderable, then $\mathbb{A}$ produces on input $(X,\cX,\leq)$
an order $\preceq$ on $\cX$ that satisfies dominance and independence.
\end{itemize}
The same holds if we replace independence by strict independence
or if we add the extension rule.
\end{Cly}
\end{minipage}
\end{mdframed}
\end{center}

\begin{proof}
We claim that \textsc{Strong $DI$-WO-Orderability} would be in \coNP{} if 
there exists a polynomial time algorithm $\mathbb{A}$
with the specifications above.
To do so, we show that there exists a certificate for no instances
if such an algorithm exists.
Observe that there exists a linear order $\leq$ on $X$ that cannot be lifted
if and only if $(X,\cX)$ is negative instance of \textsc{Strong $DI$-WO-Orderability}. 
Hence $\leq$ is a certificate (of polynomial size) for the fact that 
$(X,\cX)$ is a negative instance. 
Furthermore, one can check the certificate by running $\mathbb{A}$ on $(X, \cX, \leq)$.
Then, one only needs to check that the produced order does not satisfy dominance and 
strict independence. By definition, this can only be the case if $(X,\cX)$ is a
negative instance of \textsc{Strong $DI$-WO-Orderability}.
The argument for strict independence and extension is analogous.
\end{proof}

\subsection{Partial Orders}\label{sec:partial}

In this section, we investigate the 
effect of dropping the requirement that the lifted order should be total.
For dominance and independence, we already have seen that they 
always can be jointly satisfied, if we expect the lifted order to be only 
a preorder.

\begin{center}
\begin{mdframed}[style=mystyle]
\begin{minipage}{\textwidth}

\medskip
\begin{Thm}[Folklore]
For every set $X$, linear order $\leq$ on $X$ and family of sets
$\cX \subseteq \powerset{X}$, there is a preorder that 
satisfies dominance, independence and the extension rule.
\end{Thm}
\end{minipage}
\end{mdframed}
\end{center}

\begin{proof}
In Example~\ref{Exp:precmm2}, we defined the preorder $\preceq_{pmm}$ and showed that
it always satisfies dominance and independence. By definition, it also satisfies 
the extension rule.
\end{proof}

In other words, every family of sets 
is strongly $DI(E)$-orderable if we only require the lifted order to be a preorder.
On the other hand, Barbera and Pattanaik's theorem
tells us that this is not the case for dominance and strict independence.

We observe that in many applications partial orders are more
common than preorders.
For example, if we want to avoid total orders when applying lifted orders in combinatorial voting
as shown in Figure~\ref{fig:1}, then we need voting rules that can handle
these incomplete orders. However, most voting rules for incomplete preferences 
require partial orders as input \citep{boutilier2016}
though voting rules for even weaker preference models exist \citep{xia2011,endriss2019}.
Therefore, we are often more interested in the complexity 
of deciding for a given family of sets if
there exists a partial order
that satisfies dominance and (strict) independence
resp.\ dominance, (strict) independence and the extension rule.
As strict independence and independence coincide for 
strict orders, we can focus just on strict independence here.

For $\leq$-orderability, this problem was already considered 
by \citet{Maly2017}. They showed that it is possible to decide 
whether a family of sets is $\leq$-$DI^S$- or $\leq$-$DI^SE$-orderable
using a polynomial time fix point construction. In particular, this
result is constructive.

\begin{center}
\begin{mdframed}[style=mystyle]
\begin{minipage}{\textwidth}

\medskip
\begin{Thm}[\citeauthor{Maly2017}]\label{Cly:DI^S-partial}
\textsc{$DI^S$-PO-orderability} and \textsc{$DI^SE$-PO-orderability} are in \P.
\end{Thm}
\end{minipage}
\end{mdframed}
\end{center}

In this paper, we consider \emph{strongly partially $DI^S$-orderable} families of sets,
i.e., families of sets $\cX \subseteq \powerset{X}$ such that
for every linear order $\leq$ on $X$
there exists a partial order on $\cX$ that satisfies dominance and strict
independence with respect to $\leq$.
It turns out that it is still difficult
to decide whether a given family of sets is strongly partially $DI^S$-orderable.
We show this by a reduction from \textsc{Taut}.

\begin{center}
\begin{mdframed}[style=mystyle]
\begin{minipage}{\textwidth}

\medskip
\begin{Thm}\label{partialCoNP}
\textsc{Strong $DI^S$-PO-Orderability} is \coNP-complete.
\end{Thm}
\end{minipage}
\end{mdframed}
\end{center}

\begin{proof}
Let $\phi$ be an instance of \textsc{Taut},
i.e., the problem of checking whether a 3-DNF is a tautology.
We construct an instance $(X,\cX)$
of \textsc{Strong $DI^S$-PO-orderability}.
For every variable $X_i$ in $\phi$ we add new elements
$x_i^\mathrm{f}$ and $x_i^\mathrm{t}$ to $X$.
We call the set of these elements $V$.
We will treat every order on $X$ as encoding a truth assignment
by equating $x_i^\mathrm{f} < x_i^\mathrm{t}$ to $X_i$ is true and
$x_i^\mathrm{t} < x_i^\mathrm{f}$ to $X_i$ is false.
Then, we add for every disjunct $C_j$
new variables $y_j^\mathrm{t}, y_j^\mathrm{f}$. We call the set of these elements $Y$.
Essentially, we want to add sets such that $\{y_j^\mathrm{t}\} \prec \{y_j^\mathrm{f}\}$
holds for the minimal partial order satisfying dominance and strict independence
with respect to $\leq$
if and only if $C_j$ is not satisfied by the truth assignment encoded by $\leq$.
Then, we will add sets that lead to a contradiction if 
$\{y_j^\mathrm{t}\} \prec \{y_j^\mathrm{f}\}$ holds for all disjuncts. 

To achieve this, we add for every disjunct $C_j$ elements $c_j$ as well as $d_j^k$
and $e_j^k$ for $k \leq 3$. 
Finally, we add new elements
$u,v,z_1$ and $z_2$. The elements $u$ and $v$ will be used to 
generate a contradiction if 
$\{y_j^\mathrm{t}\} \prec \{y_j^\mathrm{f}\}$ holds for all disjuncts. 
The elements $z_1$ and $z_2$ have the same purpose as the elements
$v_1$ and $v_2$ in the proof of Proposition~\ref{DI^S-NP-hard},
i.e., the preference between $z_1$ and $z_2$ determines if 
$x_i^\mathrm{f} < x_i^\mathrm{t}$ encodes that $x_i$ is set to true or to false.

Next we fix a class of linear orders on $X$ that we call
critical linear orders. We want to show that for a critical linear order $\leq$
there exists a partial order on $\cX$ satisfying dominance and 
strict independence only if $\leq$ encodes a satisfying truth assignment.
As any possible truth assignment is encoded by a critical linear order,
this means that $\cX$ cannot be strongly partially $DI^S$-orderable if 
$\phi$ is not a tautology.
We call any linear order on $X$ that is derived by replacing $V$ with an arbitrary linear 
order on the elements in $V$ in the following linear order
a critical linear order:
\begin{multline*}
u < c_1 < \dots < c_m < y_1^\mathrm{t} < \dots < y_m^\mathrm{t} < \\ d_1^1 < \dots < d_m^3
< V < e_1^1 < \dots < e_m^3 < y_1^\mathrm{f} < \dots < \\y_m^\mathrm{f} < z_1 < z_2 < v
\end{multline*}  
In the following, we write $\preceq_{\infty}$ for the minimal partial order
satisfying dominance and strict independence with respect to some linear order on $V$.

Next, we build the family $\cX$.
First, we make sure, that any order satisfying 
dominance and strict independence with respect to 
a critical linear order $\leq$ must reflect the truth assignment encoded by $\leq$.
To this end, we add
$\{x_i^\mathrm{f}\}$, $\{x_i^\mathrm{f},x_i^\mathrm{t}\}$ and $\{x_i^\mathrm{t}\}$ for all $x_i^\mathrm{f},x_i^\mathrm{t}\in V$.
Then, for every linear order $\leq$
we have
$\{x_i^\mathrm{f}\} \prec_{\infty} \{x_i^\mathrm{t}\}$ if $x_i^\mathrm{f} < x_i^\mathrm{t}$ and,
on the other hand, 
$\{x_i^\mathrm{t}\} \prec_{\infty} \{x_i^\mathrm{f}\}$ if $x_i^\mathrm{t} < x_i^\mathrm{f}$

Next, we add sets such that $\{y_i^\mathrm{t}\} \prec_{\infty} \{y_i^\mathrm{f}\}$
holds for a critical order $\leq$ if the assignment encoded by $\leq$
does not satisfy disjunct $C_i$.
Essentially, we add for every variable $X_i$ that appears positive in disjunct $C_j$
a collection of sets that imply $\{y_i^\mathrm{t}\} \prec_{\infty} \{y_i^\mathrm{f}\}$ if
$\{x_i^\mathrm{t}\} \prec_{\infty} \{x_i^\mathrm{f}\}$ holds. 
Similarly, we add a collection of sets that imply $\{y_i^\mathrm{t}\} \prec_{\infty} \{y_i^\mathrm{f}\}$ if
$\{x_i^\mathrm{f}\} \prec_{\infty} \{x_i^\mathrm{t}\}$ holds for every variable that appears negatively in $C_j$.
Now, let $C_j = X_{i_1} \wedge X_{i_2} \wedge X_{i_3}$ be a disjunct. Then, we add sets 
\[\{y_j^\mathrm{t}\}, \{y_j^\mathrm{t},d_j^k\}, \{y_j^\mathrm{t},d_j^k,x^\mathrm{f}_{i_k}\}, \{d_j^k,x^\mathrm{f}_{i_k}\}\]
for all $k \in \{1,2,3\}$
as well as 
\begin{multline*}
\{x^\mathrm{t}_{i_k},e_j^k\}, \{x^\mathrm{t}_{i_k},e_j^k, z_1\}, \{x^\mathrm{t}_{i_k},e_j^k, z_1, z_2\},\\
 \{e_j^k, z_1, z_2\},
\{e_j^k,z_1, z_2, y_\mathrm{f}^j\},\{z_1, z_2, y_\mathrm{f}^j\}, \{z_2, y_\mathrm{f}^j\},\{y_\mathrm{f}^j\}.
\end{multline*}
If any of the variables occurs negatively in $C_j$, we switch $x_{i_k}^\mathrm{f}$ and $x_{i_k}^\mathrm{t}$
in the construction.
We claim that these sets ensure that 
$\{y_j^\mathrm{t}\} \prec_{\infty} \{y_j^\mathrm{f}\}$ holds for any critical linear order whenever at least one 
literal in $C_j$ is false. 
We have 
\[\{y_j^\mathrm{t}\} \prec_{\infty} \{y_j^\mathrm{t}, d_j^k\} \prec_{\infty} 
\{y_j^\mathrm{t}, d_j^k, x_{i_k}^\mathrm{f}\} \prec_{\infty} \{d_j^k, x_{i_k}^\mathrm{f}\} \prec_{\infty} \{x_{i_k}^\mathrm{f}\}\]
by dominance and, hence, by transitivity $\{y_j^\mathrm{t}\}  \prec_{\infty} \{x_{i_k}^\mathrm{f}\}$.
Similarly, we have $\{x_{i_k}^\mathrm{t}\} \prec_{\infty} \{y_j^\mathrm{f}\}$.
Hence, $\{x_{i_k}^\mathrm{f}\} \prec_{\infty} \{x_{i_k}^\mathrm{t}\}$ implies $\{y_j^\mathrm{t}\} \prec_{\infty} \{y_j^\mathrm{f}\}$
by transitivity.

Now, we add sets that lead to a contradiction if $\{y_j^\mathrm{t}\} \prec_{\infty} \{y_j^\mathrm{f}\}$
holds for all $j$ for a critical linear order $\leq$, i.e., if the assignment encoded
by $\leq$ does not satisfy any disjunct.
Roughly, we build a sequence of sets
$A_1, \overline{A_1},A_2, \overline{A_2}, \dots ,A_m, \overline{A_m}$ such that 
\begin{enumerate}
\item $\{u,v\} \prec_{\infty} A_1$ ,
\item $A_j \prec_{\infty} \overline{A_j}$ if $\{y_j^\mathrm{t}\} \prec_{\infty} \{y_j^\mathrm{f}\}$ for all $j\leq m$,
\item $\overline{A_j} \prec_{\infty} A_{j+1}$ for all $j<m$.
\item $\overline{A_m} \prec_{\infty} \{u,v\}$,
\end{enumerate}
Then, if $\{y_j^\mathrm{t}\} \prec_{\infty} \{y_j^\mathrm{f}\}$
holds for all $j$ we have $\{u,v\} \prec_{\infty} \{u,v\}$ and hence 
no partial order on $\cX$ can satisfy dominance and strict independence 
with respect to $\leq$.
For every $j \leq m$ we have:
\begin{itemize}
\item $A_j = \{u,c_{1},\dots ,c_j,y_j^\mathrm{t},y_{j-1}^\mathrm{f}, y_{j-2}^\mathrm{f}, \dots , y_{1}^\mathrm{f},v\}$ 
\item $\overline{A_j} = \{u, c_{1},\dots ,c_j, y_j^\mathrm{f}, y_{j-1}^\mathrm{f}, y_{j-2}^\mathrm{f}, \dots, y_{1}^\mathrm{f},v\}$
\end{itemize}
First we add 
\[\{u\}, \{u,c_1\}, \{u,c_1,y_1^\mathrm{t}\},\{u,c_1, y_1^\mathrm{t}, v\}, \{u,v\}.\]
Then, we know for any critical linear order 
that
\[
\{u\} \prec_{\infty} \{u,c_1\} \prec_{\infty}\{u,c_1,y_1^\mathrm{t}\} 
\]
holds by dominance and therefore we have 
$\{u,v\} \prec_{\infty} \{u,c_1,y_1^\mathrm{t},v\}$.
This is the desired property (1) i.e., $\{u,v\} \prec_{\infty} A_1$.
Now, we add for every disjunct 
$\{c_j,y_j^\mathrm{t}\}$ and $\{c_j, y_j^\mathrm{f}\}$.
Then, we add new sets that are constructed by incrementally adding to both sets, one by one,
first all elements $c_{j-1}$ to $c_1$, then all elements $y_{j-1}^\mathrm{f}$ to $y_1^\mathrm{f}$
and finally $u$ and $v$ in that order. In other words we add
\begin{multline*}
\{c_{j-1},c_j,y_j^\mathrm{t}\} \mbox{ and } \{c_{j-1},c_j, y_j^\mathrm{f}\}, \\
\{c_{j-2},c_{j-1},c_j,y_j^\mathrm{t}\} \mbox{ and } \{c_{j-2},c_{j-1},c_j, y_j^\mathrm{f}\}, 
\dots,\\
 \{c_{1},\dots ,c_j,y_j^\mathrm{t}\} \mbox{ and } \{c_{1},\dots ,c_j, y_j^\mathrm{f}\}
\end{multline*}
as well as
\begin{multline*}
\{c_{1},\dots ,c_j,y_j^\mathrm{t}, y_{j-1}^\mathrm{f}\} \mbox{ and } \{c_{1},\dots ,c_j, y_j^\mathrm{f},y_{j-1}^\mathrm{f}\},
\dots,\\ \{c_{1},\dots ,c_j,y_j^\mathrm{t},y_{j-1}^\mathrm{f}, \dots , y_{1}^\mathrm{f}\} \mbox{ and } \{c_{1},\dots ,c_j, y_j^\mathrm{f},y_{j-1}^\mathrm{f}, \dots, y_{1}^\mathrm{f}\}
\end{multline*}
and finally
\[\{u,c_{1},\dots ,c_j,y_j^\mathrm{t}, y_{j-1}^\mathrm{f}, \dots , y_{1}^\mathrm{f}\} \mbox{ and } \{u, c_{1},\dots ,c_j, y_j^\mathrm{f}, y_{j-1}^\mathrm{f}, \dots, y_{1}^\mathrm{f}\},\]
as well as
\[\{u,c_{1},\dots ,c_j,y_j^\mathrm{t}, y_{j-1}^\mathrm{f}, \dots , y_{1}^\mathrm{f},v\} \mbox{ and } \{u, c_{1},\dots ,c_j, y_j^\mathrm{f}, y_{j-1}^\mathrm{f}, \dots, y_{1}^\mathrm{f},v\}.\]
By construction  
\[\{u,c_{1},\dots ,c_j,y_j^\mathrm{t}, y_{j-1}^\mathrm{f}, \dots , y_{1}^\mathrm{f},v\}
\prec_{\infty} \{u, c_{1},\dots ,c_j, y_j^\mathrm{f}, y_{j-1}^\mathrm{f}, \dots, y_{1}^\mathrm{f},v\}\]
holds for the minimal partial order satisfying dominance and strict independence for any linear order on $V$
if and only if $\{y_j^\mathrm{t}\} \prec_{\infty} \{y_j^\mathrm{f}\}$ holds for that partial order.
This is the desired property (2) of the sequence.

Next, we add the following sets:
\[\{u,c_{1},\dots ,c_j\}, \{u,c_{1},\dots ,c_{j+1}\} \text{ and } \{u, c_{1},\dots ,c_{j+1}, y_{j+1}^\mathrm{t}\}.\]
Then, we add new sets derived as above by adding to both sets first all 
elements $y_j^\mathrm{f}$ to $y_1^\mathrm{f}$ 
and then $v$, one by one, in that order until we reach
\[\{u,c_{1},\dots ,c_j, y_j^\mathrm{f}, y_{j-1}^\mathrm{f}, \dots, y_{1}^\mathrm{f},v\} \mbox{ and }
\{u, c_{1},\dots ,c_{j+1}, y_{j+1}^\mathrm{t},y_j^\mathrm{f}, y_{j-1}^\mathrm{f}, \dots, y_{1}^\mathrm{f},v\}.\]
Then,
the desired property (3)
\[\{u,c_{1},\dots ,c_j, y_j^\mathrm{f}, y_{j-1}^\mathrm{f}, \dots, y_{1}^\mathrm{f},v\} \prec_{\infty}
\{u, c_{1},\dots ,c_{j+1}, y_{j+1}^\mathrm{t},y_j^\mathrm{f}, y_{j-1}^\mathrm{f}, \dots, y_{1}^\mathrm{f},v\}\]
holds for the critical linear order by strict independence because
\[\{u,c_{1},\dots ,c_j\} \prec \{u,c_{1},\dots ,c_{j+1}\} \prec \{u, c_{1},\dots ,c_{j+1}, y_{j+1}^\mathrm{t}\}\]
holds by dominance.
Finally, we add 
$\{v\}$
and then 
$\{y_{1}^\mathrm{f},v\}$,
$\{y_2^\mathrm{f}, y_{1}^\mathrm{f},v\}$ 
and so on till we reach
\[\{c_1, \dots, c_m,y_m^\mathrm{f}, y_{m-1}^\mathrm{f}, \dots, y_{1}^\mathrm{f},v\}.\]
This forces for any critical linear order the desired property (4):
\[\{u, c_1, \dots, c_m,y_m^\mathrm{f}, y_{m-1}^\mathrm{f}, \dots, y_{1}^\mathrm{f},v\} \prec_{\infty} \{u,v\}.\]
Now, by construction and transitive we have for any critical linear order
\begin{multline*}
\{u,v\} \prec_{\infty} \{u,c_1,y_1^\mathrm{t},v\} \prec_{\infty}\\ \{u,c_1,y_1^\mathrm{f},v\} \prec_{\infty}
\{u,c_1,c_2,y_2^\mathrm{t},y_1^\mathrm{f},v\} \prec_{\infty} \dots \\
\prec_{\infty} \{u, c_1, \dots, c_m,y_m^\mathrm{f}, y_{m-1}^\mathrm{f}, \dots, y_{1}^\mathrm{f},v\} \prec_{\infty} \{u,v\}
\end{multline*}
if (and only if) $\{y_j^\mathrm{t}\} \prec_{\infty} \{y_j^\mathrm{f}\}$ holds for all disjuncts, 
i.e., if the critical linear order encodes an unsatisfying assignment.
It follows that if $\phi$ is not a tautology, then $(X,\cX)$ is not 
strongly partial $DI^S$-orderable.

It remains to show that $(X,\cX)$ is strongly partial $DI^S$-orderable
if $\phi$ is a tautology. Let $\leq$ be a linear order on $X$.
As in the proof of Proposition~\ref{DI^S-NP-hard}
we can assume, by Lemma \ref{Lem:inverse}, that $z_1 < z_2$.
We construct a partial order $\preceq$ that satisfies dominance and strict independence
with respect to $\leq$.
To avoid complicated case distinctions, we will describe the construction 
only for disjuncts with all positive variables.
The only change in construction required for negative variables  
is switching $x_i^\mathrm{f}$ and $x_i^\mathrm{t}$.

First we add the forced preferences between
$\{x_i^\mathrm{f}\}, \{x_i^\mathrm{f},x_i^\mathrm{t}\}$ and $\{x_i^\mathrm{t}\}$,
i.e., $\{x_i^\mathrm{f}\} \prec \{x_i^\mathrm{f},x_i^\mathrm{t}\} \prec \{x_i^\mathrm{t}\}$
if $x_i^\mathrm{f} < x_i^\mathrm{t}$ holds and 
$\{x_i^\mathrm{t}\} \prec \{x_i^\mathrm{f},x_i^\mathrm{t}\} \prec \{x_i^\mathrm{f}\}$
if $x_i^\mathrm{t} < x_i^\mathrm{f}$ holds.
Next, we consider the sets containing an element $d_j^k$.
We add all preferences that are implied by dominance
between the following sets:
\[\{y_j^\mathrm{t}\}, \{y_j^\mathrm{t},d_j^k\}, \{y_j^\mathrm{t},d_j^k,x^\mathrm{f}_{i_k}\} ,\{d_j^k,x^\mathrm{f}_{i_k}\}, \{x^\mathrm{f}_{i_k}\}\]
Then we close under transitivity.
We recall that applying dominance and transitivity can never lead to a contradiction.
Furthermore, we claim that $\prec$ restricted to these sets already satisfies strict independence:
The only possible application of strict independence on these sets is 
that any preference between 
$\{y_j^\mathrm{t}\}$ and $\{x^\mathrm{f}_{i_k}\}$ has to be lifted to
$\{y_j^\mathrm{t},d_j^k\}$ and $\{d_j^k,x^\mathrm{f}_{i_k}\}$. 
By construction however, there can only be a preference 
between 
$\{y_j^\mathrm{t}\}$ and $\{x^\mathrm{f}_{i_k}\}$ forced by dominance and transitivity 
if the same preference holds between
$\{y_j^\mathrm{t},d_j^k\}$ and $\{d_j^k,x^\mathrm{f}_{i_k}\}$
as for $\{y_j^\mathrm{t}\}$ dominance can only force a relation to 
$\{y_j^\mathrm{t},d_j^k\}$ 
and for $\{x^\mathrm{f}_{i_k}\}$ it can only force a relation to $\{d_j^k,x^\mathrm{f}_{i_k}\}$.
Moreover, because we assume that no variable occurs twice in a disjunct,
a preference between $\{y_j^\mathrm{t}\}$ and $\{x^\mathrm{f}_{i_k}\}$ 
cannot later be introduced through sets containing
another $d_{j'}^{k'}$.
Indeed the only preferences we need to add for sets containing different elements $d_{j_1}^{k_1}$ and $d_{j_2}^{k_2}$,
in order to satisfy dominance and transitivity is 
$\{x_i^\mathrm{f}, d_{j_1}^{k_1}\} \prec \{x_i^\mathrm{f}, d_{j_2}^{k_2}\}$
for all $x_i^\mathrm{f}$ and all $d_{j_1}^{k_1}, d_{j_2}^{k_2}$ such that $d_{j_1}^{k_1} < d_{j_2}^{k_2}$
holds. 

Using a similar construction, we can order all sets containing an element
$e_j^k$ if we replace $x_i^\mathrm{f}$ by $x_i^\mathrm{t}$ and $y_j^\mathrm{t}$ by $\{z_1,z_2,y_j^\mathrm{f}\}$.
Finally, we add the enforced preference between $\{z_2,y_j^\mathrm{f}\}$
and $\{y_j^\mathrm{f}\}$ as well as $\{z_1,z_2,y_j^\mathrm{f}\} \prec \{z_2,y_j^\mathrm{f}\}$.
The later is enforced by dominance as we assume $z_1 < z_2$. 
Then, we close everything under transitivity.
By construction, this does not produce any new instances of strict independence.

We now consider the sets containing an element $c_i$ for some $i$.
Observe that $\{z_1,z_2,y_j^\mathrm{f}\} \prec \{z_2,y_j^\mathrm{f}\}$
implies that only $\{x_{i_k}^\mathrm{t}\}  \prec \{y_j^\mathrm{f}\}$ can be forced but not 
$\{y_j^\mathrm{f}\} \prec \{x_{i_k}^\mathrm{t}\}$.
This implies that $\{y_j^\mathrm{f}\} \prec \{y_j^\mathrm{t}\}$
never holds.
Moreover, $\{y_j^\mathrm{t}\} \prec \{y_j^\mathrm{f}\}$ can only be enforced 
if $\{x_i^\mathrm{f}\} \prec \{x_i^\mathrm{t}\}$ holds for a variable occurring in disjunct $C_j$
i.e., if disjunct $C_j$ is not satisfied.
As $\phi$ is a tautology, there is disjunct $C_l$ that is satisfied by the
assignment encoded by $\leq$.
Hence, $\{y_l^\mathrm{t}\}$ and $\{y_l^\mathrm{f}\}$ are incomparable.
We partition the sets containing an element $c_i$ in partitions $P_1, \dots, P_m$
based on the largest $i$ for which they contain $c_i$.
We set $S_1 \prec S_2$ if $S_1 \in P_{i_1}$,
$S_2 \in P_{i_2}$ and
one of the following holds:
\begin{itemize}
\item $c_{i_1} < c_{i_2}$ and $i_1,i_2 < l$
\item $c_{i_1} < c_{i_2}$ and $l < i_1,i_2$
\end{itemize}
Any set that contains $c_i$ also contains $y_i$
except $\{u,c_1, \dots, c_i\}$. Hence, all sets from 
different partitions differ by at least two elements 
except $\{u,c_1, \dots, c_i\}$ and $\{u,c_1, \dots, c_{i+1}\}$.
If dominance forces a preference between these sets,
it is satisfied by construction for $i,i+1 \neq l$.
Now, for any set in any partition $P_i$ such that $i \neq l$
we set $S \prec S'$ if $y_i^\mathrm{t} \in S$ and $y_i^\mathrm{t} \not\in S'$.
This covers all applications of strict independence in a partition.
Finally, we add all preferences that are forced by dominance in a partition
and close under transitivity.
We observe that $S, S \cup \{x\} \in P_i$ implies either 
$y_i^\mathrm{t} \in S, S \cup \{x\}$ or $y_i^\mathrm{t} \not \in S, S \cup \{x\}$,
hence this cannot lead to a contradiction.
Now, for a set $S$ in $P_l$ such that $y_l^\mathrm{t} \in S$ we set 
\begin{itemize}
\item $S' \prec S$ if $S' \in P_i$ for $i < l$ and $c_i < c_l$
\item $S \prec S'$ if $S' \in P_i$ for $i < l$ and $c_l < c_i$
\end{itemize}
Furthermore, for a set $S$ in $P_l$ such that $y_l^\mathrm{t} \not\in S$ we set 
\begin{itemize}
\item $S' \prec S$ if $S' \in P_i$ for $l < i$ and $c_i < c_l$
\item $S \prec S'$ if $S' \in P_i$ for $l < i$ and $c_l < c_i$
\end{itemize}
And finally, we add again all preferences forced by dominance and close by transitivity.
As $\{y_l^\mathrm{t}\}$ and $\{y_l^\mathrm{f}\}$ are incomparable in $\preceq$
this order is consistent. 
Furthermore, we did not add any preferences between any set not containing $c_l$
and sets containing $c_{l+1}$.
Hence, 
$\{u,c_1,y_1^\mathrm{t}, z_1, z_2,v\}$
and $\{u, c_1, \dots, c_m,y_m^\mathrm{f}, \dots, y_{1}^\mathrm{f},v\}$
are incomparable in $\preceq$. This allows us to add any preferences
forced by dominance and strict independence regarding $\{u\},\{v\}$
and $\{u,v\}$ without creating a contradiction.
By construction, $\preceq$ is now a partial order that satisfies
dominance and strict independence.
\end{proof}

We can extend the result to \textsc{Strong $DI^SE$-PO-Orderability}
using the same idea of doubling the singletons as before.

\begin{center}
\begin{mdframed}[style=mystyle]
\begin{minipage}{\textwidth}

\medskip
\begin{Cly}\label{partialCoNP-Exp}
\textsc{Strong $DI^SE$-PO-Orderability} is \coNP-complete.
\end{Cly}
\end{minipage}
\end{mdframed}
\end{center}

\begin{proof}
\textsc{Strong $DI^SE$-PO-Orderability} is in \coNP{}
by the same argument as \textsc{Strong $DI^S$-PO-Orderability}.
We modify the reduction used to prove Theorem~\ref{partialCoNP}
to show that it is also \coNP-hard.
All singletons appearing in that reduction are of the form 
$\{x_i^\mathrm{t}\}$, $\{x_i^\mathrm{f}\}$, $\{y_i^\mathrm{t}\}$, $\{y_i^\mathrm{f}\}$
$\{u\}$ or $\{v\}$ for some $i$. We change the reduction in a way that only 
singletons of the form $\{u\}$ or $\{v\}$ appear.
As before, we can achieve this by doubling the elements of the form
$x_i^a$ and $y_i^a$ for $a \in \{\mathrm{t},\mathrm{f}\}$, that means we replace 
$x_i^a$ by two elements $x_i^{a,1}$ and $x_i^{a,2}$,
and replace $y_i^a$ by $y_i^{a,1}$ and $y_i^{a,2}$. 
As in the proof of Corollary~\ref{Cly:DI^SE-Pi-comp}, 
we add the sets 
$\{x_i^{\mathrm{t},1}, x_i^{\mathrm{t},2}, x_i^{\mathrm{f},1}\}$,
$\{x_i^{\mathrm{t},2}, x_i^{\mathrm{f},1}\}$ and $\{x_i^{\mathrm{t},2},x_i^{\mathrm{f},1}, x_i^{\mathrm{f},2}\}$.
Then $x_i^{\mathrm{t},1} < x_i^{\mathrm{t},2} < 
x_i^{\mathrm{f},1} < x_i^{\mathrm{f},2}$ implies 
$\{x_i^{\mathrm{t},1}, x_i^{\mathrm{t},2}\} \prec \{x_i^{\mathrm{f},1}, x_i^{\mathrm{f},2}\}$.
and $x_i^{\mathrm{t},1} > x_i^{\mathrm{t},2} >
x_i^{\mathrm{f},1} > x_i^{\mathrm{f},2}$ implies 
$\{x_i^{\mathrm{t},1}, x_i^{\mathrm{t},2}\} \succ \{x_i^{\mathrm{f},1}, x_i^{\mathrm{f},2}\}$.
Furthermore, we add $\{y_j^{\mathrm{t},2},d_j^k\}$, $\{d_j^k, x_i^{\mathrm{t},1}\}$,
$\{x_i^{\mathrm{f},2},e_j^k\}$ and $\{z_2, y_j^{\mathrm{t},1}\}$
if variable $X_i$ appears positively in disjunct $j$. If $X_i$ appears 
negatively, we switch $\mathrm{t}$ and $\mathrm{f}$.
Then, $\{x_i^{\mathrm{t},1}, x_i^{\mathrm{t},2}\} \prec \{x_i^{\mathrm{f},1}, x_i^{\mathrm{f},2}\}$
implies $\{y_j^{\mathrm{t},1}, y_j^{\mathrm{t},2}\} \prec \{y_j^{\mathrm{f},1}, x_i^{\mathrm{f},2}\}$ 
as desired.

The second part of the construction must be modified by always adding first 
$y_j^{\mathrm{t},1}$ and then $y_j^{\mathrm{t},2}$ instead of $y_j^{\mathrm{t}}$ and 
similarly 
$y_j^{\mathrm{f},1}$ and then $y_j^{\mathrm{f},2}$ instead of $y_j^{\mathrm{f}}$.
It can be checked that this suffices to force $\{u,v\} \prec \{u,v\}$
whenever no disjunct is satisfied. 

Now, constructing a partial order that satisfies dominance and strict independence 
works as before. Furthermore, $\{u\}$ and $\{v\}$ are the only singletons in the family $\cX$.
As $\{u,v\} \in \cX$, any partial order that satisfies dominance also satisfies the extension rule.
\end{proof}

As mentioned before, independence and strict independence coincide 
in the case of partial orders.
This may be problematic if one wishes to use lifted orders in an application
where pre-orders can not be used, for example, if one wants to combine lifted 
orders with a voting rule that takes partial orders as input.
If in such a situation strict independence either turns out to be not satisfiable
in conjunction with dominance or it is considered to restrictive,
it may seem at first that the order lifting approach can not be used at all. 
However, it is possible to define a weaker 
version of independence that provides a different interpretation
of the monotonicity idea in the case of partial orders.

\begin{Axm}[Weak independence]
For all $A, B \in \mathcal{X}$ and for all $x \in X \setminus (A \cup B)$
such that $A \cup \{x\}, B \cup \{x\} \in \mathcal{X}$:
\[A \prec B  \text{ implies } B \cup \{x\} \not\prec A \cup \{x\}.\]
\end{Axm}

Clearly, any total relation satisfies weak independence if and only if it
satisfies independence. However, 
in contrast to strict independence and independence 
it is always possible to find a partial order
that satisfies dominance and weak independence.

\begin{Prop}
\LetX. Furthermore, let $\preceq$ be a preorder on $\cX$ that satisfies dominance
and independence. Then, the corresponding strict order is an antisymmetric, irreflexive
and transitive binary relation that satisfies dominance and weak independence.
\end{Prop}

\begin{proof}
By definition the corresponding strict order $\prec_s$
of $\preceq$ is an antisymmetric, irreflexive, transitive relation.
Furthermore, by the definition of dominance, $\preceq$ 
satisfies dominance if and only if $\prec_s$
satisfies dominance. Now assume $A,B, A \cup \{x\}, B \cup \{x\} \in \cX$ and
$A \prec_s B$. Then also $A \prec B$ must hold. Hence, by independence
either $A \cup \{x\} \sim B \cup \{x\}$ or $A \cup \{x\} \prec B \cup \{x\}$.
In both cases we have $B \cup \{x\} \not\prec_s A \cup \{x\}$.
Therefore, $\prec_s$ satisfies weak independence.
\end{proof}

If we add to $\prec_s$ the preference $A \sim A$ for all $A \in \cX$
we get a partial order that satisfies dominance and weak independence.

\subsection{Succinct Domain Restrictions}\label{sec:succ}

The results in the previous sections assume that the family of sets is given explicitly.
However, in many applications, the family of sets is instead only given implicitly,
via some condition that has to be satisfied by the sets in order to be admissible.
In this section, we consider a specific, well studied succinct representation -- 
families represented by boolean circuits -- and show that it can lead to a massive
blow up in complexity.
While boolean circuits are not commonly used for succinct representation
in practice, they are a useful tool to determine the effect that succinct representation can have in the worst case
as they are very expressive. In particular, boolean circuits
can simulate many succinct representation that are commonly used in practice, like propositional formulas or knowledge compilation languages such as NNFs \citep{Darwiche02}.
This worst case behavior is especially interesting in the light of the results of \citet{JAIR}
which show that many of the studied problems become polynomial time solvable if the families of sets is represented
as the set of connected subgraphs of a graph.
Now, how can boolean circuits be used for succinct representation?
Consider for example the following, simple family:
\[\{1,2\}, \{1,2,4\}\]
Assume that we know that we have a family of sets over four elements.
Then the family above can represented by the boolean string $11001101$.
This string in turn can be represented by a boolean circuit,
such that on the input of the number $i-1$ (in binary) the circuit 
outputs the $i$-th bit of the string. The circuit pictured in Figure~\ref{fig:booleanCirc}
is an example for a circuit representing the string $11001101$.

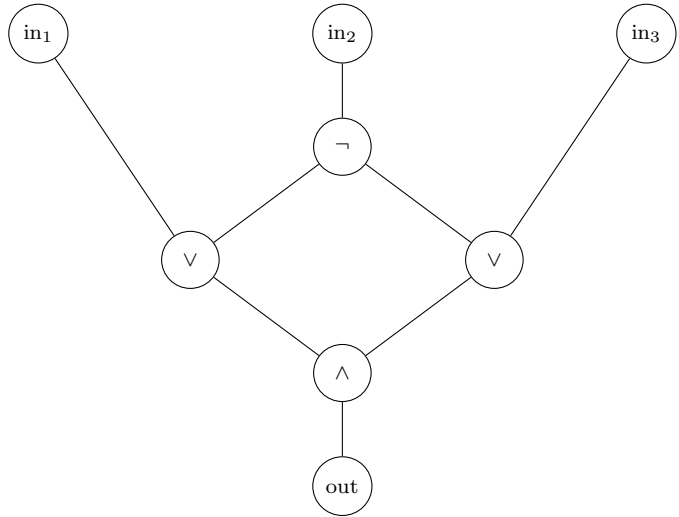
\begin{figure}
\begin{center}
\begin{tikzpicture}[font={\footnotesize}]
\node[draw,circle,minimum width = 5ex] (in1) at (0,0) {in$_1$};
\node[draw,circle,minimum width = 5ex] (in2) at (4,0) {in$_2$};
\node[draw,circle,minimum width = 5ex]  (in3) at (8,0) {in$_3$};
\node[draw,circle,minimum width = 5ex]  (neg) at (4,-1.5) {$\neg$};
\node[draw,circle,minimum width = 5ex] (vee1) at (2,-3) {$\vee$};
\node[draw,circle,minimum width = 5ex] (vee2) at (6,-3) {$\vee$};
\node[draw,circle,minimum width = 5ex] (and)  at (4,-4.5) {$\land$};
\node[draw,circle,minimum width = 5ex] (out)  at (4,-6) {out};
\draw (in1) -- (vee1) ;
\draw (in2) -- (neg) ;
\draw (in3) -- (vee2) ;
\draw (neg) -- (vee1) ;
\draw (neg) -- (vee2) ;
\draw (vee1) -- (and) ;
\draw (vee2) -- (and) ;
\draw (and) -- (out) ;
\end{tikzpicture}  
\end{center}
\caption{A boolean circuit representing the string $11001101$}
\label{fig:booleanCirc}
\end{figure}

In general, this representation can be exponentially smaller than the string 
(and hence the family of sets), as we only need $\log(n)$ many input gates to
represent a string of length~$n$. However, while every string of length~$n$ can be represented
by a boolean circuit with $n$ gates only strings with a enough ``logical structure'' can be succinctly represented
using only $O(\log(n))$ logical gates in addition to the $\log(n)$ input gates.

Before we show our results, we quickly review the 
basic results on succinctly represented problems from the literature
and recall the definitions and lemmas we need.
The study of succinct problems goes back to
\citet{Wigderson1983} for graphs that are succinctly represented by a boolean circuit.
Later this approach was extended by \citet{balcazar1992} to arbitrary problems that are 
succinctly represented by boolean circuits
in the following way.

\begin{Def}
We say a boolean circuit $C_w$ with two output gates represents a binary string $w$
if for every input of a binary number $i$ the following holds:
\begin{itemize}
\item the first output is $1$ if and only if $i \leq |w|$
\item if the first output is $1$ then the second output equals the $i$-th bit of $w$.
\end{itemize}

The succinct version $Q_S$ of a problem $Q$ is: Given a boolean circuit $C_w$
representing a boolean string $w$ decide whether $w \in Q$.
\end{Def}

For example, \textsc{Succinct Sat} -- the succinct version of \textsc{Sat} --
can be defined as follows:

\begin{problem}
  \problemtitle{Succinct Sat}
  \probleminput{A boolean circuit $C_w$ representing a word $w$.}
  \problemquestion{Is the 3-CNF represented by $w$ satisfiable?\footnotemark}
\end{problem}
\footnotetext{It is not important what specific encoding is used as long as
the number of variables and clauses as well as the i-th variable in the j-th clause
can be read in \texttt{polylog-time}. Any reasonable encoding will satisfy this requirement.}

\textsc{Succinct Sat} has been shown to be \NEXP-complete \citep{papadimitriou1986,papadimitriou1994computational}.
Hence, \textsc{Succinct Taut} is \coNEXP-complete.
Succinct versions of the problems considered
in this paper be can defined similarly.
The main tool to determine the complexity of succinct problems are so-called Conversion Lemmas. 
We use the Conversion Lemma by \citet{balcazar1992}.
Stronger versions of this lemma 
exist, for example 
by \citet{veith1998}.
However, the Conversion Lemma of \citet{balcazar1992} 
suffices for our purposes and
has the advantage 
that only comparably simple reductions are used,
namely \texttt{ptime} reductions and \texttt{polylog-time} reductions.
\texttt{polylog-time} reductions are reductions that --
given random access to the input --
need only $O(\log^c(n))$-time to output an arbitrary bit of the output.
The following definition is taken from \citep{Murray2017}.

\begin{Def}
An algorithm $R:\{0,1\}^*\times \{0,1\}^* \to \{0,1,\star\}$ is a 
\texttt{polylog-time} reduction from $L$ to $L'$ if there are constants $c \geq 1$ and $k \geq 1$
such that for all $x \in \{0,1\}^*$,
\begin{itemize}
\item $R(x,i)$ has random access to $x$
\item $R(x,i)$ runs in $O((\log(|x|))^k)$ time for all
$i \in \{0,1\}^{\lceil 2c\log(|x|)\rceil}$
\item there is an $l_x \leq |x|^c+c$ such that $R(x,i) \in \{0,1\}$ and
for all $i\leq l_x$ and $R(x,i) = \star$ for all $i>l_x$.
\item $x \in L$ iff $R(x,1)\cdot R(x,2)\cdots R(x,l_x)\in L'$.
\end{itemize}
\end{Def}

Here $\cdot$ is the string concatenation and $\star$ is the out of bounds character 
that marks the end of a string.
Now, we can formulate the Conversion Lemma of \citet{balcazar1992}.

\begin{Lem}[Conversion Lemma]\label{Conversion}
Let $Q$ and $Q^*$ be decision problems.
If there is a \texttt{polylog-time} reduction from $Q$ to $Q^*$ then
there is a \texttt{ptime} reduction from $Q_S$ to $Q_S^*$.
\end{Lem}

We can use the Conversion Lemma 
to prove the following theorem.

\begin{center}
\begin{mdframed}[style=mystyle]
\begin{minipage}{\textwidth}

\medskip
\begin{Thm}\label{Succinct1}
\textsc{Succinct $DI^S$-LO-Orderability} is \NEXP-complete.
\textsc{Succinct strong $DI^{S}$-LO-Orderability} 
is \NEXP-hard.
The same holds if we add the extension rule.
\end{Thm}
\end{minipage}
\end{mdframed}
\end{center}

\begin{proof}
\textsc{Succinct $DI^{S}$-LO-Orderability} can be solved in \NEXP-time 
by explicitly computing the family $\cX$ and then solving the (exponentially larger)
explicit problem in \NP-time.

For the hardness, we only have to check that the presented reduction is 
computable in \texttt{polylog-time}. Then, by the Conversion Lemma, 
there is a ptime reduction from \textsc{Succinct Sat} to 
the problems mentioned above.
The \NEXP-hardness of both problems then follows 
as \textsc{Succinct Sat} is known to be \NEXP-complete. 
We have to show that we can compute a single bit of the output 
in \texttt{polylog-time} if we have random access to the input.
For this, we have to take the binary representation of \textsc{Sat}
into account. Unfortunately, neither \citet{papadimitriou1986} nor
\citet{papadimitriou1994computational}
specify a binary representation for the \NEXP-hardness proof.
However, the proof of \NEXP-hardness is not sensitive to the representation
as long as it is reasonable. The same is true for this proof.
Reasonable means in our context that it is possible to determine 
the number of variables $n$ and clauses $m$ in \texttt{polylog-time}.
For any sensible encoding of 3-CNF
this is either explicitly encoded or can be determined via binary search.
Furthermore, we assume that one only needs \texttt{polylog-time} 
to read the i-th variable in the j-th clause.
This is trivially true if we assume that 
every clause is encoded by the same amount of bits.
An example of an encoding that is reasonable in our sense
is the following encoding: The encoding starts with the number of variables $n$
encoded in binary. Then, the formula is encoded as follows
\begin{itemize}
\item Each variable is encoded by a unique string of length $\log(n) +1$.
\item Furthermore, let $s_0$ be a string of length $\log(n) +1$ that does not
encode a variable.
\item Each block of $3(\log(n) + 1) + 3$ bits encodes one clause $C_j$,
the first bit determines whether the first variable appears positively
or negatively in $C_j$ and the next $\log(n) +1$ equal the encoding of the first
variable. The remaining literals are encoded similarly. If the clause contains
less than $3$ literals, then the empty slots are filled with $s_0$.
\end{itemize}
It is easy to see that the proof by \citet{papadimitriou1994computational} 
of the \NEXP-hardness of \textsc{Succinct Sat}
works for such an encoding.

Now, we fix a binary representation for 
instances of \textsc{$DI^{S}$-LO-orderability} resp.\ \textsc{Strong $DI^{S}$-LO-orderability}.
First, we encode the number of elements $k$ of $X$ in binary.
Then, the family $\cX$ is encoded as a series of strings of 
length $k$, where a $1$ in position $l$ means the $l$-th element
of $X$ is in the set and a $0$ in position $l$ means the $l$-th element
is not in the set. For an instance of \textsc{$DI^{S}$-LO-orderability},
the linear order $\leq$ is given by
the natural order on these positions. 

First, observe that the size of $X$ is $4n +12m +3$
and the size of $\cX$ is $p(n,m)$ for some polynomial $p(x,y)$.
Therefore, we can determine it in \texttt{polylog-time}.
Now, assume we want to decide whether the $i$-th bit of the output is 
$0$ or $1$. It is clear that this can be done in \texttt{polylog-time}
if the $i$-th bit is part of the representation of the size of $X$.
Assume that the $i$-th bit determines 
if the $l$-th element $x$ is part of a $k$-th set $A$. 
We can assume that we fixed an order in which we generate the sets
in $\cX$ such that we can compute from $m$, $n$ and $i$
which set $A$ is supposed to be.
Say we first compute the Class 1 sets. For these, 
we only need to know the number of variables in the \textsc{Sat}
instance. Next, we compute the Class 2 sets. Observe that,
for a fixed number of variables and clauses, the
only elements that depends on the \textsc{Sat} instance 
are the elements of the form $x_j^+$ or $x_j^-$.
For these elements, it suffices to know 
if $v_j$ occurs (positively or negatively) in a specific clause in the right position.
In our exemplary encoding of \textsc{Sat} it suffices to read
$\log(n) +2$ bits to determine this. The position of these bits only depends
on which position in which clause we want to check.
Finally, the Class 3 sets are all either of the form $\{z_i^a\}$,
$\{z_i^a, \min_i^a\}$ resp. $\{z_i^a, \min_i^a\}$ or they equal
a Class 2 set with a specific number of elements removed. 
In the first case, the elements of of the set again only depend
on the number of variables and clauses and in the second case
the elements can be computed the same way as a Class 2,
as the removed elements also only depend on on the number of variables and clauses.
\end{proof}

We observe that the argument above does not use any properties of
the reduction that are unique to $DI^S$-LO-orderability. Therefore,
it is straightforward to check that the hardness 
of the other strong and $\leq$-orderability properties can be lifted in the same way.
Moreover, $\NEXP$-membership follows for all $\leq$-orderability problems
by the same argument as above.

\begin{center}
\begin{mdframed}[style=mystyle]
\begin{minipage}{\textwidth}

\medskip
\begin{Cly}\label{Succinct3}
\textsc{Succinct $DI$-LO-orderability},
\textsc{Succinct $DI$-WO-Orderability},
\textsc{Succinct $DI^{S}$-WO-Orderability} are \NEXP-complete.
\textsc{Succinct strong $DI$-LO-orderability}, \textsc{Succinct strong $DI$-WO-Orderability} 
\textsc{Succinct strong $DI^{S}$-WO-Orderability} 
are \NEXP-hard.
The same holds if we additionally add the extension rule.
\end{Cly}
\end{minipage}
\end{mdframed}
\end{center}

Moreover, we note that the Conversion Lemma can also be applied the same way
to the reduction from \textsc{Taut} to \textsc{Strong partial $DI^S$-Orderability}.

\begin{center}
\begin{mdframed}[style=mystyle]
\begin{minipage}{\textwidth}

\medskip
\begin{Thm}\label{Succinct2}
\textsc{Succinct strong $DI^S$-PO-Orderability} 
is \coNEXP-complete.
\end{Thm}
\end{minipage}
\end{mdframed}
\end{center}

This analysis still leaves some gaps. It can be shown that
\textsc{Succinct Strong $DI^{S}$-WO-Orderability} is in $\Pi_2^E$, the second
level of the exponential hierarchy, by a similar argument as the one used to show that 
\textsc{Succinct $DI^{S}$-WO-Orderability} is in \NEXP.
It seems very likely that this upper bound is tight and that 
\textsc{Succinct Strong $DI^{S}$-WO-Orderability} is indeed $\Pi_2^E$-complete. However, 
as the succinct version of $\Pi^2$-\textsc{Sat} is, to the best of our knowledge, not known
to be $\Pi_2^E$-hard, the Conversion Lemma does not suffice to show this. Closing this gap is,
therefore, left to future work.
The same holds for the other problems regarding strong orderability
where the lifted order needs to be a linear or weak order.

On the other hand, we do not provide a lower bound for 
\textsc{Succinct $DI^{S}(E)$-PO-Orderability},
because we do not have a lower bound on the complexity of 
\textsc{$DI^{S}(E)$-PO-Orderability} even in the non-succinct case.
We would conjecture that \textsc{$DI^{S}(E)$-PO-Orderability} is \P-complete,
but leave a proof of this to future work.

\subsection{Weak Orderability}\label{sec:weak}

In this section, we consider weak orderability.
We will restrict our attention to 
weak orderability with respect to dominance and strict independence,
resp.\ with respect to dominance, strict independence and the extension
rule. The question if these results also hold for regular independence
is left for future work.
First, we show that \textsc{Weak $DI^S$-WO-Orderability}
is \NP-complete. This requires a completely different construction.
We reduce this time from \textsc{Betweenness}
which was shown to be \NP-hard by \citet{opatrny1979}.

\begin{problem}
\problemtitle{Betweenness}
\probleminput{A set $V = \{v_1, v_2, \dots, v_n\}$ and a set of triples $R\subseteq V^3$.}
\problemquestion{Does there exist a linear order on $V$ such that 
$a<b<c$ or $a>b>c$ holds for all $(a,b,c)\in R$?}
\end{problem}

The idea of the following reduction is as follows:
From a \textsc{Betweenness} instance $(V,R)$
we construct an
\textsc{Weak $DI^S$-WO-Orderability} instance $(X,\cX)$. 
The set $X$ is constructed by adding some 
auxiliary variables to $V$. Then we use these
auxiliary variables to build for every triple $(a,b,c)$
in $R$ a collection of sets that are not $\leq$-$DI^S$-orderable
for any linear order on $X$ that violates the betweenness condition
for $a,b,c$, i.e., if $b < a,c$ or $a,c < b$
holds. The union of these collections for every triple in $R$
will be the family $\cX$. Then we have to show that 
$\cX$ is $\leq$-$DI^S$-orderable if 
$a < b< c$ or $a> b> c$ holds for all triple $(a,b,c)$ in $R$.

\begin{center}
\begin{mdframed}[style=mystyle]
\begin{minipage}{\textwidth}

\medskip
\begin{Thm}\label{WeakNP}
\textsc{Weak $DI^S$-WO-Orderability} is \NP-complete.
\end{Thm}
\end{minipage}
\end{mdframed}
\end{center}

\begin{proof}
It is clear that \textsc{Weak $DI^S$-WO-Orderability} is in \NP{}
because we can guess a linear order $\leq$ on $X$
and a linear order $\preceq$ on $\cX$
at the same time and then check in polynomial
time if $\preceq$ satisfies dominance and 
strict independence with respect to $\leq$.

It remains to show that it is also \NP-hard.
We do this by a reduction from \textsc{Betweenness}.
For the reduction we need two gadgets.
First, we need a gadget that guarantees for three elements $a,b,c$
that either $a < b,c$ or $a > b,c$ has to hold.
This can be done using a similar idea as for the impossibility result
of \citet{barbera1984}.
The second gadget, for elements $x,x',y,y'$, leads to a
contradiction if $x,x' < y,y'$ or $y,y' < x,x'$ holds.

\paragraph{The gadget $A(x,y,z)$:}
Let $X = \{x,y,z\}$. Then, we write $A(x,y,z) := \powerset{X} \setminus \{x\}$.
We claim that $A(x,y,z)$ is not $DI^S$-orderable with respect to $\leq$
for $y < x < z$ or $z < x < y$. We assume $y < x < z$. The other case can be treated analogously.
Assume for the sake of contradiction that there is an order $\preceq$ on $A(x,y,z)$ that satisfies
dominance and strict independence. Then, $\{y\} \prec \{y,x\}$ by dominance  
and hence $\{y,z\} \prec \{y,x,z\}$ by strict independence.
On the other hand, $\{x,z\} \prec \{z\}$ by dominance and hence
$\{y,x,z\} \prec \{y,z\}$ by strict independence. A contradiction.

\paragraph{The gadget $B(x,x',y,y')$:}
We write $B(x,x',y,y')$ for the set
\[
\{\{x,y,y'\},\{x',y,y'\},\{x,x'\},\{x,x',y\},\{x,x',y'\},
\{x,x',y,y'\}, \{y\},\{y'\},\{y,y'\}\}.\]
We claim that $B(x,x',y,y')$ is not $DI^S$-orderable 
with respect to $\leq$ if $x,x' < y,y'$ or $y,y' < x,x'$.
We assume $x,x' < y < y'$. The other cases follow by symmetry.
Then, $\{x,x'\} \prec \{x,x',y\}$ by dominance and hence $\{x,x',y'\} \prec \{x,x',y,y'\}$
by strict independence. On the other hand $\{y,y'\}\prec \{y'\}$ by dominance and
hence by strict independence first $\{x',y,y'\} \prec \{x',y'\}$ and second 
$\{x,x',y,y'\} \prec \{x,x',y'\}$, a contradiction.

\paragraph{The reduction:}
Given an instance of \textsc{Betweenness}
\[(V = \{v_1,\dots,v_n\}, T = \{(v_i,v_j,v_k), \dots ,(v_{i'},v_{j'},v_{k'})\})\]
we build an instance $(X,\cX)$ of \textsc{Weak $DI^S$-WO-Orderability}. 
First, we add for every $v_i \in V$ an element $v_i$ to $X$.
Furthermore, we add for every triple $(v_i,v_j,v_k)$
new elements $y_{ijk},y'_{ijk},z_{ijk}$ and $z'_{ijk}$.
If no ambiguity arises, we omit the index $ijk$.
Finally, for every triple $(v_i,v_j,v_k)$ we add sets 
\begin{multline*}
A(v_i,y_{ijk},y_{ijk}'), A(v_j,y_{ijk},y_{ijk}'), A(v_i,z_{ijk},z_{ijk}'), A(v_j,z_{ijk},z_{ijk}'),\\
A(v_k,z_{ijk},z_{ijk}'), A(v_i,y_{ijk},z_{ijk}), A(v_k,y_{ijk},z_{ijk}),
B(v_i,v_j,y_{ijk},y_{ijk}'), \\
B(v_i,v_k,z_{ijk},z_{ijk}') \mbox{ and } B(v_j,v_k,z_{ijk},z_{ijk}') 
\end{multline*}
to $\cX$.
We claim that $\cX$ is $DI^S$-orderable with respect to $\leq$
if and only if the projection of $\leq$ to $V$ is a positive solution 
to the given \textsc{Betweenness} instance.
The idea is, roughly, that none of these gadgets leads to a contradiction if we 
set either \[v_i < y < y' < v_j < z < z' < v_k\] or 
\[v_i > y > y' > v_j > z > z' > v_k,\]
hence there is a way to avoid a contradiction if 
either $v_i < v_j < v_k$ or $v_i > v_j > v_k$ holds.
On the other hand, we can show that there is no way to extend an 
order that sets either $v_j < v_i,v_k$ or $v_i,v_k <v_j$ without running 
into a contradiction.

First, we show that $\cX$ is not $DI^S$-orderable with respect to $\leq$
if the projection of $\leq$ is not a positive instance of \textsc{Betweenness}.
Let $(v_i,v_j,v_k)$ be a triple that is violated by $\leq$,
i.e., either $v_j < v_i,v_k$ or $v_i,v_k <v_j$.
We assume $v_j < v_i,v_k$. The other case can be treated analogously.
Assume for the sake of contradiction that $\cX$ is $DI^S$-orderable 
with respect to $\leq$. Then observe that $A(v_i,y,y')$ and $A(v_j,y,y')$
imply that $v_i$ and $v_j$ cannot lie between $y$ and $y'$ in $\leq$.
Therefore, we must have either $v_i,v_j<y,y'$, $y,y' <v_i,v_j$,
$v_i,<y,y'<v_j$ or $v_j,<y,y'<v_i$. 
However, the first two cases are ruled out by $B(v_i,v_j,y,y')$
and the third case is ruled out by $v_j < v_i$, hence we know $v_j < y,y' < v_i$.
Similarly, $A(v_j, z,z'), A(v_k,z,z')$ and $B(v_j,v_k,z,z')$ imply $v_j < z,z' < v_k$.
By $A(v_i,y,z)$ and $A(v_k,y,z)$ we know that $v_i$ and $v_k$ cannot lie 
between $y$ and $z$, hence we must have $v_j < y,z < v_i,v_k$.
Now, $A(v_j, z,z'), A(v_k,z,z')$ imply that neither $v_i$ nor $v_k$ can 
lie between $z$ and $z'$. Hence we must have  $v_j < y,z, z' < v_i,v_k$.
However, this is ruled out by $B(v_i,v_k,z,z')$. A contradiction.

Now, assume $\leq$ is a positive instance of \textsc{Betweenness}.
We extend $\leq$ to an order on $\cX$ by setting for all triples 
$(v_i,v_j,v_k)$ the order $v_i < y < y' < v_j < z < z' < v_k$ 
if $v_i < v_j < v_k$ and $v_i > y > y' > v_j > z > z' > v_k$ otherwise.
We can do this in a way such that there is no element between $y$ and $y'$
and $z$ and $z'$.
Now, we can order the sets made up by the new elements $y,y',z,z'$
with an order $\preceq$ satisfying 
dominance and strict independence: We lift $\leq$ to the singletons and place
the sets of the form $\{y,y'\}$ between $\{y\}$ and $\{y'\}$ and 
sets of the form $\{z,z'\}$ between $\{z\}$ and $\{z'\}$.
Finally we place sets of the form $\{y',z\}$ right after $\{y'\}$.

For every triple $(v_i,v_j,v_k)$ with auxiliary elements $y,y',z,z'$ 
such that $v_i < v_j < v_k$ holds, we add
the sets 
\[\{v_i,y\}, \{v_i,y,y'\}, \{v_i,y'\}, \{v_i, y, v_j\},
 \{v_i, y,y',v_j\}, \{v_i,v_j\}, \{v_i,y',v_j\}\]
in this order just below $\{y\}$.
For every of these sets $A$, we have $\min(A) = v_i$.
Furthermore, for all $A \prec B$ we have $\max(A) \leq \max(B)$.
Hence, this sequence satisfies dominance. 
Next we add
\[\{y,v_j\}, \{y,y',v_j\}, \{y',v_j\}, \{v_i, y',z\}, \{y',z\}, \{y',z, v_k\}\] 
in this order above $\{y'\}$.
It can be checked that this also satisfies dominance.
Furthermore, we add
\begin{multline*}
\{v_i,z\}, \{v_i,z,z'\}, \{v_i,z'\}, \{v_i, z, v_k\}, \{v_i, z,z',v_k\},\\
 \{v_i,v_k\}, \{v_i,z',v_k\}, \{v_j,z\}, \{v_j,z,z'\}, \{v_j,z'\},\\
 \{v_j, z, v_k\}, \{v_j, z,z',v_k\}, \{v_j,v_k\}, \{v_j,z',v_k\}
\end{multline*}
in this order just below $\{z\}$.
For the first half, we have again $\min(A) = v_i$
and $A \prec B$ implies $\max(A) \leq \max(B)$.
For the second half, the same holds with 
$\min(A) = v_j$. Therefore, this block satisfies dominance.
Furthermore, we have $\max(A) \geq z$ and in the earlier blocks, 
the only set $B$ with $\max(B) > z$ is $\{y',z, v_k\}$, 
for which dominance does not imply any preferences with sets in this block.
Hence, dominance is also satisfied with in relation to the earlier blocks.
Finally, we have $\min(A) \leq z < \max(A)$ for all sets in the block.
Therefore, we can place $\{z\}$ above the block without violating dominance. 
We conclude the construction by placing the sets
\[\{z,v_k\}, \{z,z',v_k\}, \{z',v_k\}\] in this order above $\{z'\}$. 
Again, this does not contradict dominance.
If $v_i > v_j > v_k$ holds, we produce exactly the reverse order.
In order to see that it satisfies strict independence, we have to distinguish two cases:
Let $A$ and $B$ be sets in $\cX$ such that $A \cup \{x\}$ and $B \cup \{x\}$
are also in $\cX$ and $A \prec B$. First, assume $x = y_{ijk}$ for some triple $(v_i,v_j,v_k)$.
First, assume $A \cap \{v_i,v_j,v_k\} = B \cap \{v_i,v_j,v_k\}$.
Then, we must have $A = \{v_i,v_j\}$ and $B = \{v_i,y',v_j\}$
and hence $A \cup \{x\} \prec B \cup \{x\}$
by definition.
Now, assume $A \cap \{v_i,v_j,v_k\} \neq B \cap \{v_i,v_j,v_k\}$. 
Then, the order of $A$ and $B$ as well as the order of $A \cup \{x\}$ and $B \cup \{x\}$
does not depend on the auxiliary variables.
Hence $A \prec B$ implies $A \cup \{x\} \prec B \cup \{x\}$.
The cases $x \in \{y',x,x'\}$ are similar.

So assume $x = v_i$. Then, observe that if $A \prec B$, because $B$ contains $v_j$ or $v_k$ and $A$ 
doesn't, then $A \cup \{x\} \prec B \cup \{x\}$ for the same reason. Otherwise observe that
the order of $A$ and $B$ can only depend on the auxiliary elements. However, by assumption,
these do not change when we add $x$, hence $A \cup \{x\} \prec B \cup \{x\}$.
\end{proof}

A close inspection of this proof shows that (1) we did not use the
fact that the lifted order needs to be total and (2) that the lifted order
satisfies the extension rule. Hence, the \NP-hardness carries over to 
partial weak $DI^S$-orderability and if we add the extension rule.

\begin{center}
\begin{mdframed}[style=mystyle]
\begin{minipage}{\textwidth}

\medskip
\begin{Cly}\label{Cly:Weak}
\textsc{Weak $DI^SE$-WO-Orderability},
\textsc{Weak $DI^S$-PO-Orderability} and 
\textsc{Weak $DI^SE$-PO-Orderability} are \NP-complete.
\end{Cly}
\end{minipage}
\end{mdframed}
\end{center}

\begin{proof}
The fact that all three problems are in \NP{} follows by the same argument
used to show that \textsc{Weak $DI^S$-WO-Orderability} is in \NP.
Now, we claim that the same reduction used above shows without modification that
all three problems are \NP-hard. First of all, the argument that the constructed 
instance $(X,\cX)$ is not weakly $DI^S$-orderable if $(V,R)$ is not a positive
instance of \textsc{Betweenness} does not rely on the fact that 
$\leq$ needs to be total.\footnote{Observe that previous reductions 
also did not explicitly mention the totality of the lifted order,
but used the fact that (strict) independence implies reverse independence,
which only holds for total orders.}
Furthermore, when constructing the order witnessing that $(X,\cX)$ is 
weakly $DI^S$-orderable if $(V,R)$ is a positive instance of
\textsc{Betweenness}, ``we lift $\leq$ to the singletons''.
Therefore, this order satisfies the extension rule by definition.
\end{proof}

Unfortunately, this reduction cannot easily be adapted to weak 
orderability with respect to regular independence.
Proposition~\ref{Cly:DomIndless5} tells us that at least $6$ 
elements are necessary to produce a conflict between dominance and
independence, therefore it seems highly unlikely that one can find
a gadget in the style of $A(a,b,c)$
that enforces a specific preference between three elements for dominance 
and independence.  
Nevertheless, it seems likely that the problem remains NP-complete
when strict independence is replaced with regular independence,
though a completely new reduction may be needed to show this.

\subsection{Strengthenings of Dominance}\label{sec:strengdom}

A possible way to overcome the high complexity of
recognizing orderable families could be to strengthen 
the axioms that we consider. If we consider very strong axioms 
then only very particular families of sets will be orderable
with respect to these axioms. Strict 
independence is already a very strong axiom, therefore,
we focus on strengthenings of dominance.
Furthermore, we will only study the simplest decision
problem, namely $\leq$-LO-orderability.
More concretely, we will show that no reasonable strengthening 
of dominance together with strict independence
makes the $\leq$-LO-orderability problem easier.
Technically, we say an axiom is dominance-like if it extends 
dominance and is implied by a very strong axiom that we call 
maximal dominance, which was first introduced by \citet{Maly2017}.
We show that the $\leq$-LO-orderability problem
with respect to strict independence and any dominance-like axiom is \NP-complete.
First we introduce maximal dominance.

\begin{Axm}[Maximal Dominance]\label{MaxDom}
For all $A,B \in \mathcal{X}$, 
\[\left(\max (A) \leq \max (B) \land \min (A) < \min (B) \right) \text{ implies } A \prec B;\]
\[\left(\max (A) < \max (B) \land  \min (A) \leq \min (B)\right) \text{ implies } A \prec B.\]
\end{Axm}

Now, we can define dominance-like axioms. As the name suggests, these are axioms that
lie between dominance and maximal dominance.

\begin{Def}
We say a axiom is \emph{dominance-like} if it implies dominance and is implied
by maximal dominance. 
\end{Def}

One example of a dominance-like axiom is set-dominance.

\begin{Axm}[Set-Dominance]\label{SetDom}
For all $A,X \in \mathcal{X}$ such that
$A \cup X \in \mathcal{X}$ and $A \cap X = \emptyset$:
\[
y<x \text{ for all } x \in X, y \in A\text{ implies } A \prec A \cup X;
\]
\[
x < y \text{ for all } x \in X, y \in A \text{ implies } A \cup X \prec A.
\]
\end{Axm}

Furthermore, if we consider Fishburn's and G\"ardenfors' extensions 
as axioms, they are both dominance-like.

The following reduction is based on a similar idea as the one used by \citet{Maly2017} but 
is significantly more general and hence proves a much stronger result.
We prove the \NP-hardness of $\leq$-orderability by a reduction from \textsc{Betweenness}, as in the case of weak orderability,
However, this time the \textsc{Betweenness} instance is not encoded in the linear order
on $X$ but in the linear order on $\cX$.

\begin{center}
\begin{mdframed}[style=mystyle]
\begin{minipage}{\textwidth}

\medskip
\begin{Thm}\label{StrengthDom1}
Let $A$ be a dominance-like axiom. Then it is \NP-hard to decide for a
given triple $(X, \cX,\leq)$ if there exists a linear order on
$\cX$ that satisfies axiom $A$ and strict independence.
\end{Thm}
\end{minipage}
\end{mdframed}
\end{center}

\begin{proof}
Let $(V,R)$ be an instance of \textsc{Betweenness}
with $V=\{v_1, v_2, \ldots,v_n\}$.
We construct a triple $(X,\mathcal{X},<)$ such that 
\begin{itemize}
\item if $(V,R)$ is a positive instance of \textsc{Betweenness},
then there is a linear order $\preceq$ on $\cX$ that satisfies maximal dominance and strict independence w.r.t.\ $\leq$,
\item if $(V,R)$ is a negative instance of \textsc{Betweenness},
then there exists no linear order $\preceq$ on $\cX$ that satisfies dominance and strict independence w.r.t.\ $\leq$.
\end{itemize}
Then, for any dominance-like axiom $A$ we know that there exists a linear order on $\cX$
that satisfies axiom $A$ and strict independence w.r.t.\ $\leq$ if and only if $(V,R)$ is a positive instance of
\textsc{Betweenness}.

We set $X = \{1, 2, \dots , N\}$ equipped with the usual linear order, for $N$
large enough. We will clarify later what large enough means.
Then, we construct the family $\mathcal{X}$ stepwise.
The family contains for every $v_i \in V$ a set $V_i$ of the following form:
\[V_i := \{K+i , \dots ,N - (K+i)\}.\]
Here, $K$ is a large enough constant. Again we will clarify later what large enough means.

Furthermore, for every triple from $R$ we want to enforce
$A\prec B \prec C$ or $A \succ B \succ C$ by adding two families of sets
as shown in Figure~\ref{fig:1.1} and Figure~\ref{fig:1.2} with $q,x,y,z \in X$.
 The solid arrows represent preferences that are forced through dominance 
and strict independence. The family in Figure~\ref{fig:1.1} makes sure that every total strict order 
satisfying independence that contains $A \prec B$ must also contain $B \prec C$.
Similarly, the family in Figure~\ref{fig:1.2} makes sure that $A \succ B$ leads to $B \succ C$.

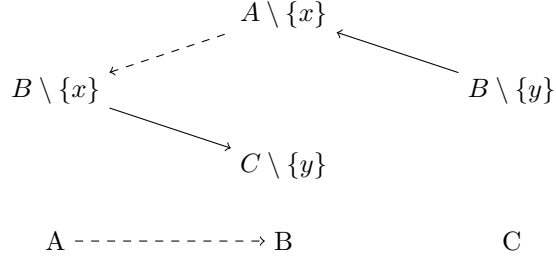
\begin{figure}[t]
\begin{center}
\begin{tikzpicture}[scale=1]
\node (a) at  (0,0) {A} ;
\node (b) at (3,0) {B};
\node (c) at (6,0) {C};
\node (bx) at (0,2) {$B \setminus \{x\}$};
\node (ax) at (3,3) {$A \setminus \{x\}$};
\node (by) at (6,2) {$B \setminus \{y\}$};
\node (cy) at (3,1) {$C \setminus \{y\}$};
\draw[<-] (ax) -- (by);
\draw[->] (bx) -- (cy);
\draw[dashed, ->] (a) -- (b);
\draw[dashed,<-] (bx) -- (ax);
\end{tikzpicture}
\caption{Family that forces that $A \prec B$ leads to $B \prec C$}\label{fig:1.1}
\end{center}
\end{figure}

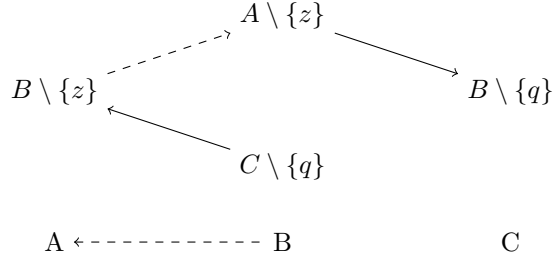
\begin{figure}[t]
\begin{center}
\begin{tikzpicture}[scale=1]
\node (a) at  (8,0) {A} ;
\node (b) at (11,0) {B};
\node (c) at (14,0) {C};
\node (bz) at (8,2) {$B\setminus \{z\}$};
\node (az) at (11,3) {$A \setminus \{z\}$};
\node (bq) at (14,2) {$B\setminus \{q\}$};
\node (cq) at (11,1) {$C\setminus \{q\}$};
\draw[->] (az) -- (bq) ;
\draw[->](cq)--(bz);
\draw[dashed, <-] (a) -- (b);
\draw[dashed, ->] (bz) -- (az);
\end{tikzpicture}
\caption{Family that forces that $A \succ B$ leads to $B \succ C$}\label{fig:1.2}
\end{center}
\end{figure}

We implement this idea for all triples inductively.
For every $1 \leq i \leq |R|$, pick a triple $(v_l,v_j,v_m)\in R$ and
set $k = K + n + 4i$.
Let $(A,B,C) = (V_l, V_j, V_m)$ be the triple of sets encoding
the triple of elements $(v_l, v_j, v_m)$. We add the following sets:
\begin{multline*} A \setminus \{k\}, B \setminus \{k\}, B \setminus \{k+1\}, C \setminus \{k+1\},\\
A \setminus \{k+2\},  B \setminus \{k+2\}, B \setminus \{k+3\}, C \setminus \{k+3\}.
\end{multline*}
We call the sets encoding the elements of $V$ together with the sets added in this step the Class~1 sets.
These sets correspond to the sets $A \setminus \{x\}, B \setminus \{x\}, \dots , C \setminus \{q\}$
in Figure \ref{fig:1.1} and Figure \ref{fig:1.2}.
Observe that the inductive construction guarantees that every constructed set is unique.
We now have to force the preferences 
\begin{multline*}
A \setminus \{k\} \prec B \setminus \{k+1\}, C \setminus \{k+1\} \prec B \setminus \{k\},\\
B \setminus \{k+3\} \prec A \setminus \{k+2\}, B \setminus \{k+2\} \prec C \setminus \{k+3\}. \quad (\star)
\end{multline*}

We define for every pair $A, B \in \mathcal{X}$ 
a family of sets $\mathcal{S}(A,B)$ forcing $A \prec B$.
Assume $\min (B)\leq \min (A)$ and $\max (A) \leq \max (B)$. 
Then, $\mathcal{S}(A,B)$ contains the following sets
\[\{x_{AB}\}, \{x_{AB}, y_{AB}\},\{y_{AB}\}, \{x_{AB}, z_{AB}\}, \{y_{AB}, z_{AB}^*\},
A \cup \{z_{AB}\}, B \cup \{z_{AB}^*\}\]
where $z_{AB}^* < \min (B) < \min (A) <  x_{AB} < y_{AB} < \max (A) < \max (B) < z_{AB}$ holds
(See Figure~\ref{FigSets}).

Additionally, we add sets that enforce $A \cup \{z_{AB}\}\prec \{x_{AB}, z_{AB}\}$  
by dominance:
Let $A = \{a_1, \dots ,a_l\}$ be enumerations of $A$
such that $i< j$ implies $a_i < a_j$.
We add 
\[
\{z_{AB}\}, \{a_l, z_{AB}\},\{a_{l-1}, a_l, z_{AB}\},  \dots, \{a_2, \dots ,z_{AB}\} \mbox{ and }\{a_1, z_{AB}\}
\]
 to $\mathcal{X}$.
This forces $\{a_2, \dots ,z_{AB}\} \prec \{z_{AB}\}$
by dominance and hence by one application 
of strict independence  $A \cup \{z_{AB}\} \prec \{a_1,z_{AB}\}$.
Finally, we add the sets $\{a_1\}, \{a_1, x_{AB}\}$ and $\{a_1, x_{AB}, z_{AB}\}$,
which leads to $\{a_1, z_{AB}\} \prec \{a_1, x_{AB}, z_{AB}\}$.
Then we have
\[
A \cup \{z_{AB}\} \prec \{a_1, z_{AB}\} \prec \{a_1, x_{AB}, z_{AB}\} \prec \{x_{AB}, z_{AB}\} 
\]
hence $A \cup \{z_{AB}\} \prec \{x_{AB}, z_{AB}\}$. 
Therefore, we have $A \preceq\{x_{AB}\}$ by reverse independence.
Analogously, we enforce $\{y_{AB}\} \preceq B$.
Therefore, transitivity implies $A \prec B$ by
$A \preceq \{x_{AB}\} \prec \{y_{AB}\} \preceq B$.
The case 
$\min (A)\leq \min (B)$ and $\max (B) \leq \max (A)$
can be treated analogously.

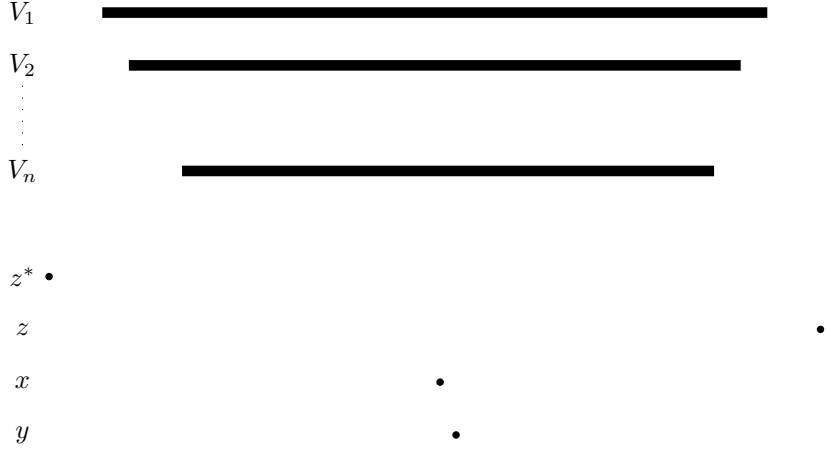
\begin{figure}[t]
\centering
\begin{tikzpicture}[scale=0.35]
\node (x1) at (-2,20) {$V_1$};
\node (x2) at (-2,18) {$V_2$};
\node (xn) at (-2,14) {$V_n$};
\node (x) at (-2,6) {$x$};
\node (y) at (-2,4) {$y$};
\node (z) at (-2,8) {$z$};
\node (z*) at (-2,10) {$z^*$};
\fill (-1,10) circle (4pt); 
\fill (28,8) circle (4pt); 
\fill (13.7,6) circle (4pt); 
\fill (14.3,4) circle (4pt); 
\draw[line width=4pt] (1,20) -- (26,20);
\draw[line width=4pt] (2,18) -- (25,18);
\draw[line width=4pt] (4,14) -- (24,14);
\draw[loosely dotted] (x2) -- (xn);
\end{tikzpicture}
\caption{Sketch of the sets $V_1,\dots ,V_n, x, y, z$ and $z^*$.}
\label{FigSets}
\end{figure}

Now, we add the following families of sets to enforce the desired preferences
\begin{multline*}
\mathcal{S}(A \setminus \{k\}, B \setminus \{k +1\}), \mathcal{S}(C \setminus \{k+1\}, B \setminus \{k\}),\\
\mathcal{S}(B \setminus \{k+3\}, A \setminus \{k+2\}),\mathcal{S}(B \setminus \{k+2\}, C \setminus \{k+3\}).
\end{multline*}
We call the sets added in this step the Class 2 sets.

We repeat this with a new triple $(v_i', v_j', v_m') \in R$ until we treated all triples in $R$.
By this construction every linear order on $\mathcal{X}$
that satisfies dominance and strict independence must set the preferences listed in $(\star)$.
This concludes the construction of $\cX$. We now can determine the necessary sizes for $N$ and $K$.
$N$ and $K$ need to be large enough
such that all $z_{AB}^*$ used in the construction are smaller than $K$, all $z_{AB}$
are larger than $N -K$ and all $x_{AB}$ and $y_{AB}$ are larger than $K +n + 4|R|$.
It is clear that this can be achieved with $N$ and $K$ that are polynomial in $|(V,R)|$.

Now, assume there is a linear order on $\mathcal{X}$ satisfying
dominance and strict independence.
We claim that the relation defined by $v_i \leq v_j$
iff $V_i \preceq V_j$ is a positive witness for $(V,R)$.
By definition this is a linear order.
So assume there is a triple $(a,b, c)$
such that $a > b < c$ or $a< b > c$ holds.
We treat the first case in detail: 
$a > b < c$ implies $A \succ B \prec C$.
This implies by the strictness of $\prec$ and strict dominance
$A \setminus \{k\} \succ B \setminus \{k\}$ and
$B \setminus \{k+1\} \prec  C \setminus \{k+1\}$. However, then 
\[A \setminus \{k\} \succ B \setminus \{k\} \succ C \setminus \{k+1\} \succ
B \setminus \{k+1\} \succ A \setminus \{k\} \]
contradicts the assumption that $\prec$ is transitive and irreflexive.
Similarly, the second case leads to a contradiction.
This shows that if $(V,R)$ is a negative instance of \textsc{Betweenness},
then there is no order on $\cX$ that satisfies dominance and strict independence
with respect to $\leq$.

Now, assume that there is a linear order on $V$ satisfying the restrictions from $R$.
We use this to construct a linear order on $\mathcal{X}$ that satisfies maximal dominance 
and strict independence with respect to $\leq$.
First, we add all preferences implied by maximal dominance.
Observe that no two Class 1 sets are comparable by maximal dominance.
Moreover, we set $V_i \preceq V_j$ iff
$v_i \leq v_j$ holds.
Then, we project this order to all sets of the form $V_i \setminus \{x\}$,
We claim that this order satisfies all applications of strict 
independence between Class 1 sets. If $A = V_i$ for $i \leq n$, then there is no
set $A \cup \{x\}$ in $\mathcal{X}$. If $A = V_i \setminus \{x\}$ for some $i \leq n$ 
and $x \in X$, then $x$ is the only element of $X$ such that $A \cup \{x\} \in
\mathcal{X}$ holds. But then there can only be one other set $B$ with $B \cup \{x\} 
\in \mathcal{X}$ and $B = V_j \setminus \{x\}$ hence a preference between $A$ and 
$B$ was introduced by reverse strict independence.

Next, we consider the Class 2 sets.
We distinguish three types of Class 2 sets. We say a set $X$ is
\begin{itemize}
\item type 1 if $z^*_{AB} \in X$,
\item type 2 if $z_{AB},z^*_{AB} \not\in X$,
\item type 3 if $z_{AB} \in X$.
\end{itemize}
Then, we set $X \prec Y$ if $\type(X) < \type (Y)$.
Furthermore, for all type 1 sets $X,Y$ we set $X \prec Y$ if 
\begin{itemize}
\item $z^*_{AB} \in X, z^*_{CD} \in Y$ and $z^*_{AB} < z^*_{CD}$,
\item $z^*_{AB} \in X,Y$ and $x_{AB} \not\in X, x_{AB} \in Y$,
\item $z^*_{AB}, x_{AB} \in X,Y$ and $\max(X \triangle Y) \in Y$,
\item $z^*_{AB} \in X,Y x_{AB} \not\in X,Y$ and $\max(X \triangle Y) \in Y$.
\end{itemize}

It is straightforward to check that this order satisfies strict independence and is compatible 
with maximal dominance on all type 1 sets.
Similarly, we set for all type 3 sets 
\begin{itemize}
\item $z_{AB} \in X, z_{CD} \in Y$ and $z_{AB} < z_{CD}$,
\item $z_{AB} \in X,Y$ and $x_{AB} \not\in X, x_{AB} \in Y$,
\item $z_{AB}, x_{AB} \in X,Y$ and $\min(X \triangle Y) \in X$,
\item $z_{AB} \in X,Y x_{AB} \not\in X,Y$ and $\min(X \triangle Y) \in X$.
\end{itemize}
Again, it is straightforward to check that this order satisfies strict independence and is compatible 
with maximal dominance on all type 3 sets.

Now, we covered all possible applications of strict independence on Class 1 sets 
and on type 1 and 3 sets. The only possible application of strict independence 
that includes Class 1 and Class 2 sets is adding $z_{AB}$ or $z_{AB}^*$
to a Class 1 set and a Class 2 set. Then, by construction both resulting sets
are type 1 resp.\ 3 sets. Therefore, we can apply reverse strict independence. 
By construction, this does not lead to a cycle if any only if we started with an
positive instance of \textsc{Betweenness}.

It remains to consider applications of strict independence that include type 2 sets.
All type 2 sets are of one of the following forms:
\[\{x_{AB}\},\{y_{AB}\}, \{x_{AB},y_{AB}\}, \{a_1\}, \{a_1, x_{AB}\}, \{b_l\}, \{b_l,y_{AB}\}\]
Now, any application of strict independence where the same element is added to two singletons 
is clearly satisfied by any order that satisfies maximal dominance.
This leaves the case that $z_{AB}$ or $z_{AB}^*$ is added to two different 
type 2 sets. Now, by construction, we have $\{a_1\} \prec \{a_1,x_{AB}\} \prec \{x_{AB}\}$
and $\{a_1, z_{AB}\} \prec \{a_1,x_{AB}, z_{AB}\} \prec \{x_{AB}, z_{AB}\}$.
Therefore, the case where $z_{AB}$ is added is satisfied. The case that $z_{AB}^*$
is added is similar.

It can be checked that this covers all possible applications of strict independence.
Finally, we can extend this order to a weak order because extensions do not produce new
instances of strict independence. 
\end{proof}

\section{Discussion}\label{sec:diss}

Lifting a preference order on elements of some universe to a preference
order on subsets of this universe 
respecting certain axioms is a fundamental problem, but
impossibility results by Kannai and Peleg and by Barber{\`a} and Pattanaik 
pose severe limits on when
such liftings exist if
\emph{all} non-empty subsets of the universe have to be ordered.
We observed that
these impossibility results may be avoided
if \emph{not all} non-empty subsets of the universe have to be ordered.
This raises the questions how hard it is to recognize 
families for which dominance and (strict) independence are jointly 
satisfiable.

Our results show that, in general, we cannot easily recognize 
families of sets on which we can avoid 
Kannai and Peleg's or Barbera and Pattanaik's impossibility results.
While these results are largely negative, we observe that 
determining if the family of sets is strongly orderable is 
important in many applications but not always time-sensitive.
Therefore, we believe that the order lifting approach studied here
may be useful for applications where partial orders are
acceptable, as we can construct a partial order that satisfies dominance and strict 
independence in polynomial time, whenever such an order exists.
Moreover, we have seen that we can
always find a partial order that satisfies dominance and weak independence.
This provides a practical solution if partial orders are acceptable
and the considered family of sets turns out not to be strongly
$DI^S(E)$-orderable or if the high complexity means that checking for 
strong orderability with respect to dominance and strict independence
is not feasible.

Additionally, we observe that an important implication
of the presented hardness results is that the characterization 
results from \citet{JAIR} can not easily be generalized to arbitrary
families of sets. Indeed, any property that characterizes strong
orderability on arbitrary families of sets must be $\Pi^p_2$-complete
to check. This precludes nearly all simple properties of 
families of sets from characterizing strong orderability.
On the other hand, we observe that hypergraph
colorability is NP-hard to check. Hence, a generalization of
the characterization of weak orderability in terms of 
colorability by \citet{JAIR} is not necessarily ruled out by the NP-completeness 
of checking weak $DI^S$-orderability.

There are some remaining gaps in our results that are left to future work.
First of all, the complexity of weak $DI$- and $DIE$-orderability
is left open. Furthermore, it remains open
if strengthening strict independence 
influences the complexity of the studied problems.
Moreover, our research opens several new directions for future studies.

First of all, 
Kannai and Peleg's or Barbera and Pattanaik's impossibility results
are the most prominent but certainly not the only impossibility results.
Other interesting impossibility results where for example proven by \citet{GeistE11}
and \citet{jones1982}. These other results also assume that the whole power set needs to be ordered.
Therefore, one could study the questions raised in this thesis the same way also for these
other impossibility results.

Moreover, the representation of families by boolean circuits is 
extremely powerful and therefore leads to an exponential blow up in complexity.
On the other hand, the  representation of sets by the connectivity condition on graphs
used by \citet{JAIR} can decrease the complexity of the studied problems,
but is very restrictive.
Therefore, future research is needed to identify
succinct representations that are as expressive as possible without
an exponential blow up in complexity.
Implicit representations of families of sets 
that appear in important applications would be especially interesting to study.
Knowledge representation often uses logic formalisms 
towards this end. For instance, formulas 
can be viewed as concise representations of the families of their models. 
Together with an order of the atoms in the formulas, it is natural to 
ask how to rank these models and for which classes of formulas 
such a lifting respects certain criteria.

We note that the presented reductions all require relatively large sets of objects $X$.
However, in many practical applications the number of objects is often rather small.
For example, in many voting scenarios the number of candidates seldomly exceeds 10.
Therefore, it would be worth exploring the parametrized complexity of the discussed problems 
in particular with regards to the parameter $|X|$.

Furthermore, while we studied the effect of allowing the lifted order to be incomplete,
we always assumed that the order on $X$ is a linear order.
However, there are important applications, for example in argumentation \citep{BeirlaenHPS18,AAAI21},
where it is necessary to lift partial orders.
Hence, it would be interesting to study the effect of allowing the order on $X$ to be incomplete
and, in particular, to explore whether the studied problems become 
easier if only few comparisons between objects are specified.
This could either be achieved by means parametrized complexity theory
or by aiming for classification results in the style of \citet{JAIR}.

Finally, several interesting questions arise when applying the order lifting approach
in specific settings. For example, if lifted orders are used in voting or (ordinal) allocation
to generate the input to a voting rule or allocation method,
then any axiom essentially represents a domain restriction in that it forces a specific structure 
on the lifted orders. For example, if the lifted order is a linear order that satisfies strict independence 
and we have four alternatives $A,B,C,D$ such that $C = A \cup \{x\}$ and $D = B \cup \{x\}$,
then the lifted orders of all voters must order $A$ and $B$ the same as $C$ and $D$.
This added structure might influence which voting rule or allocation method satisfies the most
desirable axioms.
Therefore, it would be highly interesting to study 
the interplay between lifting procedures and axioms on the one hand and voting rules and other 
social choice mechanisms on the other hand.

\section*{Acknowledgments}

The author wants to thank the anonymous reviewers for their valuable feedback and helpful suggestions.
Furthermore, the author would like to thank the Austrian Science Fund (FWF) which supported this work under Grant number P31890 and J4581.
This paper is an extension of an AAAI 2020 paper \citep{AAAIpaper}.
Moreover, the results presented in this
paper also appeared in the PhD thesis of the author \citep{thesis}.

\bibliographystyle{plainnat}
\bibliography{literature}

\end{document}